\documentclass[a4paper,twocolumn,11pt,accepted=2024-01-08]{quantumarticle}
\pdfoutput=1
\usepackage[utf8]{inputenc}

\usepackage[T1]{fontenc}

\usepackage{tikz}
\usepackage{lipsum}

\usepackage[utf8]{inputenc}
\usepackage[T1]{fontenc}
\usepackage[numbers,sort&compress]{natbib}
\usepackage[english]{babel}

\usepackage{amsmath,amsthm,amssymb,dsfont,amscd,mathrsfs}
\usepackage[centercolon=true]{mathtools}
\usepackage{graphicx,microtype,slashed,bigints}
\usepackage{hyperref}
\usepackage[capitalize]{cleveref}
\hypersetup{
  linkcolor=red,
  colorlinks=true
}

\theoremstyle{plain}
\newtheorem{lemma}{Lemma} 
\newtheorem{proposition}{Proposition}
\newtheorem{corollary}{Corollary} 
 
\newtheorem{definition}{Definition}

\def\Tr{\operatorname{Tr}}
\def\dim{\operatorname{dim}}

\def\>{\rangle}
\def\<{\langle}

\def\T+{\mathsf{T}_+}

\def\Evd#1{\set{T}_1(#1)}
\def\Eva#1{\set{T}(#1)}

\usepackage{stackengine}

\newcommand\stackrqarrow[2]{%
    \mathrel{\stackunder[2pt]{\stackon[0.3pt]{$\rightsquigarrow$}{$\scriptscriptstyle#1$}}{%
            $\scriptscriptstyle#2$}}}

\newcommand\Item[1][]{%
  \ifx\relax#1\relax  \item \else \item[#1] \fi
  \abovedisplayskip=0pt\abovedisplayshortskip=0pt~\vspace*{-\baselineskip}}

\newcommand{\bra}[1]{\langle#1|}
\newcommand{\ket}[1]{|#1\rangle}
\newcommand{\braket}[2]{\langle#1|#2\rangle}
\newcommand{\ketbra}[2]{{\ket{#1}\bra{#2}}}

\newcommand{\hilb}[1]{\mathcal{#1}}
\newcommand{\set}[1]{{\sf #1}}

\newcommand{\Lin}[1]{\mathcal{L}(#1)}

\begin{document}
\title{No-signalling constrains quantum
  computation with indefinite causal structure}

\author{Luca Apadula}
\email[]{Luca.Apadula@oeaw.ac.at} 
\affiliation{University of Vienna, Boltzmanngasse 5, 1090 Vienna, Austria }\affiliation{Institute for Quantum Optics and Quantum Information (IQOQI), Austrian Academy of Sciences, Boltzmanngasse 3, 1090 Vienna, Austria}

\author{Alessandro Bisio}
\email[]{alessandro.bisio@unipv.it} \affiliation{Dipartimento di
  Fisica, Universit\`a di Pavia, via Bassi 6, 27100 Pavia, Italy}
\affiliation{Istituto Nazionale di Fisica Nucleare, Sezione di Pavia,
  Italy}

\author{Paolo Perinotti}\email[]{paolo.perinotti@unipv.it}
\affiliation{Dipartimento di Fisica, Universit\`a di Pavia, via Bassi 6, 27100 Pavia, Italy}
\affiliation{Istituto Nazionale
  di Fisica Nucleare, Sezione di Pavia, Italy}

\begin{abstract}
  Quantum processes with indefinite causal
  structure emerge when we wonder which are the
  most general evolutions, allowed by quantum
  theory, of a set of local systems which are not
  assumed to be in any particular causal order.
  These processes can be described within the
  framework of \emph{higher-order} quantum theory
  which, starting from considering maps from
  quantum transformations to quantum
  transformations, recursively constructs a
  hierarchy of quantum maps of increasingly higher
  order.  In this work, we develop a formalism for
  quantum computation with indefinite causal
  structures; namely we characterize the
  computational structure of higher order quantum
  maps.  Taking an axiomatic approach, the rules
  of this computation are identified as the most
  general compositions of higher order maps which
  are compatible with the mathematical structure
  of quantum theory.  We provide a mathematical
  characterization of the admissible composition
  for arbitrary higher order quantum maps.  We
  prove that these rules, which have a
  computational and information-theoretic nature,
  are determined by the more physical notion of
  the signalling relations between the quantum
  systems of the higher order quantum maps.
\end{abstract}

\maketitle

\section{Introduction.} 
One of the main reasons to use
the notions of channel, positive operator valued
measures (POVM) and quantum instrument is that they
provide a concise description of physical devices
avoiding a detailed account of their
implementation in terms of unitary interactions
and von Neumann measurements.
% The most general (deterministic) transformation 
% of the type $A \to B$, i.e. that maps states of
% system $A$ to states of system $B$, is a quantum channel 
This is very
useful when dealing with optimization problems,
like e.g. state estimation
\cite{helstrom1976quantum}, where one can look
for the optimal measurement among all those which are
allowed by quantum mechanics.

However, if we consider transformations, rather than states,
as information carriers, these notions exhibit the same limitations.
For example, in quantum channel
discrimination \cite{kitaev1997quantum,childs2000quantum,acin2001statistical,d2001using,duan2007entanglement,sacchi2005optimal,chiribella2008memory,pirandola2019fundamental}
we need to optimize both the (possibly entangled) input
state and the mesurement.
It is then convenient
to turn the task into an optimization 
of a single  object that could
describe the joint action of the 
input state and of  the final POVM avoiding any
redundancy of the description.
This 
object (called tester or process POVM in the
literature \cite{PhysRevA.77.062112,PhysRevA.80.022339} )
would describe the most general map, allowed by
quantum theory,  that takes a channel as an input and
outputs a probability.

Channel discrimination is just an example of a more
general pattern, since any optimization of an information processing task
can be phrased as follows: look for the optimal
process among all the ones that $i)$ accept  as input
an object with the given structure (or \emph{type})
$x$ and that $ii)$ 
output an object with the target type $y$.
This is the intuitive definition of a
\emph{higher order map} of type $x \to y$.
Since any map
can be considered as the input (or output)
of another map, we can recursively construct  
a full hierarchy of maps of
increasingly higher order.
For example, a channel
from system $A$ to $B$
is a map of type $A \to B$ and a channel
discrimination strategy is a map of type
$(A
\to B) \to I$, i.e. a map which transforms a channel into a
probability
($I$ denotes the trivial
one-dimensional system).

The theory of higher order maps (or higher order
quantum theory)
\cite{PhysRevA.80.022339,doi:10.1098/rspa.2018.0706}
is the appropriate framework for the optimization of
information processing scenarios in which
inputs and outpus can be more general objects than
quantum states
\cite{chiribella2008optimal,PhysRevLett.102.010404,bisio2010optimal,gutoski2012measure,chiribella2016optimal,jenvcova2016conditions,sedlak2019optimal,mo2019quantum,dong2021success,soeda2021optimal}.

\begin{figure}[t]
    \includegraphics
   [width=0.47\textwidth]{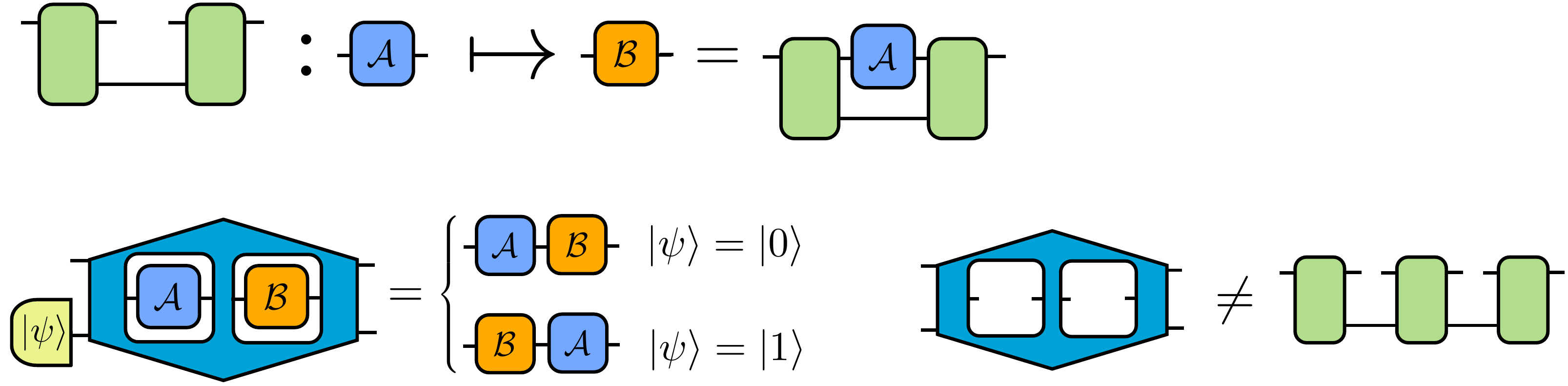}
   \caption{\emph{(Top)} A quantum circuit with an
     open slot is a higher order map that transforms
     channels into channels.
\emph{(Bottom)}     The quantum switch map cannot
be realized as a circuit with open slots.}
   \label{circuital}
 \end{figure}

However, this is only the first half of the story.
The second half concerns the notion of causal
order in quantum theory.  Within the hierarchy of
higher order maps, we have a subset of maps (known
as \emph{combs} \cite{PhysRevA.80.022339,bisioacta} ) that can
be realized as a quantum circuit with open slots:
the action of the map is given by ``filling'' the
empty slots with another map which has a
compatible structure (see Fig.~\ref{circuital}).
However, there also exist maps for which such a
realization is not possible
\cite{Oreshkov:2012aa}. In this case, one say
that the open slots which should accommodate the
input of the maps exhibit an
\emph{indefinite causal order}.
Typically, this happens whenever we consider a map
whose input is an object representing many parties
whose local operations are not assumed to be in any
particular causal order.
The paradigmatic example is the quantum switch map
\cite{PhysRevA.88.022318,PhysRevLett.121.090503}  (see Fig.~\ref{circuital})
which takes as an input two quantum channels, say
$\mathcal{A}$ and $\mathcal{B}$ and outputs the
coherent superposition of the sequential
applications of the two channels in two different
order, i.e.  $\mathcal{A}\circ\mathcal{B}$ and
$\mathcal{B}\circ\mathcal{A}$.
We see that, in order for this map to be well
defined,
the crucial assumption is that we cannot say whether the
transformation $\mathcal{A}$ occurs before or
after the transformation $\mathcal{B}$.
Recent results
have shown that higher order maps with indefinite
causal order can outperform circuital ones in a
variety of tasks
\cite{Oreshkov:2012aa, COLNAGHI20122940,PhysRevLett.113.250402,PhysRevLett.120.120502,zhao2020quantum,bavaresco2021strict,PhysRevResearch.3.043012}.
Motivated by the theoretical advances,
pioneering experiments have been proposed
\cite{Procopio:2015ab,PhysRevA.93.052321,guo2020experimental,PRXQuantum.2.010320}.

What is the relation between the functional
description of a map, i.e. its type, and its causal
structure?
This question links a fundamental physical property
of a process (its causal structure)  
with a computational feature (its
type). The study of higher order maps is an
outstanding chance to understand how causal
structures are influenced by the quantum nature of
the physical systems.
In this work, we answer this question proving how
the signalling properties of a higher order map are
constrained
by its type.

This result is based on the characterization of
the \emph{compositional structure} of higher order
maps, i.e. the ways in which we can connect the
quantum systems of different higher order maps
\footnote{It is worth noticing that our analysis
  could be reminiscent of the one of
  Ref. \cite{lorenz2021causal} where the causal
  structure of unitary operators is
  investigated. However, the notion of
  compositional structure has a completely
  different meaning. The authors of Ref.
  \cite{lorenz2021causal} refer to the
  decomposition of a unitary operator into smaller
  gates. In our work, we refer at the
  compositional structure of the hierarchy of
  higher order maps as the set of rules that
  characterize how the maps in the hierarchy can
  be connected to each other.  The categorical
  framework of
  Ref. \cite{8005095} is closer
  to the one presented in this contribution. 
  It considers composition rules, even though
  more restricted with respect to the ones that we
  study in the present paper,and it proves that higher order
  quantum processes form a star-autonomous
  category.}.  Let's explain
what we mean by that.  The type of a map describes
its functioning (much like an owner's manual)
but more structure lurks beneath this surface.
Besides being used to map states into states,
quantum channels can be connected together through
some of their inputs/outpus in order to form a
network of channels as in Fig.~\ref{circuital}.
Similarly, it should be possible to compose higher
order maps and form new ones.

Consistency with the
probabilistic structure of quantum theory implies
that only some of the possible compositions are
admissible, to wit, they generate another well
defined higher order map. In this paper, we
provide a formal definition for this
intuition and we prove a characterization theorem
for the compositional structure of higher order
maps.

This result is not only the cornerstone to prove
the relation between signalling structure and
the type of a map but it is a necessary step
to upgrade the theory of higher order maps to a
fully fledged
computational model which extends the circuital
one by comprising indefinite causal structures.

\section{From quantum operation to higher order
  maps.}  Within Kraus' axiomatic approach
\cite{Kraus}, a quantum operation is defined as
the most general map $\mathcal{M}$ from the set
$\Eva A := \{ \rho \in \Lin{\hilb{H}_A}, \rho \geq
0 , \Tr \rho \leq 1 \} $ of (sub-normalized)
quantum states of system $A$ to the set $\Eva B$
of quantum states of system $B$ \footnote{ We consider finite dimensional quantum
  systems},
that satisfies some physically motivated
admissibility conditions. These assumptions
guarantee that the probabilistic structure of
quantum theory is preserved and they are: $i)$
convex-linearity
$\mathcal{M}(p \, \rho + (1-p) \sigma) =
p\mathcal{M}( \rho )+ (1-p) \mathcal{M}(\sigma) $
for any $ \rho , \sigma \in \Eva A$ and
$ 0 \leq p \leq 1 $, $ii)$ complete positivity
$(\mathcal{M} \otimes \mathcal{I}_A) \psi \in
\Eva{BC}$ for any $ \psi \in \Eva{AC} $
\footnote{We denote as $AC$ the tensor product
  system $\hilb{H}_A \otimes \hilb{H}_B$ while
  $\mathcal{I}_A$ denotes the identity map on
  $\Lin{\hilb{H}_A}$} and $iii$) a
sub-normalization condition
$\Tr [\mathcal{M}(\rho)] \leq 1$
$\forall \rho \in \Eva A$ that prevents the
occurence of probabilities greater than one. In
particular, if a quantum operation $\mathcal{M}$
satisfies the identity
$\Tr [\mathcal{M}(\rho)] = 1$
$\forall \rho \in \Evd{A}$, where $\Evd{A}$ is the
set of normalized (or deterministic) states
$\Evd{A} := \{ \rho \in \Eva{A}, \Tr \rho = 1 \}$,
we say that $\mathcal{M}$ is a
\emph{deterministic} quantum operation (or a
quantum channel).

For our purposes, it is convenient to use the
Choi-Jamiolkovski isomorphism
\cite{CHOI1975285,JAMIOLKOWSKI1972275} in order to
 represent linear maps between operator spaces.
A map
$\mathcal{M} : \Lin{\hilb{H}_A} \to
\Lin{\hilb{H}_A} $ is represented by its
\emph{Choi operator}
$M \in \Lin{\hilb{H}_A\otimes \hilb{H}_B}$ defined
as $M:= ( \mathcal{I} \otimes \mathcal{M}) \Phi $,
where
$\Phi := \sum_{i,j} \ket {ii} \bra{jj} \in
\Lin{\hilb{H}_A\otimes \hilb{H}_A} $ and we have
that
$\mathcal{M}(\rho) = \Tr_A[(\rho^{\theta} \otimes
I_B) M]$ ($\Tr_A$ denotes the partial trace on
$\hilb{H}_A$, $\rho^{\theta}$ is the transpose of
$\rho$ and $I_B$ the identity operator on
$\hilb{H}_B$).  In terms of the Choi operator, a
linear $\mathcal{M}$ satisfies the admissibility
conditions (i.e. it is a valid quantum operation)
if and only if $0 \leq M \leq D$ where $D$ is the
Choi operator of a deterministic quantum operation
which satisfies $\Tr_B [D] = I_A$.  In what
follows, we will often implicitly assume the
Choi-Jamiolkovski isomorphism when considering
maps between operator spaces: the sentence ``the
map $M\in \Lin{\hilb{H}_A \otimes \hilb{H}_B}$''
should be understood as ``the map
$\mathcal{M}: \Lin{\hilb{H}_A} \to
\Lin{\hilb{H}_B}$ whose Choi operator is
$M\in \Lin{\hilb{H}_A \otimes \hilb{H}_B}$''
(which systems are the inputs and which ones are
the outputs, will be clear from the context).

The idea that leads to the notion of a higher order
map is that 
quantum operations themselves can be
 inputs of a ``second order'' map that
 transforms quantum operations into quantum
 operations.  We then boost this idea to its full generality:
 every kind of map can be
 considered as input or output of some higher
 order map. In this way, we recursively generate a
 hierarchy of maps of increasingly higher order.
 
 The most distinctive piece of information that
 goes with a higher order map is what is its input
 and what is its output. This information
 identifies the place of a map within the
 hierarchy and it is provided by the notion of
 \emph{type}, which is defined as follows. We
 start by asserting that a quantum state
 $\rho\in \Eva{A} $ of system $A$ has
 \emph{elementary type} $A$.  The general case is
 then defined recursively: a map that transforms
 maps of type $x$ into maps of type $y$ has type
 $(x \to y)$. % For example, we say that a quantum
 %operation from system $A$ to system $B$ has type
 %$A \to B$.
  From this definition, it follows that
 a type $x$ is a string like
 $x=(((A_1\to A_2)\to (A_3\to A_4 ))\to A_5 )$
 where $A_i$ are elementary types (i.e. quantum
 systems).  We use the special label $I$ to denote
 the \emph{trivial type} of the system with
 dimension $1$: therefore, measurements on system
 $A$ have the type $A \to I$.

 The hierarchical structure of the set of types
 is expressed by the partial ordering $\prec$
 which is defined as follows.
First, we say that $x \prec_p y$ if there exists $z$ such that 
 either $y = x \to z$ or $y = z \to x$. Then, the
 relation $x \prec y$ is defined as the transitive closure of
 the binary relation $\prec_p$.
 For example we have
 $A \to B \prec_q (A\to B) \to I \prec_q (C\to D) \to ((A\to
 B) \to I) $ and therefore $A \to B \prec (C\to D) \to ((A\to
 B) \to I) $

 %Many results about
 %the structure of higher order maps are proved by
 %induction on the well founded relation
 %$\prec$.

\begin{figure}[t]
    \includegraphics
   [width=0.35\textwidth]{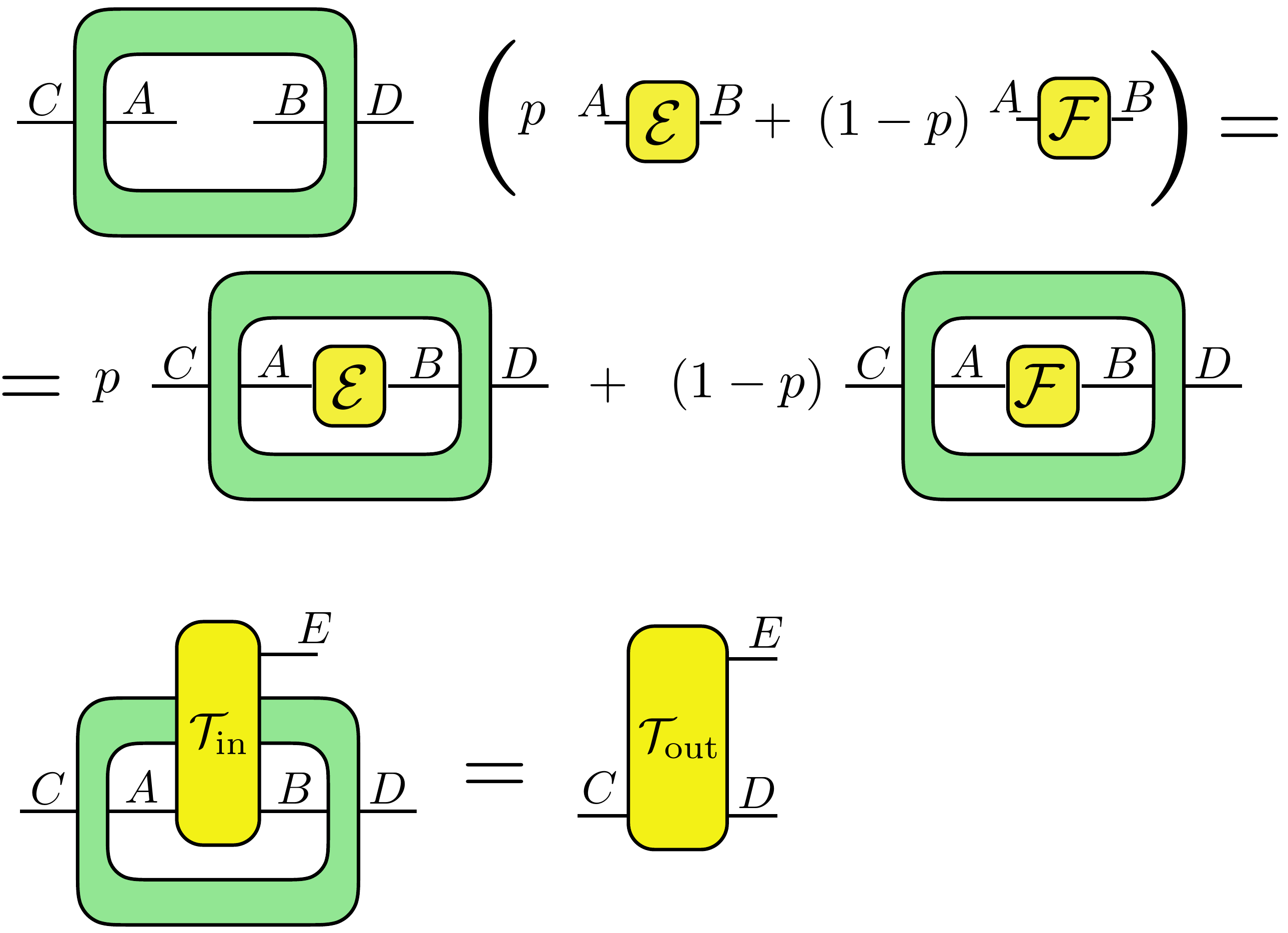}
   \caption{\emph{(Top)} The linearity condition for a higher
     order map of type $(A\to B) \to (C \to D)$.
\emph{(Bottom)}     The completely admissible-preserving
     condition for the same map.}
   \label{completepositivity}
 \end{figure}

 We now need to characterize those maps which are
 \emph{physical} or \emph{admissible}, i.e. they
 are compatible with the probabilistic structure
 of quantum theory. This is achieved by a recursive
 generalization of the admissibility conditions
 for quantum operations that we cited at the
 beginning of this section \footnote{We will only
   give an overview of the axiomatic
   construction. For the full treatment see the
   supplemental material  \cite{supplement} and
   Ref. \cite{doi:10.1098/rspa.2018.0706}}.
 The linearity
 assumption is straightforwardly extended: a map
 of type $x$ is an operator in
 $\mathcal{L}(\hilb{H}_x)$ where 
 $\hilb{H}_x := \bigotimes_{i\in x}\hilb{H}_i$
and the index $i$ runs over all the (non-trivial)
elementary types occurring in the expression of
$x$ (e.g. if $M$ is a map of type $((A\to B) \to I) \to C$
then $M \in \mathcal{L}(\hilb{H}_x)$
and $\hilb{H}_x = \hilb{H}_A \otimes
\hilb{H}_B \otimes
\hilb{H}_C)$)\footnote{Throughout this paper we
  will also make implicit use of the isomorphism
  $H\otimes K\equiv K\otimes H$ and we will
  identify these two Hilbert spaces, i.e.
  $H\otimes K = K\otimes H$. }.

On the other hand, the generalization of complete
 positivity is more subtle.  In broad
 terms, we require that an admissible
 map $M$ of type $x \to y$ should transform
 admissible maps of type $x$ into
 admissible maps of type $y$ even in the presence
 of correlations with another elementary system
 $E$, where this extension is understood  by
 considering the map
 $\mathcal{M} \otimes \mathcal{I}_E$.
   If  $M$ satisfies this requirement we say
 that it is \emph{completely
   admissible-preserving}.
   In order to show how this works,
 let us consider a map
$M $ of type
 $(A \to B ) \to (C \to D) $
 (see Fig. \ref{completepositivity}).  In order to be admissible,
$\mathcal{M}$ (remember that $M$ denotes the Choi
operator of $\mathcal{M}$) should transform
 admissible maps of type $A \to B$ to admissible
maps of type $C \to D$ , i.e. quantum operations
to quantum operations.  The condition of being
completely admissible-preserving requires that
$\mathcal{M} \otimes \mathcal{I}_E$ must transform
admissible maps of type $A \to BE$ to admissible
 maps of type $C \to DE$ for any elementary
system $E$, i.e $\mathcal{M} \otimes
\mathcal{I}_E$ must transform quantum operations
with bipartite output to quantum operations
with bipartite output.

 The third admissibility condition is a
 generalization of the subnormalization constraint
 for quantum operations.  First, we say that a
 completely admissible-preserving map $M$ of type
 $x \to y$ is \emph{deterministic} if it
 transforms deterministic maps of type $x$ into
 deterministic maps of type $y$ even in the
 presence of correlations with another elementary
 system $E$. Then, we say that a completely
 admissible-preserving map $M$ of type $x \to y$
 is \emph{admissible} if there exists a set of maps
 $N_i$ of type $x \to y$ such that $i)$ the $N_i$
 are completely admissible preserving and $ii)$
 the map $M + \sum_iN_i $ is deterministic.
 Intuitively, this condition requires that an
 admissible higher order map should arise in
 a higher order instrument, i.e. a collection of
 higher order maps that sum to a deterministic one.
 Clearly, any deterministic map is admissible.  We
 notice that the subnormalization condition for
 quantum operations can be rephrased in a similar
 way by requiring that there exists a collection of
 completely positive maps that sum to a quantum
 channel.

One can check that this 
recursive construction is well defined: once the admissibility is
given for elementary types (and we know that an admissible map of
elementary type is just a quantum state), the
admissibility for arbitrary types follows.
Furthermore,  this 
axiomatic approach never explicitly 
refers to the mathematical properties of the
maps in the hierarchy. The mathematical structure
of quantum theory only enters at the ground level
of the hierarchy, i.e. for elementary types, and
propagates inductively to the whole hierarchy.
The mathematical characterization of
admissible higher order maps is given in the
following proposition,
whose proof can be found in Ref.
\cite{doi:10.1098/rspa.2018.0706},
which provides an explicit constructive formula 
\begin{proposition}
  \label{thm:charthm}
  Let $x$ be a type and $M \in
  \mathcal{L}(\hilb{H}_x)$ a map of type $x$. Let
$\set{Hrm}(\hilb{H})$ and  $\set{Trl}(\hilb{H})$ denotes  the subspace of
  hermitian operator and traceless hermitian
  operators, respectively.
  Then $M$ is \emph{admissible} if and only if $M \geq 0$
    and $M \leq D$ for a deterministic map $D$.
A map $D$ of type $x$ is \emph{deterministic} if and
  only if    $D \geq 0$ and
  \begin{align}
    \label{eq:2}
D = \lambda_x I_x + X_x, \quad \lambda_x \geq 0,
    \quad X_x \in \set{\Delta}_x \subseteq \set{Trl}(\hilb{H}_x)
  \end{align}
 where 
  $\lambda_x $ and
  $\set{\Delta}_x$ are
  defined recursively as
  \begin{align}
    \begin{split}
     & \lambda_E = \frac{1}{d_E}, \;
      \lambda_{x\to y} = \frac{\lambda_y}{d_x
      \lambda_x}, \qquad \set{\Delta}_E =
    \set{Trl}(\hilb{H}_E), \\
       & 
       \set{\Delta}_{x \to y} =
[\set{Hrm}(\hilb{H}_x) \otimes {\set{\Delta}_{y}}]
    \oplus
[\overline{\set{\Delta}}_{x}   \otimes \set{\Delta}^\perp_{y} ],  
    \end{split}
  \end{align}
  and where $\set{\Delta}^\perp$
  denotes the orthogonal
   complement (with respect to the Hilbert-Schmidt
   inner product) of $\Delta$ in
   $\set{Hrm}(\hilb{H})$ while
   $\overline{\set{\Delta}}$ is
    the orthogonal
   complement in
   $\set{Trl}(\hilb{H})$.
  We denote with $\Eva{x}$ the set of admissible
maps of type $x$ and with $\Evd{x}$ the set of
deterministic maps of type $x$.
\end{proposition}
This result is the main tool in the study of the
hierarchy of higher order maps.  For example, it
proves that the axiomatic construction
is consistent with the
characterization of quantum operations and also
incorporates
all the higher order maps that have been
considered in the literature, namely networks of
quantum operations (also known as \emph{quantum
  combs}
\cite{PhysRevA.80.022339,bisioacta}) and
\emph{process matrices}
\cite{Oreshkov:2012aa,araujo2015witnessing,oreshkov2016causal}.

\section{One type, many uses.}
For our purposes, it is important to stress
that a given map $M$ could in principle
be assigned more than one type.  If
$M \in \Evd{x} \cap \Evd{y}$ (we restrict to the
deterministic case for sake of simplicity), that
means that $M$ could be equivalentely used as a
map of type $x$ or as a map of type $y$.  For
example, let us consider a map
$M \in \Evd{(A \to B) \to (C \to D)})$ (see
Fig. \ref{fig:supermaps}).
\begin{figure}
    \centering
      \includegraphics[width=0.9\linewidth]{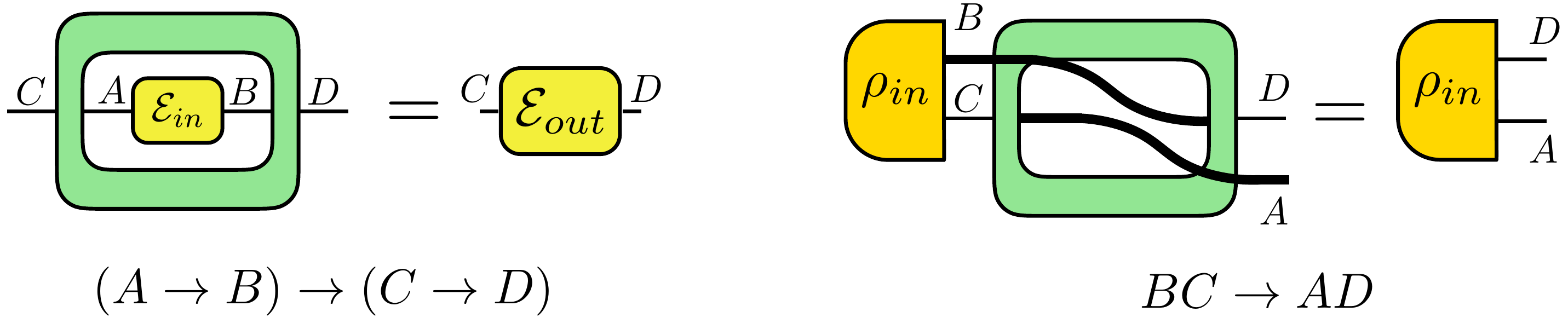}
    \caption{The same map $M$ could be regarded as a map
     from quantum operations to quantum operations
     %(i.e. $M \in \Evd{(A \to B) \to (C \to D)})$)
     or as a channel from bipartite states to
     bipartite states %(i.e. $M \in \Evd{CB \to
     % AD}$
   }
\label{fig:supermaps}
\end{figure}
From Proposition \ref{thm:charthm}, it follows that $M$ is a positive
operator which satisfies
$\Tr_D [M] = I_B \otimes M'$ , $\Tr_A [M'] = I_C$
which implies that
$\Tr_{AD} [M] = I_B \otimes I_C$. This means that
$M$ is a quantum channel from systems $C$ and $B$
to systems $A$ and $B$, namely $M \in \Evd{CB \to
  AD}$.
A relevant case is when any admissible map of type $x$ is also an
admissible map of type $y$:
\begin{definition}
  \label{def:partialorder}
  If $\Eva{x}
  \subseteq \Eva{y}$
   we say that $x$ is
\emph{included} in $y$ and we write $x \subseteq
y$. If both $x \subseteq y$ and $y \subseteq x$
we say that $x$ is equivalent to $y$ and we write
$x \equiv y$ \footnote{A straightforward consequence of Proposition
\ref{thm:charthm} is that
$\Evd{x}
\subseteq \Evd{y} \implies \Eva{x}
  \subseteq \Eva{y}$; therefore,
we can replace 
$\Eva{x}
\subseteq \Eva{y}$ with
$\Evd{x}
\subseteq \Evd{y}$
in
Definition~\ref{def:partialorder}.}.
\end{definition}
Thanks to Proposition \ref{thm:charthm} one prove
inclusions and
equivalences between types.
The previous example proved the inclusion
$(A \to B) \to (C \to D) \subseteq (CB \to AD)$.
Let us now define the types 
\begin{align}
  \label{eq:11}
  \overline{x}  := x \to I, \\
  x \otimes y := \overline{x \to \overline{y}}.
\end{align}
Then we have the 
equivalences 
\cite{doi:10.1098/rspa.2018.0706}
\begin{align}
  \label{eq:12}
  &\overline{\overline{x}} \equiv x, \\
  \label{eq:tensorequivalences}
  &\begin{aligned}
    &x \otimes y \equiv y \otimes x, \\
    &(x \otimes y)
  \otimes z \equiv x \otimes (y \otimes z) .   
  \end{aligned}
\end{align}
The maps of type $x \to I$ are linear functionals
on maps of type $x$. One can verify that the affine span of deterministic map
of type $x\to I$
is the dual affine space\cite{chiribella2016optimal} of the
affine span
of deterministic map
of type $x$.

The identities in Eq. \eqref{eq:tensorequivalences}
justify why the expression $x \otimes y$ is called
the \emph{tensor product } of $x$ and $y$. It is
worth noticing
that $\Evd{x\otimes y} $ is
strictly larger than the convex hull
$\set{Conv} \{ \Evd{x} \otimes \Evd{y}\}$. For
example, the set
$\Evd{(A\to B)\otimes (C \to D)} $ is the set of
\emph{non-signalling} channels from $AC$ to $BD$.
Process matrices, which describe the most general
quantum correlations among $N$ distant parties
that can perform local experiments
\cite{Oreshkov:2012aa,branciard2015simplest,PhysRevX.8.011047},
are maps of type
$ \overline{\bigotimes_{i=1}^N (A_i \to B_i) }$,
i.e.  functionals on non-signalling channels.  The
identities of Equation \eqref{eq:12} are useful
for proving other type idenitities, for example we
have
$\overline{A \to B} = \overline{A \to
  \overline{\overline{B}}}= A \otimes \overline{B}
$. One can easily show that
$\Evd{\overline{B}\otimes A := \{ I_B \otimes
  \rho_A , \Tr\rho_A =1 \} } $ i.e. the most
general deterministic functional on $A \to B$ is a
map which prepares an arbitrary state on system
$A$ and then discard system $B$.

It turns out that all the examples of higher order
maps that have been considered % (i.e. networks of
%channels and process matrices)
share the feature that they can always be
interpreted (and used) as multipartite quantum
channels which arise by ``stretching out'' the
wires (see Fig.~\ref{fig:supermaps}). This is a
general feature of any type in the hierarchy. The set
$\set{Ele}_x$ of the (non-trivial) elementary
types that occur in the expression of the type $x$
can be split into two disjont subsets $\set{in}_x$
and $\set{out}_x$ such that any map in $\Evd{x}$
is also a channel from $\set{in}_x$ to
$\set{out}_x$.  Moreover, we have that any map
which discards the systems $\set{in}_x$ and prepare
an arbitrary state on $\set{out}_x$ belongs to
$\Evd{x}$.
\begin{figure}
    \centering
      \includegraphics[width=\linewidth]{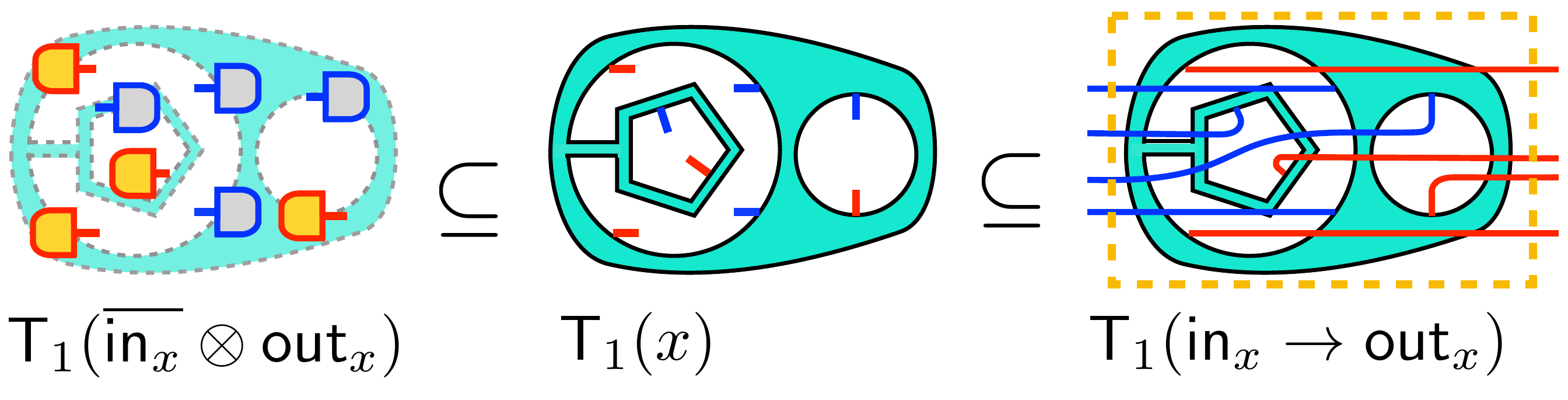}
    \caption{ Any (non-trivial) elementary type
      occurring in a type $x$ is identified as an
      input (blue) or output (red) system.
On the one hand, a higher order map of type $x$ can be
used as a
      quantum channel from $\set{in}_x$ to
      $\set{out}_x$ (right inclusions).
    On the other hand, a map which discards
    $\set{in}_x$ and prepare a state in
    $\set{out}_x$, i.e. a map of type
    $\overline{\set{in}_x} \otimes \set{out}_x$,
    is always a deterministic map of type $x$
    (left inclusion).
  }
\label{fig:inclusions}
    \end{figure}
These considerations are pictorially
represented in Fig. \ref{fig:inclusions} and 
are stated in the following Proposition.
\begin{proposition}\label{prp:inclusion}
  For any type $x$  we denote with $\set{Ele}_x$
the set   of the (non-trivial) elementary
types that occur in the expression of $x$
and  we define
  the function $K_x :  \set{Ele}_x \to \{0,1 \}$ as
   $K_x(A) = \# [``\to" ] + \# [``(" ]  \pmod2$
 where $\# [``\to " ]$ and $\#  [`` ( " ] $ denote
the number of arrows $" \to "$ and open round brackets
$" ( "$ to the right of $A$ in the expression of
$x$, respectively.
If we define the  
the sets $\set{in}_x := \{A \in \set{Ele}_x
\mbox{ s.t. } K_x(A) = 1\}$ and $ \set{out}_x : =
\set{Ele}_x \setminus  \set{in}_x$,
we have the inclusions
\begin{align}
  \label{eq:3}
  \overline{\set{in}_x} \otimes \set{out}_x
  \subseteq x \subseteq \set{in}_x \to \set{out}_x
\end{align}  
\end{proposition}
\begin{proof}
The proof of this proposition relies on Proposition
\ref{thm:charthm} and can be found in the
supplemental Material
\cite{supplement}.
\end{proof}

\section{Compositional structure of higher order
  maps.}
Applying a quantum channel to a state, is not the
only way to use a channel.
We can compose a channel from $A$ to $B$ with one from $B$
to $C$ to obtain a channel
from $A$ to $C$. More generally, multipartite
channels can be connected only through
some of their inputs and outputs,
e.g. $\begin{array}{l}
        \!\!\!\!\includegraphics[width=0.08\linewidth]{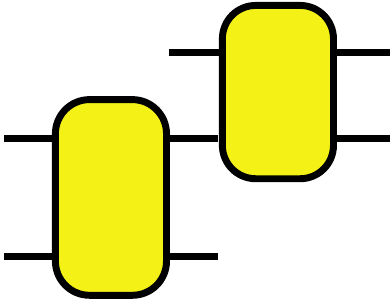}\end{array}$.

What does it mean to compose two higher order maps?
Any map of type $x \to y$ can clearly be composed
with a map of type $y \to z$ to obtain a map of
type $x \to z$ but what about more general
composition schemes? 
 Let's consider two maps
$R=\begin{array}{l}\includegraphics[width=0.17\linewidth]{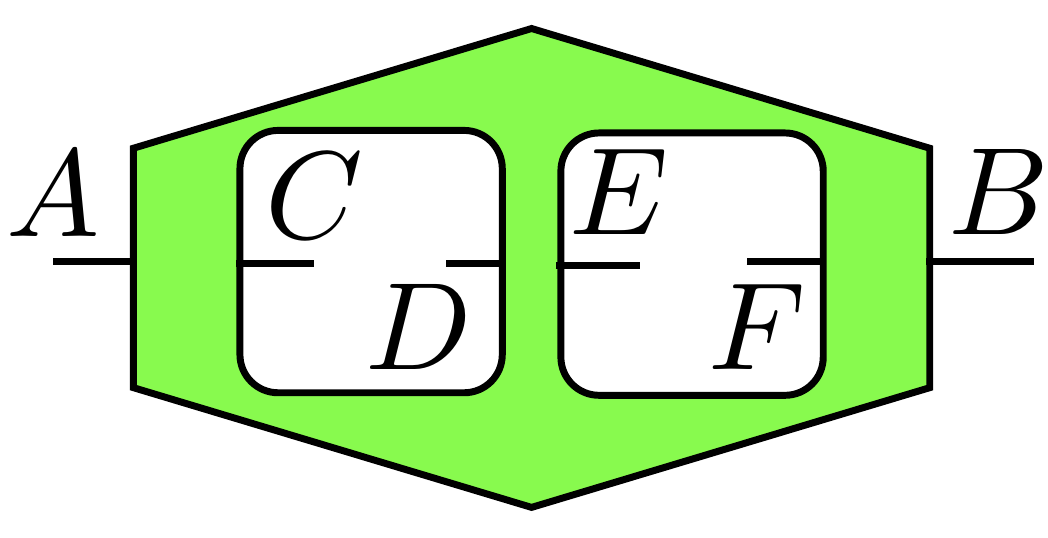}\end{array}$
of type $x$ and
$T=\begin{array}{l}\includegraphics[width=0.13\linewidth]{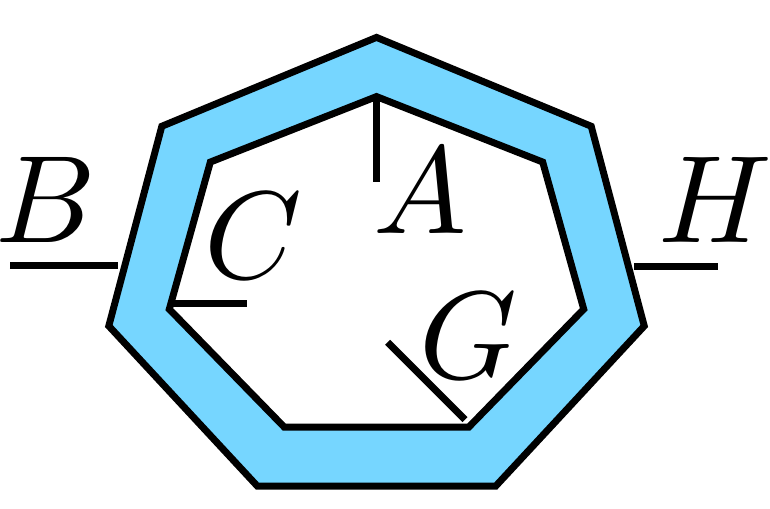}\end{array}$
of type $y$. Intuitively, to compose $R$ and $T$
should mean to connect the elementary systems that
they share, i.e.
$R * T := \begin{array}{l}\includegraphics[width=0.4\linewidth]{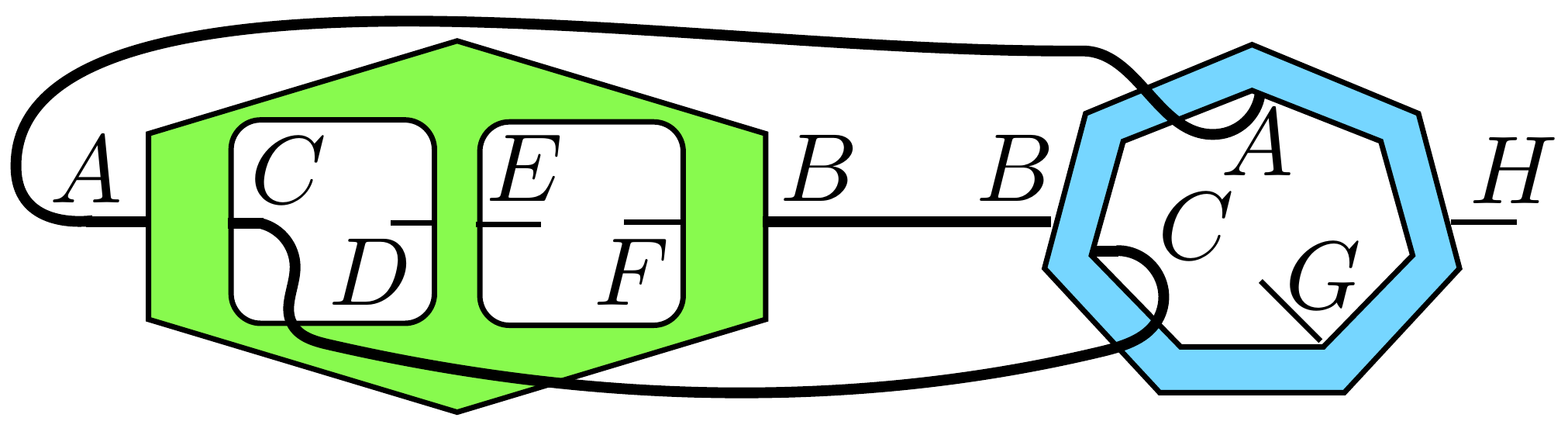}\end{array}
= \begin{array}{l}\includegraphics[width=0.3\linewidth]{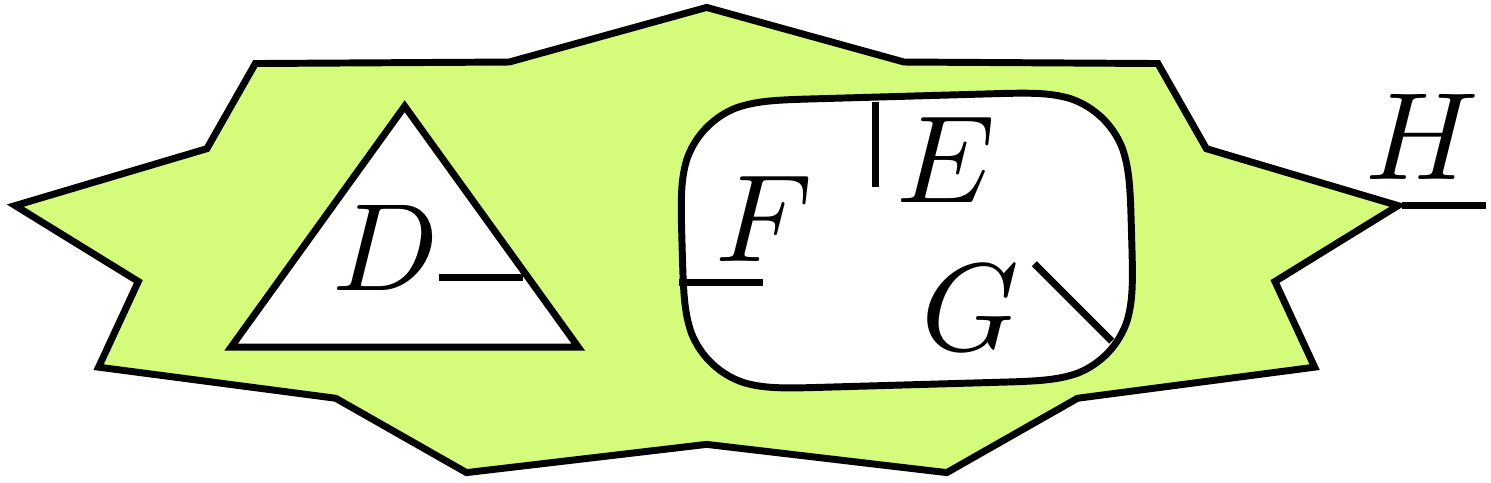}\end{array}$.
In terms of the Choi isomorphism, this operation is 
given by the so called \emph{link product} \cite{PhysRevA.80.022339}, namely 
\begin{align}
  \label{eq:5}
  \begin{aligned}
    & R *  T= \Tr_{\set{S}}
  [(R\otimes
  I_{\set{B}\setminus \set{S}})
  (I_{\set{A}\setminus \set{S}}  \otimes T^{\theta_{\set{S}}} ) ]\\
 & R \in  \mathcal{L}(\hilb{H}_{\set{A}}) , \;\;
  T \in  \mathcal{L}(\hilb{H}_{\set{B}}) ,
  \end{aligned}
\end{align}
where $\set{S}:= \set{A} \cup \set{B}$ is the set of systems that $R$ and
$S$ share and $T^{\theta_{\set{S}}}$ is the
partial traspose of $T$ on the systems of
$\set{S}$. Notice that for the case $\set{S} =
\emptyset$ we have $ R *  T = R\otimes T$

However, we cannot expect that  any composition
between two arbitrary maps of  type $x$ and $y$
should be physically admissible. For example, 
the composition of  two effects $E_1 ,E_2 \in
\Eva{\overline{A}}=\{  E\geq0 , E \leq I  \}$ acting on the same
system $A$, i.e. 
$E_1 * E_2 = \begin{array}{l}\includegraphics[width=0.12\linewidth]{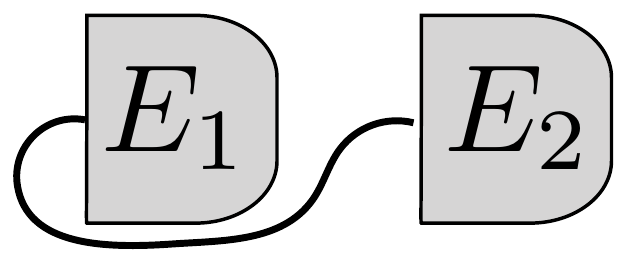}\end{array}$,
is a meaningless operation \footnote{Indeed if we
choose e.g.
$\hilb{H}_A = \mathbb{C}^3$ and $E_1 = E_2 =
\ketbra{0}{0} +\ketbra{1}{1}$ we obtain $E_1 * E_2 =
2$ which has no physical interpretation}.
At the very least, we should require that
the composition between two types $x$ and $y$ is
admissible if, whenever we compose a map of type
$x$ with a map of type $y$ we obtain something
that could be interpreted as an admissible higher
order map. This motivates the following
definition:
\begin{definition}
  \label{def:admiscomp}
Let $x$ and $y$ be two types such that they share
a set of elementary types.
We say that the composition $x * y$ is
\emph{admissible}
if
\begin{align}
  \begin{split}
  \label{eq:7}
  \forall R \in  \Evd{x} , \,  \forall T \in \Evd{y}, \; \exists
  z \mbox{ s.t. } R * T \in  \Evd{z}
  \end{split}
\end{align}
\end{definition}
\noindent We notice that Equation \eqref{eq:7}
only deals with deterministic maps. This is not
too permissive since Equation \eqref{eq:7} 
implies that 
$ R * T \in  \Eva{z}$ for any $R \in  \Eva{x}$ and
 $T \in \Eva{y}$.

 Our next goal is to provide a 
 characterization of the composition of type that
 are admissible according to Definition
 \ref{def:admiscomp}. Before doing that,
 we introduce the following slight variation of the 
composition of maps.
\begin{definition}
  Let $x$ be a type  and let $A, B \in
  \set{Ele}_x$ with $\dim{\hilb{H}_A} =
  \dim{\hilb{H}_B}$ \footnote{The condition $\dim{\hilb{H}_A} =
  \dim{\hilb{H}_B}$ makes the discussion
  more transparent but it is immaterial}. 
  We say that the \emph{contraction} $\mathcal{C}_{AB}(x)$ of $A$ and
  $B$ in $x$ is admissible if
  \begin{align}
    \label{eq:9}
    \forall R \in  \Evd{x}  \; \exists
  z \mbox{ s.t. } \mathcal{C}_{AB}(R) :=  R * \Phi_{AB} \in  \Evd{z}
  \end{align}
  where $\Phi_{AB} := \sum_{i,j} \ket {ii} \bra{jj} \in
  \Lin{\hilb{H}_A\otimes \hilb{H}_B}$.
\end{definition}
\noindent Basically, this definition determines whether it
is allowed to connect system $A$ with system $B$
for any map of type $x$, e.g.
$\mathcal{C}_{AB}(R)
= \begin{array}{l}\includegraphics[width=0.2\linewidth]{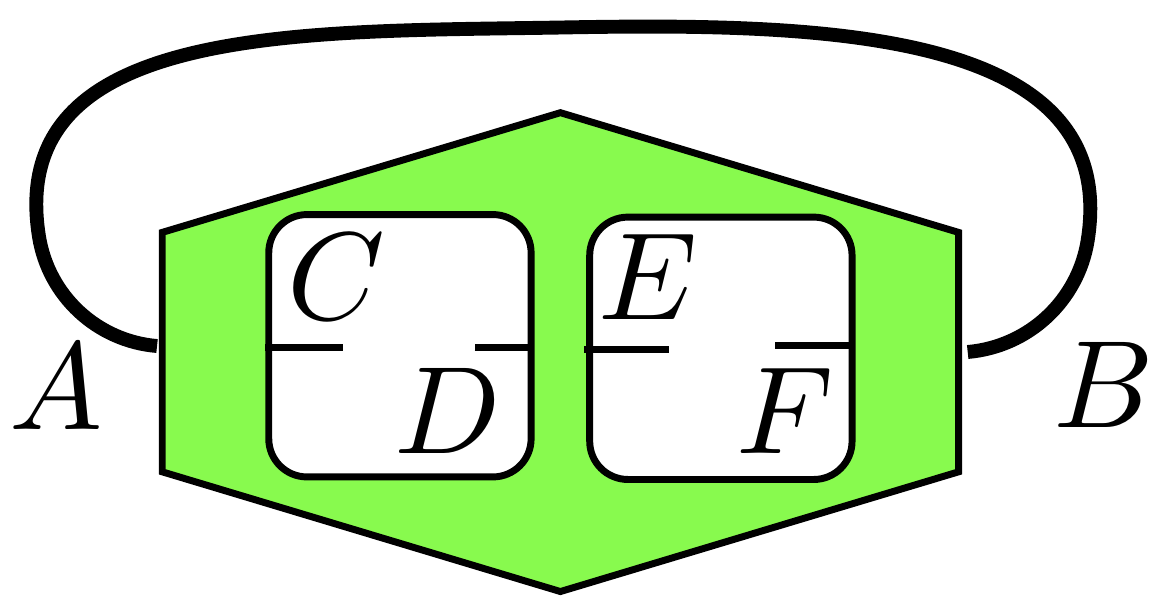}\end{array}$. A 
straightforward computation proves the
(diagrammatically trivial) fact 
that the composition of two
higher order maps can be written in terms of contractions as
follows:
\begin{align}
  \label{eq:13}
  R * S =  \mathcal{C}_{AA} ( \mathcal{C}_{BB}
  (\dots (R \otimes S)\dots ) 
\end{align}
where $A, B ,\dots $ are the elementary types
involved in the composition.  This observation
leads to the following result presented informally
here whose formalisation and proof is given in the
supplemental material \cite{supplement}.
\begin{lemma} \label{lmm:manycontractions}
  The composition $x * y$ is admissible if and
  only if the contractions
  $\mathcal{C}_{AA} ( \mathcal{C}_{BB}
  (\dots (x \otimes y)\dots ) $ are admissible.
\end{lemma}

%\begin{proof}
% The ``if'' direction of the proof is a
% straightforward consequence of Equation
% \eqref{eq:13}. Since the $\Evd{x
% \otimes y}$ is a larger set
%  than the convex hull $\set{Conv} \{
%   \Evd{x} \otimes \Evd{y}\}$, ``only if'' part is less trivial
% and it needs the characterization of $\Evd{x
% \otimes y}$ which follows from Proposition \ref{thm:charthm}
% (see
% Ref. \cite{supplement} for a detailed proof).  
% \end{proof}

\begin{figure}[t]
    \includegraphics
   [width=0.47\textwidth]{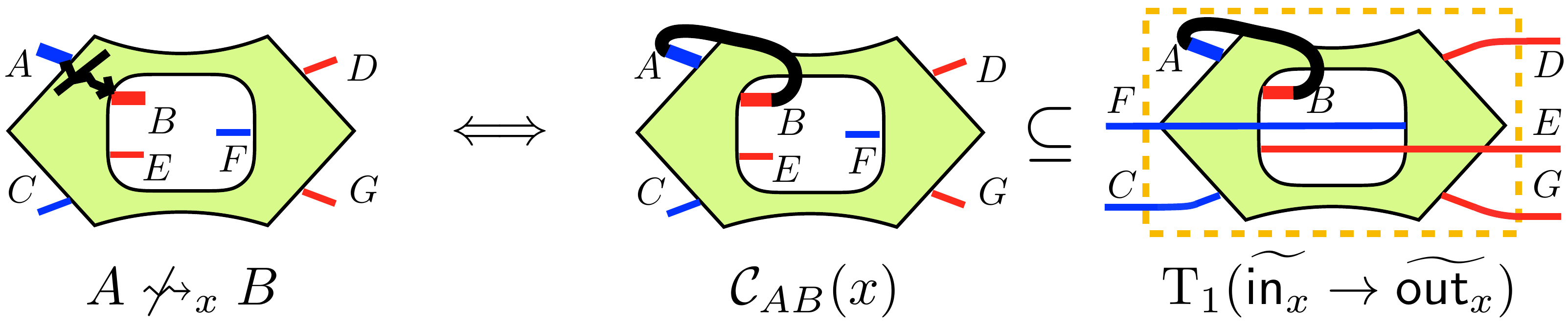}
   \caption{Admissibility of a contraction
     and no-signalling. 
     %$x = (E \to ((A \to D)\to (B \to
     %(\overline{C} \to F)) ))\to G $
     }
   \label{fig:admcontraction}
 \end{figure}

Thanks to Lemma \ref{lmm:manycontractions}  the characaterization of
  the admissible compositions is equivalent to
  the characterization of the admissible
  contractions. This is achieved by the following
  proposition which we prove in the supplemental material
  \cite{supplement} (see also Fig. \ref{fig:admcontraction}).
\begin{proposition}\label{prop:contractadmiss}
Let $x$ be a type and $A,B \in \set{Ele}_x$.
If $A,B \in \set{in}_x$ or $A,B \in \set{out}_x$
then  $\mathcal{C}_{AB} (x)$ is not admissible.
If $A \in \set{in}_x$ and  $B \in \set{out}_x$
then 
$\mathcal{C}_{AB} (x)$ is admissible if and only
if, for any $R \in \Evd{x}$ we have that 
$\mathcal{C}_{AB} (R) \in
\Evd{\widetilde{\set{in}_x} \to
  \widetilde{\set{out}_x} }$
where we defined $\widetilde{\set{in}_x} :=
\set{in}_x \setminus A$ and
$\widetilde{\set{out}_x} :=
\set{out}_x \setminus B$.
\end{proposition}
% \begin{proof}
%  ammissibile.
% Supponendo per assurdo che una contrazione AB
% input input sia ammissibile
% \end{proof}

Thanks to Proposition \ref{prop:contractadmiss},
the admissibility of contractions (and therefore
of compositions) can be stated in the more
elementary language of quantum channels: a
contraction between $A$ and $B$ is admissible if
and only if $A$ and $B$ can be connected in a
loop.  This result dramatically reduces the
complexity of the problem, since we only need to
establish that the contracted maps belongs to the
type of quantum channels
$\widetilde{\set{in}_x} \to
\widetilde{\set{out}_x}$.  It may seems that we
have to check this condition for all the
infinitely many maps that belongs to the type $x$,
but this is not the case. Indeed, it is possible
to prove that the contraction transforms the set
of deterministic maps $\Evd{x}$, which, thanks to
Proposition~\ref{thm:charthm}, is characterized by
a subspace of linear operators, to another set of
operators which is also characterized by a linear
subspace. As a consequence, a contraction is
admissible if and only if this subspace generated
by the contraction is a subspace of the one that
characterizes the type of quantum channels. Thanks
to this observation, the admissibility of a
contraction can be further simplified to a finite
combinatorial problem. The proof of these
statements requires several technical steps which
can be found in the supplemental material \cite{supplement}.

One could easily extend the statement of
Proposition \ref{prop:contractadmiss} to multiple
contractions: the contractions
$\mathcal{C}_{AA}(\mathcal{C}_{BB}( \dots
x)\dots)$ are admissible if and only if
$\mathcal{C}_{AA}(\mathcal{C}_{BB}( \dots
R)\dots) \in \set{T}_1(\tilde{\set{in}} \to \tilde{\set{out}})$.
As we prove in the supplemental material
\cite{supplement} (see. Proposition
\ref{Sprop:subsetofcontractionarenotadmissibleyoucannotrecover})
the admissibility of a set of contraction
follows by chaining
together verifications of the admissibility of
each singular contraction in a sequence  (since
different contractions commute,
the order in which we put the
set of contractions is immaterial).
A necessary condition for a set of contraction to
be admissible is that each individual contraction
in the set is admissible. On the other hand,
the
joint admissibility of a set of contractions is a
requirement which is strictly stronger than the
individual admissibility of each contraction in
the set. Indeed, $\mathcal{C}_{AB}$ and
$\mathcal{C}_{CD}$ can be admissible (e.g. $
 \begin{array}{l}\includegraphics[width=0.2\linewidth]{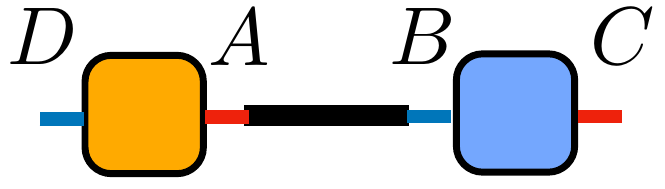}\end{array}$
 and
$\begin{array}{l}\includegraphics[width=0.2\linewidth]{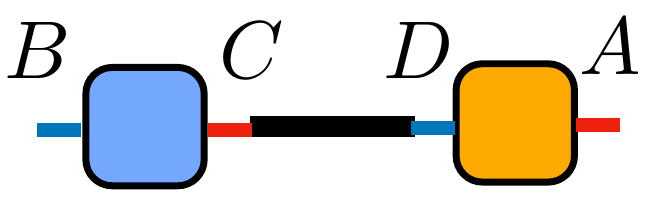}\end{array}$) but
$\mathcal{C}_{AB}\mathcal{C}_{CD}$ may not be
(e.g. $\begin{array}{l}\includegraphics[width=0.2\linewidth]{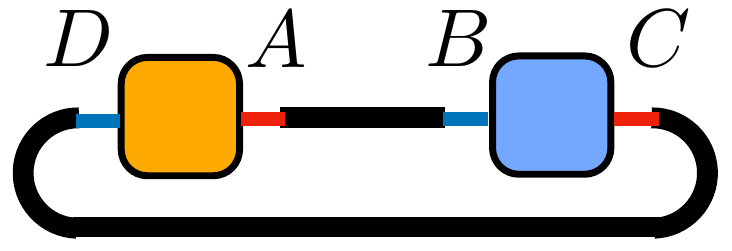}\end{array}$).

 \section{Causal structure of higher order maps}
 We have seen that the admissibility condition of
 the composition of higher order maps can be
 traced back to the admissibility of closing
 input and output of a channel
in a loop. We now wonder what is the physical
significance of the admissibility of these loops.
In particular,
we expect that this is related to how
information can flow between the quantum systems,
i.e. to the causal structure of the channel.

We remind\cite{PhysRevA.64.052309,refId0} that a
multipartite quantum channel
$R \in \Evd{A_1A_2 \dots A_n \to B_1B_2 \dots B_m
} $ does not allow for signalling from $A_i$ to
$B_j$ if and only if
$\Tr_{\neg B_j}[R] = I_{A_i} \otimes R'$, and we
write $A_i \not\rightsquigarrow_{R} B_j$,
\footnote{We use the shortcut
  $\neg B_j := B_1B_2 \dots B_{j-1}B_{j+1}\dots
  B_m $}.  Otherwise, we say that
$R \in \Evd{A_1A_2 \dots A_n \to B_1B_2 \dots B_m
} $ does allow for signalling from $A_i$ to $B_j$
and we write $A_i \not\rightsquigarrow_{R} B_j$.

 Intuitively, we would expect that if we have
 signalling from 
 $A$ to $ B$ it would not be
 possible to close this information flow in a loop
 and therefore that the contraction of $A$ with
 $B$ is not admissible. On the other hand, this
 would be possible if no information flows from
 $A$ to $B$. As the following proposition shows, this
 intuition is correct.
 \begin{proposition}\label{prop:causal}
   Let $x$ be a type, $A\in \set{in}_x$, and $B
   \in \set{out}_x$. Then we have that
$
   \mathcal{C}_{AB}(x) \mbox{ is admissible} \iff  A
 \not\rightsquigarrow_{x} B
  $
where $A
\not\rightsquigarrow_{x} B$ means that
$A
 \not\rightsquigarrow_{R} B$ for any $R \in
 \Evd{x}$. 
 \end{proposition}
 \begin{proof} 
   First we prove that $\mathcal{C}_{AB}(x)$ is
   admissible if and only if 
$x \subseteq (B \to A) \to (\widetilde{\set{in}} \to
\widetilde{\set{out}} )$. Then,  the realisation
theorem of supermaps \cite{PhysRevLett.101.060401,chiribella2008transforming} implies that
$A
 \not\rightsquigarrow_{R} B \iff R \in \Evd{(B \to A) \to (\widetilde{\set{in}} \to
\widetilde{\set{out}} )} $. The thesis
follows by combining these 
results (see the supplemental material
\cite{supplement} for the details).
    \end{proof}
 
Proposition \ref{prop:causal}  relates the causal and the
compositional structure of a type: the admissible
contractions are the ones that are allowed by the
causal structure. 
The admissibility of a set of contractions (and
therefore of a composition)
is therefore equivalent to 
verify a sequence of no-signaling
conditions.

Finally, we present a result that allows to
determine the signalling relations
between the elementary systems of type $x$ from
the expression of the type itself. The proof
can be found in the supplemental material
\cite{supplement}.
 \begin{proposition} \label{prp:algosignalling}
   Let $x$ be a type, $A \in \set{in}_x$,
   $B \in \set{out}_x$. Then, there exists a
   unique type $y \prec x$ such that
   $\{A, B \} \subseteq \set{Ele}_y$ and $y' \prec y
   \implies \{A, B \} \not \subseteq
   \set{Ele}_y$. Moreover, we have
   \begin{align}
     \label{eq:8}
     A \in \set{in}_y \implies A
\rightsquigarrow_{x} B, \quad
     A \in \set{out}_y \implies A
     \not\rightsquigarrow_{x}  B .
   \end{align}
 \end{proposition}
 Thanks to Proposition \ref{prp:algosignalling} we
 have an efficient (i.e. polynomial in the size of the
 type) algorithm for determining whether or not a
 type $x$ allows for signalling between
 $A$ and $B$: $i)$ find $y$ by
 computing the smallest substring of the type $x$
 which contains both $A$ and $B$ and has balanced
 brackets; $ii)$ Compute $K_y(A)$ and $K_x(A)$
 (see Prop.  \ref{prp:inclusion}); then $x$ allows
 signalling form $A$ to $B$ if and only if
 $K_y(A)=K_x(A)$. 

 \section{Discussion} We extended
 higher order quantum theory making it into a
 higher order \emph{computation} in which the
 object of the framework can be composed together.
 The construction and the results that we derived
 have been grounded on
assumptions which translate
the requirement that the probabilistic structure
of quantum theory should be preserved.

In particular, no notions of local interventions
or of causal
order between parties has been used. However, somewhat
surprinsingly, the \emph{compositional} structure of higher order
maps is intimately related with the causal
relations among the systems that constitute a
given type, since it is the signalling structure
that dictates whether a given contraction is
admissible or not.

Our results focused on the
signalling relation between couples of systems,
but the compositional structure generally involves
more systems at the same time. It is therefore
interesting to deepen our understanding of the
relation between causal and compositional
structure in this more general case:
how does a contraction modify the causal structure
of a type?  Is it possible to retrieve this
information only from the expression of the type itself?

A further line of development of this research is
the generalization of the framework beyond quantum
theory itself. Since the axioms of
higher order quantum theory do not rely on the
mathematical structure of quantum theory
\footnote{The most important assumption is
  the Choi isomorphism, that can be always
  provided in theories where local
  discriminability holds.}, one could wonder what
would emerge if we replaced the fundamental level
of hierarchy with a general operational
probabilistic theory. Since in a general
operational theory causality and non-signalling
exhibits a finer structure, we could expect that
these features propagates and maybe get emphasized
throughout the hierarchy. This could inform the
search for principles that single out quantum
mechanics and deepen our understanding of quantum
computing.

\begin{acknowledgments}
  L.A. acknowledge financial support by the
  Austrian Science Fund (FWF) through BeyondC
  (F7103-N48), the Austrian Academy of Sciences
  ({\"O}AW) through the project “Quantum Reference
  Frames for Quantum Fields” (ref. IF 2019 59
  QRFQF), the European Commission via Testing the
  Large-Scale Limit of Quantum Mechanics (TEQ)
  (No. 766900) project, the Foundational Questions
  Institute (FQXi) and the Austrian-Serbian
  bilateral scientific cooperation
  no. 451-03-02141/2017-09/02. L.A. also
  acknowledge the support of the ID 61466 grant
  from the John Templeton Foundation, as part of
  The Quantum Information Structure of Spacetime
  (QISS) Project (qiss.fr).
  A.B. acknowledge financial support from
  PNNR MUR project CN0000013-ICSC.
  The authors
  acknowledge useful discussions about the topic
  of the present paper with {\v{C}}. Brukner,
  A. Baumeler, M. Renner and E. Tselentis.
  
\end{acknowledgments}

\bibliographystyle{unsrtnat}
\bibliography{bibliography}

\clearpage

\appendix
\section*{\Large{SUPPLEMENTAL MATERIAL}}

  \section{Axiomatic approach to higher order
  quantum theory}

In this section we present the axiomatic framework for
higher order quantum theory. Most the material of
this section is
a review of
Ref. \cite{doi:10.1098/rspa.2018.0706}, which we
refer for a more exhaustive presentation. 

\subsection{The hierarchy of types}
The starting point of the framework is the notion of
\emph{type}:
\begin{definition}[Types]\label{Sdef:type}
  Every finite dimensional quantum system
  corresponds to an \emph{elementary type}
  $A$. The elementary type corresponding to the
  tensor product of quantum systems $A$ and $B$ is
  denoted with $AB$. The type of the trivial
  system, one dimensional system, is denoted by
  $I$. We denote with $\set{EleTypes}$ the set of
  elementary types. 

  Let
  $A := \set{EleTypes} \cup \{(\} \cup \{)\} \cup
  \{\rightarrow\}$ be an alphabet. We define the
  set of types as the smallest subset
  $\set{Types}\subset A^{\ast}$ such that:
\begin{itemize}
\item $\set{EleTypes} \subset \set{Types}$,
\item if $x,y\in \set{Types}$ then $(x\rightarrow y)\in\set{Types}$.
\end{itemize}
where $A^\ast$ is the set of words given
the alphabet $A$. % If $x$ is a type, we denote with
% $\set{Ele}_x$ the set of elementary types occurring
% in the expression of $x$.
\end{definition}
From  the previous definition it follows that a
type
is a string of
elementary systems, arrows and
balanced  brackets, that is  every open 
 bracket $"("$ is balanced by a closed one $")"$,
 for example:
 \begin{align}\label{Seq:3}
 x=(((A\to B)\to (C\to A ))\to (B \to
   E) )   
 \end{align}
 where $A,B ,\dots$
 are elementary types.
 % \begin{remark}
 %   From Definition~\ref{Sdef:type} we have that the
 %   same elementary type can appear more than once
 %   in the expression of a type, see e.g. Equation
 %   \eqref{Seq:3}.
 %   However, it is more convenient 
 % \end{remark}

 The hierachical structure of the set of types
 allows to define the following partial order
 between types.

\begin{definition}[Partial ordering $\preceq$]\label{Spartialordering}
Given the types $x$ and $y$, then  $x\preceq_p y$
if there exist a type $z$ such that either $y = x
\to z$ or $y = z \to x$. The relation
$x\preceq y$ is defined as the transitive closure
of the relation  $x\preceq_p y$.
\end{definition}
According to definition \ref{Spartialordering} we
say that $x\preceq y$ if the type $x$ appears in
the string which defines $y$, for example if
$ w := (x\rightarrow y)\rightarrow z$ we have
$x,z,y \preceq w$.  It is worth noticing that the
relation $\preceq$ is a well-founded and
Noetherian induction can be used.
\begin{lemma}
  The relation $x\preceq y$ is well-founded.
\end{lemma}
\begin{proof}
  Let's denote with $\# x$ the number of
  elementary types occurring in the expression of
  $x$.  We have that
  $x\preceq y \implies \# x < \#y$.  The binary
  relation $\preceq$ is well-founded if there
  exist no infinite sequence of types
  $\{ x_n \}_{n \in \mathbb{N}}$ such that
  $x_{n+1} \preceq x_n$ for any $n \in \mathbb{N}$
  (we are assuming the axiom of choice) \cite{hrbacek2017introduction}.
  If such a sequence
  $\{ x_n \}_{n \in \mathbb{N}}$ would exist, then
  $\{\#x_n \}_{n \in \mathbb{N}}$ would be an
  infinite sequence of strictly decreasing natural
  numbers.
\end{proof}

We now prove a couple of properties of the
hierarchy of types as a partially ordered set.
\begin{lemma}\label{Ssubstring}
  Let $s$ be a substring of the type $x$, then $s$
  is a type such that $s\preceq x$ if and only if
  the following conditions are met:
  \begin{enumerate}
  \item $s$ has balanced brackets
    \item the first symbol of $s$ is either an
      elementary type or the open bracket $"("$
    \item the last symbol of  $s$ is either an
      elementary type or the closed bracket $")"$
  \end{enumerate}
\end{lemma}
\begin{proof}
The necessity condition follows immediately from the definition of type.\\
For the sufficient condition, let us proceed by induction. Given $x$ an elementary type, 
the thesis is trivially satisfied. At this point we suppose that $\forall y\preceq x$ the 
thesis holds. So we  consider $y_1,y_2\preceq x$  such that
$x=y_1\rightarrow y_2$. If $s$ is a
substring of $ y_1$ then $s \preceq y_1$ by
induction hypothesis. Finally $s \preceq y_1$ and
$y_1 \preceq x$ imply that   $s \preceq x$ because
the relation $\preceq $ is transitive.
The same resoning applies if $s$ is a substring of $s_2$.

It remains the case in which  $s$ is a substring
of $ x$, with no empty intersection with both $y_1$ and $y_2$.
By definition of type we can rewrite $x=(y'_1)\rightarrow (y'_2)$. This makes evident that the 
string $s$ has at least one unbalanced bracket corresponding to the external ones in $(y'_1)$ and $(y'_2)$.
\end{proof}
\begin{lemma}\label{Sinnertype}
Let $z$ be a type and let  $A,B$ two elementary
types occurring in the expression of $z$  
such that $A$ precedes $B$ in the expression of
$z$.
Let us also assume that $A$ and $B$ do not appear
as the label of a
bipartite system, i.e. there exists at least one
arrow $"\to"$ between $A$ and $B$ in the
espression of $z$.
Then, there exists a unique type $x$ such that $x\preceq z$ and
$x=x_1\rightarrow x_2$ with either $A\preceq x_1
\land B\preceq x_2$
or $B\preceq x_1 \land A\preceq x_2$.
\end{lemma}
\begin{proof}
Consider the smallest substring of $z$ with balanced
brackets which contains $A$ and $B$ and that it
does not start or end with an arrow. From Lemma
\ref{Ssubstring} we know that this substring
identifies a type $x \preceq z$. Since there
exists an arrow between $A$ and $B$, we know that
$x$ is not elementary.
Then it must be $x = x_1 \to x_2$. Neither $x_1$
or $x_2$ must contain both $A$ and $B$, otherwise
$x$ would not be smallest string with balanced
parentheses that contain $A$ and $B$. We then
conclude that $A \in x_1$ and $B \in x_2$.
\end{proof}

\subsection{Higher order maps}

We now relate the hierarchy of types to the linear
maps on Hilbert spaces.
We start from the following preliminary
definition, which is still devoid of any physical
content but it is useful in order to develop the formalism

\begin{definition}[Generalized map]
If $x$ is a type, the set of generalized maps of type $x$, denoted by $\set{T}_{\mathbb{R}}(x)$, is defined by the following recursive definition.
\begin{itemize}
\item if A is an elementary type, then every
  $M \in \mathcal{L}(\hilb{H}_A)$ is a generalized
  map of type $A$,
  i.e.$\set{T}_{\mathbb{R}}(A) :=
  \mathcal{L}(\hilb{H}_A )$.
\item if $x$, $y$ are two types, then every Choi
  operator of linear maps
  $M : \set{T}_{\mathbb{R}}(x) \rightarrow
  \set{T}_{\mathbb{R}}(y)$, is a generalized map
  $M$ of type $(x \to y)$.
\end{itemize}
\end{definition}
Since we are describing linear map in terms of
their Choi operator, the choice of an orthonormal
basis is understood. 
The following characterization easily follows from the previous definition
\begin{lemma}[Characterization of generalized maps]
  Let $x$ be a type. Then
  $\set{T}_{\mathbb{R}}(x) =
  \mathcal{L}(\hilb{H}_x )$ where
  $\hilb{H}_x := \bigotimes_i \hilb{H}_i$ and
  $\hilb{H}_i$ are the Hilbert spaces
  corresponding to the elementary types $\{A_i\}$
  occurring in the expression of $x$.
\end{lemma}
\begin{proof}
  See Ref. \cite{doi:10.1098/rspa.2018.0706}
\end{proof}

The next step is to establish the set of
requirements
that a generalized map of certain type must satisfy
in order to represent a physical process.
Clearly, given two quantum systems $A$
and $B$, not every element
$M\in \mathcal{L}(\hilb{H}_A\otimes\hilb{H}_B)$
represents a Choi operator of an allowed quantum
transformation from $A$ to $B$, i.e. it is not a
physically \emph{admissible}  map of type $A \to
B$. Indeed, $M$ 
is an admissible process of type $A \to B$ if and
only if
it is the Choi operator of a quantum operations, i.e.
$0\leq M\leq N$, where $\text{Tr}_B[N]=I$.

Our goal is to generalise Kraus axiomatic
characterization of quantum operation to map of
arbitrary type, i.e. a map is admissible if it
preserves the probabilistic structure of quantum
theory. In Kraus' axiomatic characterization the
requirements that translates the compatibility
with the probabilistic structure of quantum theory
are linearity, normalization and complete positivity.
When applying this idea to maps of arbitrary type, the non
trivial problem is to find an appropriate
generalization of complete positivity.
In order to do that we will use the
following 
notion of \emph{type extension}:
\begin{definition}[Extension with elementary types]
  Let $x$ be a type and $E$ be an elementary type,
  then the extension $x|| E$ of $ x$ by the elementary type $E$ is defined recursively as follows: 
\begin{itemize}
\item If  $A$ is an elementary types $A||E:=AE$,
\item  for any $x$,$y$ types, $(x\rightarrow y)||E:=(x\rightarrow y||E)$.
\end{itemize}
\end{definition}
We are now in the position to define 
a notion of extension for generalized maps:
\begin{definition}[Extended map]
  Let $x$ be a non-elementary type, $E$ an
  elementary type and
  $M\in \set{T}_{\mathbb{R}}(x) $. We denote with
  $M_E$ the extension of $M$ by $E$ which is
  defined recursively as follows: If $x$, $y$ are
  two types and
  $ M \in \set{T}_{\mathbb{R}}(x \rightarrow y)$
  then
  $M_E \in \set{T}_{\mathbb{R}}(x||E \rightarrow
  y||E)$ is the Choi operator of the map
  $\mathcal{M} \otimes \mathcal{I}_E :
  \set{T}_{\mathbb{R}}(x
  ||E)\rightarrow\set{T}_{\mathbb{R}}(y||E)$,
  where
  $\mathcal{I}_E
  :\mathcal{L}(\hilb{H}_E)\rightarrow\mathcal{L}(\hilb{H}_E)$
  is the identity map.
\end{definition}
The definition above enables us to formulate two set of admissibility conditions for maps.
The first stems directly  from the definition of quantum states as positive operators:
\begin{definition}[Admissible elementary map]
Let $A$ be an elementary type and $M \in \set{T}_{\mathbb{R}}(A)$. We say that:

\begin{itemize}
\item $M$ is a deterministic map if $M \geq 0$ and $\text{Tr}[M ] = 1$. $\set{T}_1 (A)$ denotes the set of deterministic maps of type $A$.

\item  $M$ is admissible if $M \geq 0$ and there exists $N \in \set{T}_1(A)$ such that $M\leq N$. $\set{T}(A)$ is the set of admissible maps of type $A$.
\end{itemize}

\end{definition}
The admissible elementary maps which are not deterministic are called \textit{probabilistic} 
elementary maps. Let us proceed showing the second set of  conditions for maps:
\begin{definition}[Admissible maps]\label{Sadmap}
 Let $x, y$ be types, $M\in \set{T}_{\mathbb{R}}(x \rightarrow y)$ be an map of type $x\rightarrow y$ and $M_E \in \set{T}_{\mathbb{R}}(x||E \rightarrow y||E)$ be the extension of $M$ by $E$. Let $\mathcal{M} :\set{T}_{\mathbb{R}}(x) \rightarrow \set{T}_{\mathbb{R}}(y)$ and $\mathcal{M} \otimes \mathcal{I}_E : \set{T}_{\mathbb{R}}(x ||E)\rightarrow\set{T}_{\mathbb{R}}(y||E)$ be the linear maps whose Choi operator are $M$ and $M_E$ respectively.
We say that $M$ is admissible if,
\begin{enumerate}
\item[(i)] for all elementary types $E$, the map $\mathcal{M} \otimes \mathcal{I}_E$ sends admissible maps of type $x|| E$ to admissible maps of type $y|| E$,

\item[(ii)] there exist $\{N_i\}^n_{i=1} \subseteq \set{T}_{\mathbb{R}}(x \rightarrow y)$, $0 \leq n < \infty$ such that, for all elementary types $E$,
\begin{itemize}
\item $\forall 1\leq i \leq n$ the map $N_i$ satisfies item $\textit{(i)}$.
\item For all elementary types $E$, the map $(\mathcal{M} + \sum^n_{i=1} \mathcal{N}_i) \otimes \mathcal{I}_E$ maps deterministic maps of
  type $x||E$ to deterministic maps of type $y||E$.
\end{itemize}
\end{enumerate}
The set of admissible maps of type $x $ is
denoted with $\set{T}(x)$. We say that an
operator $D \in \set{T}_{\mathbb{R}}(x \rightarrow
y)$ is a \emph{deterministic} map of type $x \rightarrow
y$, if $D \in \set{T}(x \rightarrow y)$ and
$(\mathcal{D} \otimes \mathcal{I}_E )$ maps
deterministic admissible maps of type $x || E$ to
deterministic admissible maps of type $y || E$.
 The set of deterministic map of type $x$ is
  denoted by $\Evd{x}$. 
\end{definition}
We can say that Definition~\ref{Sadmap} generalises the Kraus' axiomatic definition of 
quantum operations to higher-order  maps. If we consider
the simplest case $x=A\rightarrow B$,
definition~\ref{Sadmap} we have the usual notion
of completely positivity trace non-increasing map
from $\mathcal{L}(\hilb H_A)$ to
$\mathcal{L}(\hilb H_B)$.

Definition~\ref{Sadmap} gives the characterization
of admissible maps in terms of physically
motivated assumptions.
The following theorem provides the mathematical
characterization of the set of admissible maps.
\begin{proposition}
  \label{Sthm:charthm}
  Let $x$ be a type and $M \in
  \mathcal{L}(\hilb{H}_x)$ a map of type $x$. Let
  $\set{Hrm}(\hilb{H})$ and  $\set{Trl}(\hilb{H})$
  denotes  the subspace of
  Hermitian operator and traceless Hermitian
  operators, respectively.
  Then
  $M \in \Eva{x}$ if and only if $M \geq 0$
    and $M \leq D$ for a map $D \in \Evd{x}$.
A map $D$ is a  deterministic map of type $x$ ($D
\in  \Evd{x}$)   if and
  only if    $D \geq 0$ and
  \begin{align}
    \label{Seq:2}
D = \lambda_x I_x + X_x, \quad \lambda_x \geq 0,
    \quad X_x \in \set{\Delta}_x \subseteq \set{Trl}(\hilb{H}_x)
  \end{align}
 where 
  $\lambda_x $ and
  $\set{\Delta}_x$ are
  defined recursively as
  \begin{align}
    \label{Seq:5}
    \begin{split}
     & \lambda_E = \frac{1}{d_E}, \;
      \lambda_{x\to y} = \frac{\lambda_y}{d_x
      \lambda_x}, \qquad \set{\Delta}_E =
    \set{Trl}(\hilb{H}_E), \\
       & 
       \set{\Delta}_{x \to y} =
[\set{Hrm}(\hilb{H}_x) \otimes {\set{\Delta}_{y}}]
    \oplus
[\overline{\set{\Delta}}_{x}   \otimes \set{\Delta}^\perp_{y} ],  
    \end{split}
  \end{align}
  and where $\set{\Delta}^\perp$,
  denote the orthogonal
   complement (with respect the Hilbert-Schmidt
   inner product) of $\Delta$ in
   $\set{Hrm}(\hilb{H})$ while
   $\overline{\set{\Delta}}$ is
    the orthogonal
   complement in
   $\set{Trl}(\hilb{H})$.
 \end{proposition}
 \begin{proof}
  See Ref. \cite{doi:10.1098/rspa.2018.0706}.
\end{proof}
This result is the most important tool in the
study of higher-order quantum maps.  It shows that
the convex set $\Eva{x}$ of the probabilistic maps of type
$x$ is determined by the set of the convex set
$\Evd{x}$
of deterministic maps which is in turn
determined by a normalization factor and a linear
constraint over the cone of positive operators.
Moreover, this result shows that the whole
axiomatic construction is consistent with all the
known examples of higher order maps considered in
the literature.

Let us introduce a
notion of inclusion between types.

\begin{definition}[type inclusion]
 If $\Eva{x}
  \subseteq \Eva{y}$
   we say that $x$ is
\emph{included} in $y$ and we write $x \subseteq
y$. If both $x \subseteq y$ and $y \subseteq x$
we say that $x$ is equivalent to $y$ and we write
$x \equiv y$.
\end{definition}
Thanks to Proposition \ref{Sthm:charthm} one can
prove several equivalences between types.
\begin{proposition}
  \label{Sprop:typeequivalences}
  Let $x$ and $y$ be types  and let us define
  \begin{align}
    \label{Seq:51}
    \overline{x} :=  x \to I, \\
    \label{Seq:52}
    x \otimes y := \overline{x \to \overline{y}}.
  \end{align}
  Then the following type equivalences holds
  \begin{align}
    \label{Seq:53}
    \overline{\overline{x} }:=  x \\
    A \otimes B \equiv AB \label{Seq:54}\\
    x \otimes y \equiv y \otimes x \label{Seq:55}\\
    (x \otimes y) \otimes z \equiv x \otimes (y \otimes z)\label{Seq:56}.
  \end{align}
\end{proposition}
\begin{proof}
    See Ref. \cite{doi:10.1098/rspa.2018.0706}.
\end{proof}
Equation \eqref{Seq:52} defines a tensor product
between types.  We refer to this operation as a
tensor product because it coincides with the usual
tensor product for elementary types and because it
is commutative and associative.  However, it is
worth stressing that the determinstic higher order
maps of type $x \otimes y$ are not the tensor
product of the deterministic maps of type $x$ with
determinstic maps of type $y$.  The set
$\Evd{x \otimes y}$ of deterministic maps of type
$x\otimes y$ is given by the intersection of the
cone of positive operators with the \emph{affine}
hull of the set $\{ \Evd{x} \otimes \Evd{y} \}$. This set is larger the intersection of the cone of positive operators with the \emph{convex}
hull of the set $\{ \Evd{x} \otimes \Evd{y} \}$.
For example the deterministic maps of type
$(A \to B) \otimes (C \to D)$ is the set of
\emph{non-signalling} channels and not every
non-signalling channel is a convex combination of
tensor products of channels.

\section{Combinatorial strucure of higher order
  maps}
\label{Ssec:comb-struc-high}
In this section we will develop a combinatorial
formalism that turns out to be very convenient for
our purposes. Moreover, this formal language
will make evident the algorithmic structure of
many aspects of higher order quantum theory.

Given a finite dimensional Hilbert space
$\hilb{H}$, the linear space of Hermitian
operators $\set{Hrm}(\hilb{H})$ can be expressed
as direct sum of the subspace
$\set{Trl}(\hilb{H})$ of traceless operators and
the one dimensional subspace generated by the
identity operator as follows:
\begin{align}
\label{SI+T}
  \set{Hrm}(\hilb{H})=\set{L}_0\oplus \set{L}_1,
\end{align}
where
$\set{L}_0:=\{X|\set{Tr}(X)=0,
X^\dagger=X\}=\set{Trl}(\hilb{H})$ and
$\set{L}_1:=\set{span}\{I\}$.  This allows us to
decomposed any element $O$ in
$\set{Hrm}(\hilb{H})$ as $O=\lambda I+ X$, with
$\lambda \in \mathbb{R}$ and $X\in \set{L}_0$.
Let us consider the case where our Hilbert space
can be factorized in a tensor product of 
Hilbert spaces labeled by some index set $\set{A}
:= \{A_1, \dots, A_n\}$ as follows:
\begin{equation}
  \hilb{H}_{\set{A}}=\hilb{H}_{A1}\otimes \hilb{H}_{A2}
  \otimes \cdots\otimes \hilb{H}_{An}.
\end{equation}
Then, according to the previous decomposition of
$\set{Hrm}(\hilb{H})$, it is natural to define
string of $l$ bits
\begin{align}
  {\textbf{b}} := b_{A1}b_{A2} \dots b_{An}, \quad b_{Ai} \in \{0,1\}
\end{align}
and the corresponding linear subspaces
\begin{align}
 \set{L}_{\textbf{b}}:=\set{L}_{b_{A1}}\otimes
  \set{L}_{b_{A2}}
  \otimes\cdots\otimes \set{L}_{b_{An}}
\end{align}
The spaces
$\set{L}_\textbf{b}$ have the following
property:
\begin{lemma}\label{Slmm:stringdefinedlabels}
Let $\textbf{b}$ be a binary string of lenght $n$
  labeled by the set $\set{A}
:= \{A_1, \dots, A_n\}$.
  Let $\set{L}_\textbf{b}$ be the corresponding
  subset of
  $\hilb{H}= \bigotimes_{i =1 }^{n}
  \hilb{H}_{A_i}$.
 If
  $\textbf{b}\neq\textbf{b}'$ then
  $\set{L}_\textbf{b}$ and $\set{L}_{\textbf{b}'}$
  are orthogonal subspaces with respect the
  Hilbert-Schmidt product .
\end{lemma}
\begin{proof}
  See Ref. \cite{doi:10.1098/rspa.2018.0706}
\end{proof}
Thanks to the Hilbert-Schmidt orthogonality of
$\set{L}_\textbf{b}$ spaces, we can write the sum
of two different spaces
$\set{L}_\textbf{b}+\set{L}_{\textbf{b}'}$, as the
direct sum
$\set{L}_\textbf{b}\oplus\set{L}_{\textbf{b}'}$.
Let us now introduce some notation:
\begin{align}
  &W_{\set{A}}:= \mbox{ the set of all binary
    strings labelled by } \set{A}, \label{SW} \\ 
  &T_{\set{A}}:=W_{\set{A}}\setminus
    \{\textbf{e}_{\set{A}} \}, \quad \textbf{e}_{\set{A}}:=1_{A1}1_{A2}\dots 1_{An}. \label{ST}
\end{align}
Moreover, we will denote with ${\varepsilon}$ the null string
and $\emptyset$ the empty set which contains no strings.

Each set of strings $J \subseteq W_{\set{A}}$
corresponds to a subspace of $\set{Hrm}(\hilb{H}_{\set{A}})$
as follows:
\begin{align}
&\set{L}_{J}:=\underset{\textbf{b}\in
                J}{\bigoplus}\set{L}_{\textbf{b}},
 &\set{L}_{{\varepsilon}}=\mathbb{R}, \quad \set{L}_{\emptyset}=\{0\},
\end{align}
where we stress the difference between the null
string $\varepsilon$ and the empty set $\emptyset$.
Clearly, we have $\set{L}_{W _{\set{A}}}=
\set{Hrm}(\hilb{H}_{\set{A}})$ and 
$\set{L}_{T _{\set{A}}}=
\set{Trl}(\hilb{H}_{\set{A}})$ 

For any subspace $\set{L}_J\subseteq
\set{Trl}(\hilb{H})$ 
we define
\begin{align}
  \label{Seq:4}
\overline{\set{L}_J} &:=   \set{L}_{\overline{J}},
                       \qquad \overline{J}  := T \setminus  J  \\
  \label{Seq:8}
  {\set{L}_J}^\perp &:=   \set{L}_{{J}^\perp},
  \qquad {J}^\perp  := W \setminus  J,
\end{align}
where we omitted the label ${\set{A}}$ from the
string sets $W$ and $T$ in order to lighten the notation.

It is worth notice that, whenever one of the
factor in the decomposition
$\hilb{H}_{\set{A}}=\hilb{H}_{A1}\otimes
\hilb{H}_{A2}\otimes\cdots\otimes \hilb{H}_{An}$
is one dimensional, e.g. $\hilb{H}_{Ak} = \mathbb{C}$
for some $A_k$,
the non trivial spaces $\set{L}_{\textbf{b}}$ are those determined only by the bits $b_{Aj}$ 
with $j\neq k$.
Indeed, for $b_{Ak}=0$ the associated space is 
$\set{L}_{b_{Ak}}=\{0\}$ implying  $\set{L}_{\textbf{b}}=\{0\}$, while for $b_{Ak}=1$ we have 
$\set{L}_{b_{Ak}}=\mathbb{R}$, then
$\set{L}_{\textbf{b}}=\set{L}_{\textbf{b}'_{{ Ak}}}$, where $\textbf{b}'_{Ak}$ is obtained from the string $\textbf{b}$ 
by dropping the $k$-th bit.

We now introduce some
operations we can perform on strings of bits that
will be useful for our purposes.

\begin{definition}[Concatenation]
  Let $\textbf{b} = b_{A1}b_{A2} \dots b_{An}$ and
  $\textbf{b}' = b'_{A'1}b'_{A'2} \dots b'_{A'm}$ be two
  labeled strings of bits. The
  \emph{concatenation} of
  $\textbf{b} $ with
  $\textbf{b}' $
  is denoted as $\textbf{b} \textbf{b}' $
  and reads $\textbf{b} \textbf{b}'  :=
  b_{A1}b_{A2} \dots b_{An}b'_{A'1}b'_{A2} \dots b'_{Am}
  $.

  If $J\subseteq W _{\set{A}}$, $J'\subseteq W _{\set{A}'}$
  are set of strings we define the
  \emph{concatenation} of
  $J $ with
  $J' $ as as follows:
  \begin{align}
    \label{Seq:10}
   & JJ' :=\{\textbf{b}=\textbf{w}\textbf{w}' \,
     | \, \textbf{w}\in J, \textbf{w}' \in J'\}.
\end{align}
\end{definition}

We notice that the concatenation of sets of string 
is basically the cartesian product of the sets.
For example, we have, for any $ J \subseteq W _{\set{A}} $
\begin{align}
  \begin{aligned}
  {\varepsilon} J=J{\varepsilon}=J, \quad
  \emptyset  J=J{\emptyset}=\emptyset.
  \end{aligned}
\end{align}
In terms of the linear spaces
$\set{L}_{\mathbf{b}}$ and $\set{L}_{J}$  the
concatenation translates as a tensor product, e.g.
$\set{L}_{JJ'} = \set{L}_{J}\otimes \set{L}_{J'} $.

\begin{definition}[Contraction]\label{Sstringcontr}
Let  $\textbf{b}=b_{A1} b_{A2}\dots b_{An}$  be a
labeled string of $n$ bits.
For any pair $(A_i,A_j)$, the
\emph{$(A_i,A_j)$-contraction} of $\textbf{b}$  i
defined as:
\begin{align}
 & \begin{aligned}
  \mbox{if }  b_{Ai}=b_{Aj} \mbox{ then }
  \mathcal{C}_{Ai,Aj}(\textbf{b})=& \,\,
  b_{A1}\dots b_{A(i-1)} b_{A(i+1)} \dots\\
  \dots b&_{A(j-1)}  b_{A(j+1)} \dots b_n  
\end{aligned} \nonumber \\
&   \begin{aligned}
  \mbox{if }  b_{Ai}\neq b_{Aj} \mbox{ then }
  \mathcal{C}_{Ai,Aj}(\textbf{b})= {\varepsilon}
\end{aligned} \label{Seq:6}
\end{align}
% \begin{align}
% \mathcal{C}_{ij}(\textbf{b})=
% \begin{cases}
% b_1,\cdots,b_{i-1},b_{i+1},\cdots,b_{j-1},b_{j+1},\cdots,b_n\;\textit{if}\; b_i=b_j\\
% \bm{\varepsilon}\;\textit{if}\;b_i\neq b_j
% \end{cases}    
% \end{align}
If  $S$ is a set of
strings, then the $(A_i,A_j)$-contraction of $S$ is
defined as follows: $\mathcal{C}_{Ai,Aj}(S):=\{\mathcal{C}_{Ai,Aj}(\textbf{b}), \;\textbf{b}\in S \}$
\end{definition}
Let us work out an explicit example of
contraction:
\begin{align*}
  \begin{aligned}
  S := & \Big \{0_A1_B 0_C 1_D, \, 0_A0_B 0_C 1_D, \,1_A1_B
  0_C 1_D  \Big \} ;\\
  \mathcal{C}_{A,D}(S) = &                  
                      \Big \{
                            \mathcal{C}_{A,D}(0_A1_B
                            0_C 1_D) , \, 
                            \mathcal{C}_{A,D}(0_A0_B
                            0_C 1_D),  \\
  & \;\;   \mathcal{C}_{A,D}(1_A1_B 0_C 1_D)
  \Big \} =
                       \\
=  &\Big \{   1_B 0_C \Big \} .   
  \end{aligned}
\end{align*}

\begin{definition}[Composition]
  Let $\set{A} := \{A_1, \dots, A_n\}$ be a set of
  $n$ indexes
  and $\set{A'} := \{A'_1, \dots, A'_m\}$ be a set
  of $m$ indexes.
Let $\textbf{b}=b_{A1}\dots  b_{An} \in W _{\set{A}}$ and
$\textbf{b}'=b'_{A'1},\dots ,b'_{A'm} \in W _{\set{A}'}$
be two strings.
Let $\set{H} \subseteq \set{A} \times \set{A}'$
be a set of mutually disjoint couples $(A_i,
A'_j) \in \set{A} \times \set{A}'$, i.e.
if  $(A_i,
A'_j) , (A_k,
A'_l) \in \set{H}$ then $i \neq k$ and $j \neq l$.
The \emph{composition of $\textbf{b}$ and
  $\textbf{b}'$ over $\set{H}$}, which we denote as
$\textbf{b}\ast_{\set{H}} \textbf{b}'$ is defined
as follows:
\begin{align}
\textbf{b}\ast_{\set{H}}
  \textbf{b}':=\mathcal{C}_{\set
  H}(\textbf{b}\textbf{b}')
  :=\mathcal{C}_{Ai,A'j}(\cdots \mathcal{C}_{Ak,
  A'l}
  (\textbf{b}\textbf{b}')\cdots).   \label{Seq:compositionstring}
\end{align}
If $J$ and $J'$ are two set of strings which are
indexed by $\set{A}$ and $\set{A}'$ respectively,
the composition of $J$ and $J'$ over $\set{H}$
is defined as
\begin{align}
  \label{Seq:9}
  J \ast_{\set{H}} J' :=
  \{
  \textbf{b}\ast_{\set{H}} \textbf{b}' \, | \, \textbf{b}\in J, \textbf{b}' \in J' \}
\end{align}

\end{definition}
One can notice that the order in which the
contractions of Equation
\eqref{Seq:compositionstring} are carried is
immaterial. For sake of clarity, let us work out
an explicit example of string composition:

\begin{align*}
  \begin{aligned}
    \mathbf{b} &:=  0_A1_B 0_C 1_D \quad
    \mathbf{b}' :=  0_{A'}1_{B'} 0_{C'} 1_{D'} 
    \\
    \set{H} &:= \Big \{  (A,A') , \, (B,B')   \Big \}\\ 
    \mathbf{b} \ast_{\set{H}} \mathbf{b}'
= &
    \mathcal{C}_{\set{H}}( \mathbf{b} \mathbf{b}')
    = 
     \mathcal{C}_{\set{H}}(0_A1_B 0_C 1_D
     0_{A'}1_{B'} 0_{C'} 1_{D'} ) = \\
     = &\mathcal{C}_{B,B'}\mathcal{C}_{A,A'}(0_A1_B 0_C 1_D
     0_{A'}1_{B'} 0_{C'} 1_{D'} )  =\\
          =&\mathcal{C}_{B,B'}(1_B 0_C 1_D
     1_{B'} 0_{C'} 1_{D'} )  =
0_C 1_D 0_{C'} 1_{D'}
\end{aligned}
\end{align*}

\section{Types inclusion and equivalences}

The aim of this section is to prove Proposition 2
of the main text.
In doing do, we will use the combinatorial
formalism of the preceding section.
We will need to label binary string with the
\emph{non-trivial} elementary types occurring in
the expression of a type $x$.
However, this is a potential
  source of ambiguity because, according to
  Definition~\ref{Sdef:type}, the same elementary
  type can occur more than once in the expression
  of a given type, e.g. $(A \to B) \to A$ (that
  means that that two copies of system $A$ are
  involved).  Having the same label
  repeated twice is problematic when we would like
  to use the elementary systems to label the bits
  of a string as we do in Lemma
  \ref{Slmm:stringdefinedlabels}.  It is therefore
  convenient to relabel the elementary types in
  the expression of a type $x$ in such a way that
  no repetition occur. For example the type
  $(A \to B) \to A$ should be rewritten as
  $(A \to B) \to C$ where we now assigned the
  label $C$ to a copy of system $A$ (system $C$
  will be isomorphic to $A$).  From now on, we
  will assume that the same label of \emph{non-trivial}
  elementary type cannot occur more than once in
  the expression of a given type.
Since multiple occurrence of the \emph{trivial} type
$I$,
is not problematic, we will avoid to introduce
multiple lables for isomorphic trivial systems,
e.g. we will not turn expressions like $((A \to B) \to I)
\to (C \to I)$ into something like
$((A \to B) \to I_1)
\to (C \to I_2)$.

From now on, we will take for granted such a  relabeling.

Let us now reformulate the
characterization theorem of  Proposition \ref{Sthm:charthm}
in this language.

\begin{lemma} \label{Slmm:charactstring}
  Let $x$ be a type and let us denote  by
  $\set{Ele}_x = \{ A_1 ,\dots , A_n \}$
the set of \emph{non-trivial} elementary types $A_i$
occurring in the expression of a type $x$. Then,
the linear subspace
$\Delta_x$ defined in Equation \eqref{Seq:5} satisfies
\begin{equation}\label{SD}
 \Delta_x=\underset{\textbf{b}\in D_x}{\bigotimes}   \set L_{\textbf{b}}
\end{equation}
for a set $D_x$ of string which are labelled by
the set $\set{Ele}_x $.

The set $D_x$ is defined
recursively as follows:
\begin{align}
  \label{Seq:setofstrings20}
  \begin{split}
    &D_A=\{0\} \quad
   \forall A \in  \set{EleTypes}, A \neq I \\
   &D_I=\emptyset,\qquad
   D_I^\perp=\{{\varepsilon}\}, \quad 
   \\
    &D_{(x\rightarrow y)}=W_xD_y\cup \overline{D}_xD_y^\perp,
  \end{split}
\end{align}
where Equations \eqref{Seq:4}~\eqref{Seq:8} and
~\eqref{Seq:10}
are understood.
Moreover we have
\begin{align}
  &\lambda_x=\underset{A_i\in \set{Ele}_x}
    {\prod}d_{{Ai}^{-K_x(Ai)}}\label{Seq:16}\\
    &K_x(A_i):=\#[\text{''} \rightarrow \text{''}]+\#[\text{''} ( \text{''}] \;(\text{mod $2$})
\end{align}
where $\#[\text{''} \rightarrow \text{''}]$ and $\#[\text{''} ( \text{''}]$ denotes the number of arrows $\rightarrow$ and left round brackets $($ to the right of $A_i$ in the expression of $x$, respectively.
\end{lemma}
\begin{proof}
  See Ref. \cite{doi:10.1098/rspa.2018.0706}
\end{proof}
This result disply the combinatorial structure of
the linear constraints, given by Equations
\eqref{Seq:2} and \eqref{Seq:5}, which characterise
the hierarchy of higher order maps.
In particular, for the  types of the kind $x \to I $ and $x
\otimes y$ we have:
\begin{align}\label{eq:4}
  \begin{aligned}
  D_{\overline{x}} &= \overline{D_x} \\
  D_{x\otimes y} &=  \textbf{e}_x D_y \cup
  D_x\textbf{e}_y  \cup D_xD_y   
  \end{aligned}
\end{align}

We now provide a
definition of a set of operators based on a
generic set of strings, i.e. without referring
to a specified type:
\begin{definition}\label{Sopset}
Let $\set H $ be a set of non-trivial
elementary types , $S$ be a set of
strings of $\set H$ s.t. $\textbf{e}\notin S$ and
let $\lambda_{\set H}\in \mathbb{R}$ be a real number. 
Then we define the set
\begin{align}
  \label{Seq:7}
  \begin{split}
       \mathscr{M}(\lambda_\set{H},S):=\{&R\in
  \mathcal{L}(\hilb H_{\set{H}})| R\geq 0 \\
  &\mbox{ and } R=\lambda_\set{H}I_\set{H}+T, T\in\set L_{S}\} 
  \end{split}
\end{align}
\end{definition}

A given set $ \mathscr{M}(\lambda_\set{H},S)$ can
be regarded  as a set of admissible higher order
maps if the condition of the following lemma applies.

\begin{lemma}\label{Spartialorderingopset}
Given a set of operators $\mathscr{M}(\lambda_\set{H},S)$ as in definition~\ref{Sopset} and 
a type $x$, we have that 
\begin{equation*}
    \mathscr{M}(\lambda_\set{H},S)\subseteq \set T_1(x)\iff 
    \begin{cases}
    \set H=\set{Ele}_x\\
    \lambda_\set{H}=\lambda_x\\
    S\subseteq D_x
    \end{cases} 
\end{equation*}
\end{lemma}
\begin{proof}
$(\implies).$ Given that $R\in  \mathscr{M}(\lambda_\set{H},S)\implies R=\lambda_\set{H}I_\set{H}+T$ and $R\in \set T_1(x) \implies R=\lambda_x I_x+T$
then $\mathscr{M}(\lambda_\set{H},S)\subseteq \set T_1(x)$ implies that $\lambda_\set{H}=\lambda_x$, $\set{Ele}_x=\set H$ and $L_S\subseteq L_{D_x}$, namely
$S\subseteq D_x$.\\
$(\impliedby)$. It is proven analogously.
\end{proof}
We now can prove the following result
\begin{lemma}[Partial ordering of types]
  \label{Slmm:typesinclusionsstringslambda}
Given two types $x$, $y$ and $D_x$, $D_y$ the corresponding set of strings, then we have:
\begin{align}
  \label{Seq:13}
    x\subseteq y\iff 
    \begin{cases}
      \set{Ele}_x=\set{Ele}_y\\
      \lambda_x=\lambda_y\\
      D_x\subseteq D_y
    \end{cases}
\end{align}
\end{lemma}
\begin{proof}
Let us start with the necessary ($\implies$) condition. Given the definition~\ref{Spartialorderingopset}, condition 
$\set{Ele}_x=\set{Ele}_y$ is trivially satisfied and we have $R\in \set T_1(x)\implies R\in\set T_1(y)$.
Hence, given that $R=\lambda_xI_x+T$,  $\lambda_x=\lambda_y$ is understood. Let us suppose that
$D_x\nsubseteq D_y$, then $\exists \textbf{b}\in D_x$ such that $\textbf{b}\notin 
D_y$. So if we take $T\in \set T_{\textbf{b}}$, then  $T\notin \set 
L_{D_y}$ and defining $\tilde{R}=\lambda_x I_x+\epsilon T\geq 0$, with $\epsilon \in 
\mathbb{R}$ arbitrary small in order to have $\tilde{R}\in \set T_1(x)$, we obtain 
that $\tilde{R}\notin \set T_1(y)$ which contradicts the hypothesis.
The inverse implication is trivially implied by
Lemma \ref{Spartialorderingopset}
\end{proof}

Finally, we can prove the type inclusion stated in
Proposition 2 of the main text.

\begin{proposition}\label{Spropositionchannel}
  Let $x$ be a type and let us define the sets
  \begin{align}
    \label{Seq:26}
    \begin{aligned}
      \set {in}_x&:=\{A\in\set{Ele}_x \;\textit{s.t.}\;
    K_x(A)=1\} \\ 
    \set{out}_x&:=\set{Ele}_x\setminus
\set {in}_x.  
    \end{aligned}
  \end{align}
We will call $ \set {in}_x$ the set of
\emph{input}
systems of $x$ and
$ \set {out}_x$ the set of
\emph{output}
systems of $x$.
Then the following inclusion relation
holds
\begin{equation}\label{Schannel}
    \overline{\set{out}_x\rightarrow
      \set{in}_x}\subseteq  x \subseteq \set{in}_x\rightarrow \set{out}_x
\end{equation}
\end{proposition}
\begin{proof}
From Lemma \ref{Slmm:typesinclusionsstringslambda}
we need to prove that
\begin{align}
  \label{Seq:14}
\set{Ele}_{\overline{\set{out}_x\rightarrow
  \set{in}_x}}
=&
 \set{Ele}_ {x}
=
  \set{Ele}_{\set{in}_x\rightarrow \set{out}_x  }
  \\
  \label{Seq:15}
  \lambda_{\overline{\set{out}_x\rightarrow
  \set{in}_x}}
=&
\lambda_ {x}
=
\lambda_{\set{in}_x\rightarrow \set{out}_x  }
  \\
  D_{ \overline{\set{out}_x\rightarrow
  \set{in}_x}}
  \subseteq &
  D_x
  \subseteq
  D_{\set{in}_x\rightarrow \set{out}_x},
  \label{Sstchannel}
\end{align}
Equation \eqref{Seq:14} is trivially satisfied and
Equation \eqref{Seq:15} follows from Equation 
\eqref{Seq:16} by direct computation

  We will now prove Equation \eqref{Sstchannel} by induction.
  First let us consider the case in which 
  $x=A$ is an elementary type.
  By exploiting the type equivalence
  $x \equiv I \to x$
we can write the elementary type $A$ as $I \to A$.
Then we have $D_A=\{0\}$, $\set{out}_x=A$, and
$\set{in}_x=I$. By direct computation we have: 
\begin{align*}
    &\set{in}_x\rightarrow \set{out}_x=I\rightarrow A=A,\\
    &\overline{\set{out}_x\rightarrow \set{in}_x}=\overline{A\rightarrow I}=\overline{\overline{A}}=A,
\end{align*}
which clearly satisfies the thesis since
$A\subseteq A\subseteq A$.

Let us now suppose that the thesis holds for the
types $x$, $y$. we will prove that the thesis
holds for    $x\rightarrow y$.
By induction hypothesis, Equation
\eqref{Sstchannel} holds for $x$ and $y$. 
By applying Lemma \ref{Slmm:charactstring} and Equation
the terms
on the right side of Equation
\eqref{Sstchannel} 
become
\begin{align*}
    &D_{\set{in}_x\rightarrow \set{out}_x}=W_{\set{in}_x}T_{\set{out}_x}\cup\overline{T}_{\set{in}_x}T_{\set{out}_x}^\perp=W_{\set{in}_x}T_{\set{out}_x}\\
    &D_{\set{in}_y\rightarrow \set{out}_y}=W_{\set{in}_y}T_{\set{out}_y}
\end{align*}
since
$D_{\set{in}/\set{out}}=T_{\set{in}/\set{out}}$
and
$\overline{T}_{\set{in}/\set{out}}=\emptyset$. By
applying Lemma \ref{Slmm:charactstring} and Equation~\eqref{eq:4}
on  the left side of~\eqref{Sstchannel} we obtain:
\begin{align*}
  &\begin{aligned}
  D_{ \overline{\set{out}_x\rightarrow
  \set{in}_x}}&
                =\overline{D}_{\set{out}_x\rightarrow
                \set{in}_x}
                =\overline{W_{\set{out}_x}T_{\set{in}_x}}\\
                &=(W_{\set{out}_x}T_{\set{in}_x})^\perp
                \setminus \textbf{e}_1\textbf{e}_2\\
              &=(\underbrace{W^\perp_{\set{out}_x}}_{\emptyset}W_{\set{in}_x}\cup
                W_{\set{out}_x}\textbf{e}_{\set{in}_x}
                )\setminus\textbf{e}_{\set{out}_x}\textbf{e}_{\set{in}_x}
  \\
              &=W_{\set{out}_x}\textbf{e}_{\set{in}_x}\setminus\textbf{e}_{\set{out}_x}\textbf{e}_{\set{in}_x}=T_{\set{out}_x}\textbf{e}_{\set{in}_x}   
  \end{aligned} \\
  &
    \begin{aligned}
      D_{
  \overline{\set{out}_y\rightarrow\set{in}_y}}=T_{\set{out}_y}\textbf{e}_{\set{in}_y}.
    \end{aligned}
\end{align*}
The
condition in~\eqref{Sstchannel} then becomes
\begin{equation}
  \label{Seq:11}
    \textbf{e}_{\set{in}_x}T_{\set{out}_x}\subseteq D_x\subseteq W_{\set{in}_x}T_{\set{out}_x},
\end{equation}
%Let us now transform the latter relation in accordance with the allowed operations 
%on the set of strings presented in the previous section:
which also implies
\begin{equation}\label{Sbarstchannel}
   \overline{ W_{\set{in}_x}T_{\set{out}_x}}\subseteq \overline{D}_x\subseteq \overline{\textbf{e}_{\set{in}_x}T_{\set{out}_x}}.
\end{equation}
By direct computation we have:
\begin{align*}
  &\begin{aligned}
 \overline{\textbf{e}_{\set{in}_x}T_{\set{out}_x}}&=(\textbf{e}_{\set{in}_x}T_{\set{out}_x})^\perp\setminus\textbf{e}_{\set{in}_x}\textbf{e}_{\set{out}_x}  \\
 &=(T_{\set{in}_x}W_{\set{out}_x}\cup W_{\set{in}_x}\textbf{e}_{\set{out}_x})\setminus\textbf{e}_{\set{in}_x}\textbf{e}_{\set{out}_x}\\
 &=T_{\set{in}_x}W_{\set{out}_x}.   
  \end{aligned}\\
   &\begin{aligned}
     \overline{ W_{\set{in}_x}T_{\set{out}_x}}  = T_{\set{in}_x}\textbf{e}_{\set{out}_x}
   \end{aligned}
\end{align*}
By substituting  these terms in Equation~\eqref{Sbarstchannel} we obtain
\begin{equation}
\label{Seq:12}
  T_{\set{in}_x}\textbf{e}_{\set{out}_x}\subseteq D_{\overline{x}} \subseteq T_{\set{in}_x}W_{\set{out}_x}. 
\end{equation}
Now we have to check whether the thesis  holds for the type $x\rightarrow y$, that is
\begin{equation}
    D_{\overline{\set{out}_{x\rightarrow y}\rightarrow\set{in}_{x\rightarrow y}}} \subseteq D_{x\rightarrow y}\subseteq D_{\set{in}_{x\rightarrow y}\rightarrow\set{out}_{x\rightarrow y}}.
\end{equation}
By using Lemma~\ref{Slmm:charactstring} we have
\begin{align}
  \label{SLDxy}
  & \begin{aligned}
    D_{\overline{\set{out}_{x\rightarrow
          y}\rightarrow\set{in}_{x\rightarrow
          y}}}&=\textbf{e}_{\set{in}_{x\rightarrow
        y}}T_{\set{out}_{x\rightarrow y}} \\
    =\textbf{e}_{\set{out}_x\set{in}_y}&T_{\set{in}_x}W_{\set{out}_{y}}\cup\textbf{e}_{\set{out}_x\set{in}_y}W_{\set{in}_x}T_{\set{out}_{y}},
  \end{aligned}\\
  &
\label{SRDxy}
    \begin{aligned}
      D_{\set{in}_{x\rightarrow y}\rightarrow\set{out}_{x\rightarrow y}}&=W_{\set{out}_{x}\set{in}_y}T_{\set{in}_x\set{out}_y}\\
      =W_{\set{out}_{x}\set{in}_y} &T_{\set{in}_x} W_{\set{out}_{y}}\cup W_{\set{out}_{x}\set{in}_y}W_{\set{in}_x}T_{\set{out}_{y}}\\
      =W_{\set{out}_{x}}W_{\set{in}_y}&T_{\set{in}_x}W_{\set{out}_{y}}\cup W_{\set{out}_{x}}W_{\set{in}_y}W_{\set{in}_x}T_{\set{out}_{y}}\\
      =W_y W_{\set{out}_{x}}&T_{\set{in}_x}\cup
      W_x W_{\set{in}_y}T_{\set{out}_{y}}
    \end{aligned}\\
  & \begin{aligned} D_{x\rightarrow
      y}=W_{x}D_y\cup D_{\overline{x}}D_y^\perp.
        \end{aligned}
       \label{SDxy}
\end{align}

From Equation \eqref{Seq:11} and \eqref{Seq:12} we have
\begin{align*}
     &W_{x}D_y\subseteq W_x W_{\set{in}_y}T_{\set{out}_{y}}\\
    &D_{\overline{x}}D_y^\perp\subseteq W_y W_{\set{out}_{x}}T_{\set{in}_x}
\end{align*}
which, together with Equations \eqref{SDxy} and
~\eqref{SRDxy}, proves the inclusion $D_{x\rightarrow y}\subseteq
D_{\set{in}_{x\rightarrow
    y}\rightarrow\set{out}_{x\rightarrow y}}$ is
proved.
Similarly, focusing on the terms~\eqref{SDxy} and~\eqref{SLDxy}, we have
\begin{align*}
&  \begin{aligned}
\textbf{e}_{\set{out}_x\set{in}_y}T_{\set{in}_x}W_{\set{out}_{y}}&=\underbrace{\textbf{e}_{\set{out}_x}T_{\set{in}_x}\textbf{e}_{\set{in}_y}\textbf{e}_{\set{out}_{y}}}_{\subseteq D_{\overline{x}}D_y^\perp }\\
&\cup \underbrace{\textbf{e}_{\set{out}_x}T_{\set{in}_x}\textbf{e}_{\set{in}_y}T_{\set{out}_{y}}}_{\subseteq W_{x}D_y}\\
&\subseteq D_{x\rightarrow y},
  \end{aligned}\\
 & \begin{aligned}
\textbf{e}_{\set{out}_x\set{in}_y}W_{\set{in}_x}T_{\set{out}_{y}}=&\textbf{e}_{\set{out}_x}W_{\set{in}_x}\textbf{e}_{\set{in}_y}T_{\set{out}_{y}}\subseteq W_{x}D_y \\
&\subseteq    D_{x\rightarrow y}.
  \end{aligned}
\end{align*}
This 
 concludes the proof by induction.
\end{proof}
Proposition \ref{Spropositionchannel} shows that we
have two kind of non-trivial elementary types:
\emph{input} elementary types which belongs to $\set{in}_x$ and \emph{output}
elementary types which belongs to $\set{out}_x$.
This split of the set  $\set{Ele}_x$ is motivated
by the fact that each higher order map of type $x$
can always be used as a channel (or a quantum
operation) from the input
systems  of $\set{in}_x$ to the output systems of
$\set{out}_x$.

\section{Compositional structure of Higher order maps}
In the following section we provide a notion of composition for types and we study 
the relevant structure involved.
Let us consider two types $x$ and $y$.
As we did in the previous section we will assume
that the same non-trivial elemntary system does not appear twice in the
expression of the same type.
However, the same elementary system can appear
both in the expressions of  $x$ and $y$, i.e. the intersection
$\set{Ele}_x \cap \set{Ele}_y$ can be non empty.
Then,
to compose a map of type $x$ with a map of type
$y$, 
means to connect the systems of $x$ and $y$ which
have the same label.

\begin{definition}[Admissible type composition] \label{SadmT}
  Let  $x$,$y$ be two types  and let us define the
  set 
  $\set{H}:=\set{Ele}_x\cap\set{Ele}_y$.
We say that the composition $x\ast y$ is admissible if
\begin{align}
  \label{Seq:17}
    \forall R\in\set T_1(x),\; \forall S\in\set T_1(y)\quad \exists z \;\textit{s.t.}\; R\ast S\in \set T_1(z).
\end{align}
where
\begin{align}
  \label{Seq:21}
  R\ast S=\Tr _{\set{H}}
[(R\otimes I_{\set{Ele}_y\setminus \set H})
(S^{T_{\set{H}}}\otimes I_{\set{Ele}_x\setminus
  \set{H}})],
\end{align}
 $\Tr _{\set{H}} $ denotes the partial trace on
the Hilbert space $\hilb{H}_{\set{H}} :=
\bigotimes_{j \in \set{H}}\hilb{H}_{Aj}$
and $S^{T_{\set{H}}}$ the partial trasposition
with respect to the basis which has been choosed in the
definition of the Choi operator on
the space $\hilb{H}_{\set{H}} $.
\end{definition}
The operation defined in Equation \eqref{Seq:21}
is known as the \emph{link product}\cite{PhysRevA.80.022339}
of $R$ and $S$. It can be shown that the following
properties hold
\begin{align}
  \label{Seq:22}
  &R * S = S* R  \quad 
  R * (S*T) =  (R*S)*T \\
  &R \geq 0 , S \geq 0 \implies R*S \geq 0.
    \label{Seq:23}
\end{align}

One could be tempted to modify Equation
\eqref{Seq:17}
as follows:
\begin{align}
  \label{Seq:18}
    \forall R\in\set T(x),\; \forall S\in\set T(y)\quad \exists z \;\textit{s.t.}\; R\ast S\in \set T(z).
  \end{align}
However, Equation \eqref{Seq:18}  
would allow for some non-physical composition.
For instance,
let us consider the elementary  types $x={A}$ and
$y=AB$ and two arbitrary
the deterministic states $\rho \in  T_1(A)$ and
$\sigma \in  T_1(AB)$.
Since $\sigma * \rho \geq 0 $ and 
$\Tr[\sigma * \rho] \leq 0 $
we have that $\sigma * \rho \in T(A)$ for any 
$\rho \in  T_1(A)$ and
$\sigma \in  T_1(AB)$.
According to Equation \eqref{Seq:18}
that would mean that $A * AB$ is an admissible
type composition which is clearly not the case.

According to Equation  \eqref{Seq:17},
a composition is admissible if and only if it well
behaves on the set of \emph{deterministic} maps.
The following lemma proves that this condition
implies that the set of probabilistic map is also
preserved.
\begin{lemma}
  Let $x,y$ be two types such that the composition
  $x*y$ is admissible. Then we have
\begin{equation}
     \forall R\in\set T(x),\; \forall S\in\set T(y)\quad \exists z \;\textit{s.t.}\; R\ast S\in \set T(z),
\end{equation}
\end{lemma}
\begin{proof}
Let us consider $R\in \set T(x)$ and $S\in\set T(y)$ two generic
probabilistic maps. Then $\exists R'\in\set T_1(x), S'\in \set T_1(y)$ 
such that $R'\geq R$ and $S'\geq S$. Moreover, we can find $\widetilde{R}\in \set T(x)$, $\widetilde{S}\in \set T(y)$ which satisfies $R+\widetilde{R}=R'$ and $S+\widetilde{S}=S'$ respectively.
According to Equation~\eqref{Seq:17}, $\exists z$ such that $R'\ast S'\in \set T_1(z)$
\begin{equation*}
    R'\ast S'=R\ast S+R\ast\widetilde{S}+\widetilde{R}\ast S+\widetilde{R}\ast \widetilde{S}.
\end{equation*}
Therefor $R\ast S\leq R'\ast S'$ and $R\ast S\in \set T(z)$ follows. 
\end{proof}
We will now prove a collection of results which
provide a characterization of the admissible
compositions.
The first one shows that we can without loss of
generality assume that the type $z$ in Equation
\eqref{Seq:17} is of the kind $\set{in}_z \to \set{out}_z$.
\begin{lemma}
  Let $x$ and $y$  be two types. Then
  the composition $x\ast y$ is admissible
  if and only if for any $
     R\in\set T_1(x)$ and $ S\in\set
  T_1(y)$ there exist two disjoint set of
  non-trivial elementary types $\set{in}_z, \set{out}_z
  \subseteq (\set{Ele}_x \cup \set{Ele}_y) \setminus (\set{Ele}_x \cap \set{Ele}_y)$ ,
 $ \set{in}_z \cap \set{out}_z = \emptyset$
such that
$R\ast S\in \set T_1(\set{in}_z
  \to\set{out}_z )$.
\end{lemma}
\begin{proof}
  From Proposition \ref{Spropositionchannel} we
  know that $z \subseteq \set{in}_z \to
  \set{out}_z$.
  Therefore, we can replace $z$ with $\set{in}_z \to
  \set{out}_z$ in Equation \eqref{Seq:17}.
\end{proof}
The following lemma shows that a necessary
condition for a composition to be admissible
is that we must connect either input systems of $x$
with output systems of $y$
or  output systems of $x$
with input systems of $y$.
\begin{lemma}\label{Slmm:noinputinput}
   Let $x$ and $y$  be two types and let us denote
   with $\set{H} := \set{Ele}_x \cap \set{Ele}_y$
   the set of non-trivial elementary systems that
   $x$ and $y$ have in common.
   If  the composition $x\ast y$ is admissible
   then, $\set{H} \cap (\set{out}_x \cap
   \set{out}_y ) =\set{H} \cap (\set{in}_x \cap
   \set{in}_y ) = \emptyset$.
\end{lemma}
\begin{proof}
%   Let us introduce some notation:
%  \begin{align}
%    \label{Seq:29}
%    \begin{aligned}
%    &\widetilde{\set{in}}_x
%    := \set{in}_x \setminus\set{H},
%    &&
% \widetilde{\set{in}}'_x
%    := \set{in}_x \cap\set{H},
% \\
%  &   \widetilde{\set{out}}_x := \set{out}_x \setminus
%    \set{H}, &&
%     \widetilde{\set{out}}'_x := \set{out}_x \cap
%    \set{H},
%    \\
%   &  \widetilde{\set{in}}_y
%    := \set{in}_y \setminus\set{H},
% && \widetilde{\set{in}}'_y
%    := \set{in}_y \cap\set{H},
%    \\
%    & \widetilde{\set{out}}_y := \set{out}_y \setminus
%    \set{H}, &&
%     \widetilde{\set{out}}'_y := \set{out}_y \cap
%     \set{H}
%    \end{aligned}
% \end{align}
  First, we will prove that
  $\set{H} \cap (\set{out}_x \cap
  \set{out}_y ) = \emptyset $.
  By contradiction, let us assume that there exist $A
  \in \set{H} \cap (\set{out}_x \cap
  \set{out}_y ) $.
  Let us define $\set{out}'_x := \set{out}_x
  \setminus A$ and
  $\set{out}'_y := \set{out}_y
  \setminus A$ and let us consider
the higher order maps
  \begin{align}
    \label{Seq:32}
    \begin{aligned}
       R = I_{\set{in}_x} \otimes \ketbra{0}{0}_A
    \otimes \frac{1}{d_{\set{out}'_x}}
    I_{\set{out}'_x} \\
    S = I_{\set{in}_y} \otimes \ketbra{1}{1}_A
    \otimes \frac{1}{d_{\set{out}'_y}}
    I_{\set{out}'_y} .
    \end{aligned}  
  \end{align}
  where $I_{\set{J}}$ denotes the identity
  operator
  on $\hilb{H}_{\set{J}}$ and $\ket{0} ,\ket{1}$
  are two orthonormal states of system $A$.
  It is straightforward to verify that
  $R \in \Evd{ \overline{\set{out}_x\to
      \set{in}_x} }\subseteq  \Evd{x} $
  and
  $R \in \Evd{ \overline{\set{out}_y\to
      \set{in}_y} }\subseteq  \Evd{y} $.
  Since $x*y $ is admissible there must exist a type
  $z$ such that $R*S \in \Evd{z} $.
  On the other hand, from a straightforward computation 
  we have that $R*S = 0$.

  Let us now prove
  $\set{H} \cap (\set{in}_x \cap
  \set{in}_y ) = \emptyset $.
By contradiction, let us assume that
$\set{K} := \set{H} \cap (\set{in}_x \cap
\set{in}_y ) \neq \emptyset $.
 Let us define $\set{in}'_x := \set{in}_x
  \setminus \set{K}$ and
  $\set{in}'_y := \set{in}_y
  \setminus \set{K}$ and let us consider
the maps
  \begin{align}
    \label{Seq:32b}
   & \begin{aligned}
    &R = I_{\set{in}'_x} \otimes I_{\set{K}}
    \otimes \ketbra{0}{0}_{\set{out}_x} \\
   & S = I_{\set{in}'_y} \otimes I_{\set{K}}
    \otimes \ketbra{0}{0}_{\set{out}_y}     
    \end{aligned}
    \\
    \nonumber
    &\ketbra{0}{0}_{\set{A}} := \bigotimes_{i \in
    \set{A}} \ketbra{0}{0}_i.
  \end{align}
  It is straightforward to verify that
  $R \in \Evd{ \overline{\set{out}_x\to
      \set{in}_x} }\subseteq \Evd{x} $ and
  $R \in \Evd{ \overline{\set{out}_y\to
      \set{in}_y} }\subseteq \Evd{y} $.  Since
  $ x*y $ is admissible there must exist a type
  $z$ such that
  $R * S \in \Evd{z} \subseteq \Evd{\set{in}_z \to
    \set{out}_z } $.  Therefore, we must have
  \begin{align}
    \label{Seq:36}
  \Tr_{\set{out}_z} R * S = I_{\set{in}_z}.   
  \end{align}
  By
  a direct computation we have
  \begin{align}
    \label{Seq:34}
     R * S =
  d_{\set{K}} \,I_{\alpha}\otimes I_{\beta} \otimes \ketbra{0}{0}_\gamma \otimes
  \ketbra{0}{0}_\delta
  \end{align}
  where we defined the sets
  \begin{align}
    \label{Seq:35}
    \begin{aligned}
    \alpha &:= \Big( (\set{in}'_x\cup \set{in}'_y)
    \setminus \set{H} \Big) \cap \set{in}_z\\
    \beta &:= \Big( (\set{in}'_x\cup \set{in}'_y)
    \setminus \set{H} \Big) \cap \set{out}_z\\
    \gamma &:= \Big( (\set{out}_x\cup \set{out}_y)
    \setminus \set{H} \Big) \cap \set{in}_z\\
 \delta &:= \Big( (\set{out}_x\cup \set{out}_y)
    \setminus \set{H} \Big) \cap \set{out}_z  .
    \end{aligned} 
  \end{align}
  Equation \eqref{Seq:36} then becomes
  \begin{align}
    \label{Seq:37}
    d_{\set{K}} d_\beta I_\alpha \otimes
    \ketbra{0}{0}_\gamma = I_\alpha \otimes I_\gamma
  \end{align}
  which implies $I_\gamma \propto
  \ketbra{0}{0}_\gamma$, i.e. $d_\gamma = 1$. Howevere, since
  $\gamma $ is a collection of non-trivial
  elementary types it must be $\gamma =
  \emptyset$.
  Then Equation \eqref{Seq:37} becomes
  $d_{\set{K}} d_\beta I_\alpha = I_\alpha $
  which implies
  $  d_{\set{K}} d_\beta = 1$. Since we assumed
  that $\set{K}$
  were a non empty collection of non-trivial
  elementary types, we have a contradiction.
\end{proof}

We are now ready to prove the characterisation of the admissible
type compositions.
In order to make the derivation clearer, we will
first prove this 
preliminary lemma.
\begin{lemma}
\label{Slmm:tracefixeddin}
  Let $x$ and $y$ be types and let us denote
   with $\set{H} := \set{Ele}_x \cap \set{Ele}_y$
   the set of non-trivial elementary types that
   $x$ and $y$ have in common.
   If the composition $x\ast y$ is admissible,
   then
   \begin{align}
     \label{Seq:39}
     \begin{aligned}
    & \forall R \in \Evd{x} , \, \forall S \in
     \Evd{y} , \;\; \Tr [R*S] =
     d_{\set{\widetilde{in}}} \\
     &\widetilde{\set{in}} := (\set{in}_x \cup
     \set{in}_y) \setminus \set{H}.  
     \end{aligned}
   \end{align}
\end{lemma}
\begin{proof}
  Let us fix some arbitrary $R \in \Evd{x}$
  and $S \in \Evd{y}$ and let us define
  \begin{align}
    \label{Seq:40}
    \begin{aligned}
      Q_x := I_{\set{in}_x} \otimes
      \ketbra{0}{0}_{\set{out}_x} \\
      Q_y := I_{\set{in}_y} \otimes
      \ketbra{0}{0}_{\set{out}_y} \\
      \ketbra{0}{0}_{\set{A}}  :=
      \bigotimes_{i \in       \set{A}} \ketbra{0}{0}_i
    \end{aligned}
  \end{align}
where $\ket{0}_i$ is some fixed (normalized) state
on system $\hilb{H}_i$ and   $I_{\set{A}} $
denotes the identity on the hilbert space
$\hilb{H}_{\set{A}} := \bigotimes_{i \in \set{A}}
\hilb{H}_i$.
From Lemma \ref{Slmm:noinputinput} we know that
$\set{H}\subseteq  (\set{in}_x  \cap 
   \set{out}_y ) \cup (\set{out}_x \cap
  \set{in}_y ) $
and therefore we have that
\begin{align}
  \label{Seq:28}
  \begin{aligned}
 & \Tr[Q_x * Q_y] = \Tr[I_{\widetilde{\set{in}}}
  \otimes  \ketbra{0}{0}_{\widetilde{\set{out}}}]
  = d_{\widetilde{\set{in}}} \\
  &\widetilde{\set{out}} := (\set{out}_x \cup
     \set{out}_y) \setminus \set{H}.
  \end{aligned}
\end{align}
  
Since $x*y$ is admissible we have that for any
$p\in[0,1]$ there exists a type $z$ such that
\begin{align}
  \label{Seq:30}
(p R +  (1-p) Q_x)* Q_y \in \Evd{z}.
\end{align}
By taking the trace on both side we have that
\begin{align}
  \label{Seq:31}
  p \Tr[R*Q_y] + (1-p)\Tr[Q_x*Q_y] \in \mathbb{N}
  \quad \forall p \in [0,1] 
\end{align}
which, for $p$ irrational, implies that
\begin{align}
  \label{Seq:41a}
  \Tr[R*Q_y] = \Tr[Q_x*Q_y] = d_{\widetilde{\set{in}}} .
\end{align}
On the other hand, if we consider
$Q_x *(p S +  (1-p) Q_y)$ we obtain
$  \Tr[Q_x*S] = d_{\widetilde{\set{in}}} $.
Finally by considering $R *(p S +  (1-p) Q_y)$
we obtain
\begin{align}
  \label{Seq:41b}
  \Tr[R*S] = d_{\widetilde{\set{in}}} .
\end{align}
which is the thesis.
\end{proof}

We are now ready to prove the main result of this section. 

\begin{proposition}
  \label{Sprop:charapropositioncomposition}
  Let $x$ and $y$ be types and let us denote
   with $\set{H} := \set{Ele}_x \cap \set{Ele}_y$
   the set of non-trivial elementary types that
   $x$ and $y$ have in common.
   Then,  the composition $x\ast y$ is admissible
   if and only 
\begin{align}
  \label{Seq:25}
  \begin{aligned}
     &\set{H} \subseteq  (\set{in}_x  \cap 
   \set{out}_y ) \cup (\set{out}_x \cap
  \set{in}_y ) \\
     &\forall R\in\set T_1(x),\; \forall S\in\set
  T_1(y),\; R\ast S\in T_1(\widetilde{\set{in}} \to \widetilde{\set{out}}),
  \end{aligned}
\end{align}
where
we define
$\tilde{\set{in}} := (\set{in}_x  \cup 
\set{in}_y ) \setminus \set{H}$
and
$\tilde{\set{out}} := (\set{out}_x  \cup 
\set{out}_y ) \setminus \set{H}$.
\end{proposition}
\begin{proof}
  If Equation \eqref{Seq:25} is satisfied, then
  the admissibility of the composition $x*y$ is
  trivially satisfied.
  
  We now show that the admissibility of $x*y$
  implies Equation \eqref{Seq:25}.
  Let us fix some arbitrary $R \in \Evd{x}$
  and $S \in \Evd{y}$.
    Since $x*y$ is admissible, we have that there
    exist some set  $\set{in}_z$ and $\set{out}_z$ such that
    \begin{align*}
      \begin{aligned}
         & T:=  (pR + (1-p)Q_x)*Q_y \in \Evd{\set{in}_z
           \to \set{out}_z}, \\
         & p < {d_{\widetilde{\set{in}}}}^{-1}
      \end{aligned}
    \end{align*}
    where $Q_x$ have been defined in Equation
    \eqref{Seq:40}.
    By direct computation we have
    \begin{align}
      \label{Seq:43}
      \begin{aligned}
      &      T = p R*Q_y + (1-p) I_{\widetilde{\set{in}}}
      \otimes  \ketbra{0}{0}_{\widetilde{\set{out}}}.
      \end{aligned}
    \end{align}
    Let us assume that there exists a nontrivial
    elementary type $A$ such that
    $A \in {\set{in}_z}$ and
    $A \not \in \widetilde{\set{in}}$, i.e.
    $A \in \set{in}_z\cap \widetilde{\set{out}}$.
      Consider now the state
      $\ketbra{1}{1}_{\set{in}_z} = \bigotimes_{i \in       \set{in}_z} \ketbra{0}{0}_i$
      where $\ket{1}_i$ is a (normalized state) such that
      $\braket{0}{1}_i = 0$.
      Such a $\ket{1}_i$ must exist for any $i$
      since we are considering non-elementary
      types.
      Since $T$ is a channel from
      $\hilb{H}_{\set{in}_z} $ to
      $\hilb{H}_{\set{out}_z} $
      we must have
      \begin{align*}
        \begin{aligned}
           1 =& \Tr[T(I_{\set{out}_z} \otimes
          \ketbra{1}{1}_{\set{in}_z})] =\\
        =&        p\Tr[R*Q_y(I_{\set{out}_z} \otimes
           \ketbra{1}{1}_{\set{in}_z})]  +\\
          & + (1-p)\Tr[(I_{\widetilde{\set{in}}}
           \otimes
           \ketbra{0}{0}_{\widetilde{\set{out}}})
           (I_{\set{out}_z} \otimes
            \ketbra{1}{1}_{\set{in}_z})] = \\
        = &p\Tr[R*Q_y (I_{\set{out}_z} \otimes
      \ketbra{1}{1}_{\set{in}_z})] \leq \\
     \leq &p
        \Tr[R*Q_y] =p 
        d_{\widetilde{\set{in}}} < 1 
        \end{aligned}
      \end{align*}
      where we used Lemma \ref{Slmm:tracefixeddin}
      for the identity
      $\Tr[R*Q_y] = 
      d_{\widetilde{\set{in}}}$.
      Therefore, it must be $\set{in}_z \cap
      \widetilde{\set{out}} = \emptyset$,
      i.e. $\set{in}_z  \subseteq
      \widetilde{\set{in}}$. However, since we
      have 
      $d_{\set{in}_z} = d_{ \widetilde{\set{in}}}$
      from Lemma \ref{Slmm:tracefixeddin},
      it must be $\set{in}_z  =
      \widetilde{\set{in}}$ and 
      $\set{out}_z  =
      \widetilde{\set{out}}$.
      Then we have
      \begin{align}
        \label{Seq:44}
         p R*Q_y + (1-p) I_{\widetilde{\set{in}}}
      \otimes
        \ketbra{0}{0}_{\widetilde{\set{out}}} \in \Evd{\widetilde{\set{in}}\to
        \widetilde{\set{out}}} .
      \end{align}
      By taking the trace of $p R*Q_y + (1-p) I_{\widetilde{\set{in}}}
      \otimes
        \ketbra{0}{0}_{\widetilde{\set{out}}} $
        over $\hilb{H}_{\widetilde{\set{out}}}$ we
        have
      \begin{align}
        \label{Seq:45}
        \begin{aligned}
          \Tr_{\widetilde{\set{out}}}[R*Q_y]  =
        I_{\widetilde{\set{in}}} 
        \implies \\
        R*Q_y  \in \Evd{\widetilde{\set{in}}
      \to \widetilde{\set{out}} }
        \end{aligned}
      \end{align}
      If we consider the composition
      $Q_x*(pS + (1-p)Q_y) $ and we follow the
      same step as above we obtain that
      \begin{align}
        \label{Seq:46}
        Q_x*S \in \Evd{\widetilde{\set{in}}
      \to \widetilde{\set{out}} }.
      \end{align}
      Let us now consider the composition
      \begin{align}
        \label{Seq:47}
         &   T':=  (pR + (1-p)Q_x)*(pS + (1-p)Q_y )
\\
  &    p < 1-{2}^{-\frac12}.   
      \end{align}
      Since $x*y$ is admissible, there exist some
      set $\set{in}_z$ and $\set{out}_z$ such that
    \begin{align}
      \label{Seq:42}
      \begin{aligned}
         & T' \in \Evd{\set{in}_z
      \to \set{out}_z}.
      \end{aligned}
    \end{align}
Let us now 
define the sets
\begin{align}
  \label{Seq:48}
  \begin{aligned}
  \alpha &:= \widetilde{\set{in}} \cap
    \set{in}_z,
& \beta &:= \widetilde{\set{in}} \cap
    \set{out}_z\\
    \gamma &:= \widetilde{\set{out}} \cap
    \set{in}_z,
   &
    \delta &:= \widetilde{\set{out}} \cap
    \set{out}_z.    
  \end{aligned}
\end{align}
where
Since $T' \in \Evd{\set{in}_z
  \to \set{out}_z}$ we can apply it to the state
$\ketbra{0}{0}_{\set{in}_z}$.
We have
\begin{align}
  \label{Seq:49}
  \begin{aligned}
     1 = \Tr[T'(I_{\set{out}_z} \otimes
  \ketbra{0}{0}_{\set{in}_z})]  \geq\\
\geq   (1-p)^2\Tr  [(I_{\widetilde{\set{in}}}
      \otimes
  \ketbra{0}{0}_{\widetilde{\set{out}}})
  (I_{\set{out}_z} \otimes
  \ketbra{0}{0}_{\set{in}_z}) =\\
  =(1-p)^2\Tr  [I_{\beta}  \ketbra{0}{0} _{\alpha
  \cup \gamma \cup \delta}] = (1-p)^2d_\beta 
  \end{aligned}
\end{align}
If $\beta \neq \emptyset$ then $d_\beta \geq 2$
which would imply $(1-p)^2d_\beta > 1$.  Then it
must be $\beta = \emptyset$,
i.e. $\widetilde{\set{in}} \subseteq \set{in}_z$.
However, since
$d_{\widetilde{\set{in}} }= d_{\set{in}_z}$ from
Lemma \ref{Slmm:tracefixeddin}, it must be
$\widetilde{\set{in}} = \set{in}_z$.  We have then
proved that
\begin{align}
  \label{Seq:50}
  T' \in
\Evd{\widetilde{\set{in}} \to \widetilde{\set{out}}}.
\end{align}
Finally, Equations~\eqref{Seq:45}, \eqref{Seq:46} and~\eqref{Seq:50}
imply $R*S \in \Evd{\widetilde{\set{in}} \to \widetilde{\set{out}}}$.
\end{proof}

The study the admissible composition of
higher order maps is simplified by considering the
following operation.

\begin{definition}[Admissible type contraction]
  \label{Sdef:admcontr}
  Let $x$ be a type and $A,B$ be non trivial
  elementary systems, which are equivalent, i.e. 
  $\dim(A) =  \dim(B)$. We say
  that the
  \emph{contraction} $\mathcal{C}_{A,B}$ is admissible if
  \begin{align}
 & \forall R\in\set T_1(x), \;\exists
  z\,\textit{s.t.}
  \;\mathcal{C}_{A,B}(R) \in \set
  T_1(z),     \\
\label{Seq:57} & \mathcal{C}_{A,B}(R):=R\ast \Phi_{AB}\\
\label{Seq:60}&\Phi_{AB}:=\sum_{i,j}\ketbra{ii}{jj}\in
\mathcal{L}(\hilb H_A\otimes \hilb H_B).
  \end{align}

Let
$\set{H} \subseteq \set{Ele}_x \times \set{Ele}_x$
be a set of mutually disjoint pairs
$(A_i, A'_i) \in \set{Ele}_x \times \set{Ele}_x$,
of equivalent non
trivial elementary systems.
We then say that the contraction
$  \mathcal{C}_{\set H}$ is admissible if
\begin{align}
  \label{Seq:24}
  \begin{aligned}
 &   \forall R\in\set T_1(x), \;\;\exists
  z\;\mbox{ s.t. }\;
  \mathcal{C}_{\set{ H}}(R) \in \Evd{z}
  \\
  & \mathcal{C}_{\set{ H}}(R) := 
  \mathcal{C}_{A_1,A_1}(\mathcal{C}_{A_2,A_2}(\cdots\mathcal{C}_{A_n,A_n}(R)\cdots)).
  \end{aligned}     
\end{align}
\end{definition}
We notice that \eqref{Seq:22}
guarantees that Equation~\eqref{Seq:24} is well
defined, i.e. the order in which the contraction
are perfomed is immaterial.

It is worth to notice that the type contraction is
strictly connected to the the type composition.
Let $x,y$ be two types such that
$\set{Ele}_x \cap \set{Ele}_y = {A}$ and such that
the composition
$x\ast y$ is admissible. Then, for any two
maps $R \in \Eva{x}$, $S \in \Eva{y}$,
we have
\begin{equation*}
    R\ast S=\text{Tr}_A[RS^{T_A}]=(R\otimes S)\ast \Phi_{AA}=:\mathcal{C}_{A,A}(R\otimes S).
\end{equation*}
The generalisation to an arbitrary set $\set{H}:=
\set{Ele}_x\cap\set{Ele}_y$ is straightforward:
\begin{align}
  \label{Seq:20}
  \begin{aligned}
  R\ast S&=\mathcal{C}_{\set H}(R\otimes S):= \\
  &=\mathcal{C}_{A_1,A_1}(\mathcal{C}_{A_2,A_2}(\cdots\mathcal{C}_{A_n,A_n}(R\otimes S)\cdots))      
  \end{aligned}
\end{align}

The following lemma proves the analogy between the
admissibility of a type composition and the
allowed contractions of the shared elementary
systems, operated on the tensor product between
the considered types.
\begin{lemma}
\label{Slmm:contrccomposeq}
The composition $x\ast y$ is admissible if and
only if   $\mathcal{C}_{\set H}(x\otimes y)$ is
admissible, where
$\set{H} := \{ (A_i,A_i), \;\; A_i \in\set{Ele}_x
\cap \set{Ele}_y \}.$
\end{lemma}
\begin{proof}
$(\impliedby)$. This follows directly from
Definition~\ref{SadmT} and Equation~\eqref{Seq:20}.\\
$(\implies)$. Let us recall that 
$ R\in \set T_1(x\otimes y) $
if and only if $R$ is a positive operator in the
affine hull of $ T_1(x)\otimes T_1(y)$, i.e.
   $ R=\sum_ic_iA_i\otimes B_i$
with  $A_i\in \set T_1(x)$,  $B_i\in \set T_1(y)$,
$c_i\in \mathbb{R}$ and 
$    \sum_i c_i=1$.
Clearly, $C_\set{H}(R)$ is a positive  operator and
moreover
we have
\begin{align*}
    C_\set{H}(R)&=\sum_ic_iC_\set{H}(A_i\otimes B_i)=
    \sum_ic_iA_i\ast B_i
    =\sum_ic_i Z_i
\end{align*}
where $Z_i\in \set
T_1(\widetilde{\set{in}}\rightarrow
\widetilde{\set{out}})$ for any $i$ and we
defined
$\widetilde{\set{in}}= (\set{in}_x \cup \set{in}_y)\setminus
\set H$ and $\widetilde{\set{out}}=(\set{out}_x
\cup \set{out}_y)\setminus \set H$.
This implies that $C_\set{H}(R) \in T_1(\widetilde{\set{in}}\rightarrow
\widetilde{\set{out}})$.
\end{proof}
% Then we have 
% \begin{equation*}
% C_{\set H}(R)=\sum_ic_i(\lambda_{\widetilde{\set{in}}_{x\otimes y}\rightarrow \widetilde{\set{out}}_{x\otimes y}}I+T_i)=\lambda_{\widetilde{\set{in}}_{x\otimes y}\rightarrow \widetilde{\set{out}}_{x\otimes y}}I+T
% \end{equation*}
% with $T\in \set L_{\widetilde{\set{in}}_{x\otimes y}\rightarrow \widetilde{\set{out}}_{x\otimes y}}$.
% Consequently $C_{\set I}(R)\in \set T_1(\widetilde{\set{in}}_{x\otimes y}\rightarrow\widetilde{\set{out}}_{x\otimes y})$. This concludes the proof.

The following result characterize the admissible contractions.
\begin{proposition}\label{Sprop:admcontractions}
  Let $x$ be a type and $A$,$B \in
  \set{Ele}_x$ be equivalent nontrivial elementary
  types. If $A, B \in \set{in}_x$ or
  $A, B \in \set{out}_x$ then
  $\mathcal{C}_{A,B} (x)$ is not admissible. If
  $A \in \set{in}_x$ and $B \in \set{out}_x$ then
  $\mathcal{C}_{A,B}(x)$ is admissible if and only
  if, for any $R \in \set T_1(x)$ we have that
  $\mathcal{C}_{A,B}(R) \in
  T_1(\widetilde{\set{in}}_{x}\rightarrow\widetilde{\set{out}}_{x})$,
  where we defined
  $\widetilde{\set{in}}_x := \set{in}_x\setminus
  A$ and
  $\widetilde{\set{out}}_x := \set{out}_x
  \setminus B$.
\end{proposition}
\begin{proof}
This proposition can be proved along the same
lines of the proof of Proposition
\ref{Sprop:charapropositioncomposition}.
\end{proof}

Thanks to Lemma \ref{Slmm:contrccomposeq} and
Proposition \ref{Sprop:admcontractions}, the study
of the admissible compositions between types
reduces to the study of the admissible
contractions for one type only. We therefore focus
on the latter.

Our next result exploits the
language that we developped in Section
\ref{Ssec:comb-struc-high} and it shows that the
characterization of the admissible contractions
is solved by a rather simple algorithm.
The first step is to prove the following relation
between the contraction of a map, as defined in Equation~\eqref{Seq:57}  and the
contraction of a labeled string which we defined
in  Definition \ref{Sstringcontr}

\begin{lemma}\label{Sstrsetcontr}
  Let $S$ a set of binary strings which are labeled by a
  set $\set{A}:= \{A_1 ,A_2 , \dots, A_n\}$ of indexes.
  Then we have, for any $A_i,A_j \in \set{A}$
  % and
  % $\set L_S=\underset{\textbf{b}\in
  %   S}{\bigoplus}\set L_{\textbf{b}}$, then we
  % have
\begin{align}
   T\in\set L_S\implies  \mathcal{C}_{Ai,Bj}(T) \in \set L_{\mathcal{C}_{Ai,Bj}(S)}
\end{align}
\end{lemma}
\begin{proof}
  Let us consider an arbitrary $T\in\set L_S$.
  We can expand $T$ as a linear combination as follows:
  \begin{align}
    \label{Seq:58}
&    T=\sum_{\textbf{b}\in S}T_{\textbf{b}} \\
  &  T_{\textbf{b}} =  \sum_{a_1,a_2 , \dots, a_n}
    c_{a_1,a_2 , \dots, a_n}  t_{a_1} \otimes
    t_{a_2} \otimes \cdots \otimes  t_{a_n}
  \end{align}
  where $ c_{a_1,a_2 , \dots, a_n}$ are real
  coefficients and $\{ t_{a_i} \}$, with
  $a_i = 1, \dots, \dim(\set{L}_{Ai})$, is a basis
  of $\set{L}_{b_Ai}$ (if $b_{Ai} =0$ then
  $\{ t_{a_i} \}$ is a basis of the space of
  traceless Hermitian operator on $\hilb{H}_{Ai}$,
  if $b_{Ai} =0$ then $\{ t_{a_i} \} = I_{Ai}$ ).
  Then we have
  \begin{align}
\mathcal{C}_{Ai,Aj}\left(\overset{n}{\underset{k=1}{\bigotimes}}t_{a_k}\right)=
\begin{cases}
  \underset{i\neq i,j}{\bigotimes} t_{a_k}\Tr[t_{a_i}t_{a_j}]\quad b_i=b_j\\
0\quad b_i\neq b_j
\end{cases}
\end{align}
It follows that
$\mathcal{C}_{Ai,Aj}\left(\overset{n}{\underset{k=1}{\bigotimes}}t_{a_k}\right)\in\set
L_{\mathcal{C}_{Ai,Aj}(\textbf{b})}$, which for
linearity implies
$\mathcal{C}_{Ai,Aj}(T)\in\set
L_{\mathcal{C}_{Ai,Aj}(S)}$.
\end{proof}

We now can express the result of Proposition
\ref{Sprop:admcontractions}
in terms of set of binary strings.

\begin{proposition}\label{Sadmissibilitystring}
  Let $x$ be a type and let $A\in\set{in}_x$ and
  $B\in\set{out}_x$, be equivalent non-trivial
  elementary types.  Then $\mathcal{C}_{A,B}(x)$ is
  admissible if and only if
  $\mathcal{C}_{A,B}(D_x)\subseteq
  D_{\widetilde{\set{in}}_x\rightarrow\widetilde{\set{out}}_x}$,
  where
  $\widetilde{\set{in}}_x:=\set{out}_x\setminus A$
  and
  $\widetilde{\set{out}}_x:=\set{out}_x\setminus
  B$
\end{proposition}
\begin{proof}
  We know according to
  Proposition~\ref{Sprop:admcontractions}, that
  $ \mathcal{C}_{A,B}(x)$ is admissible
  $\iff\{\mathcal{C}_{A,B}(R),\;
  R\in\set{T}_1(x)\}\subseteq \set
  T_1(\widetilde{\set{in}}_x\rightarrow\widetilde{\set{out}}_x)=\mathscr{M}(\widetilde{\lambda},D_{\widetilde{\set{in}}_x\rightarrow\widetilde{\set{out}}_x})$,
  where
  $\widetilde{\lambda}={d_{\widetilde{\set{out}_x}}}^{-1}=d_B(d_{{\set{out}_x}})^{-1}$.

  First, let as assume that 
  $\mathcal{C}_{A,B}(D_x)\subseteq
  D_{\widetilde{\set{in}}_x\rightarrow\widetilde{\set{out}}_x}$.
  Then 
  we have
  $\mathscr{M}(\widetilde{\lambda},\mathcal{C}_{A,B}(D_x))\subseteq\mathscr{M}(\widetilde{\lambda},D_{\widetilde{\set{in}}_x\rightarrow\widetilde{\set{out}}_x})=\set
  T_1(\widetilde{\set{in}}_x\rightarrow\widetilde{\set{out}}_x)$.
  Let $R\in\set T_1(x)$, then
  $R=\lambda_xI+T \geq 0$, where
  $T\in \set L_{D_x}$ and $\lambda_x = d_{\set{out}_x}$.  From to
  Lemma~\eqref{Sstrsetcontr}, we have
  $\mathcal{C}_{A,B}(T)\in\set
  L_{\mathcal{C}_{A,B}(D_x)}$. Moreover, we have
  $\mathcal{C}_{A,B}(\lambda
  I)=\widetilde{\lambda}I_{\widetilde{\set{in}}_x\rightarrow\widetilde{\set{out}}_x}$.
  Since the contraction of a positive operator is
  still positive, we have that $\forall R\in \set T_1(x)$,
  $\mathcal{C}_{AB}(R)\in\mathscr{M}(\widetilde{\lambda},\mathcal{C}_{AB}(D_x))\subseteq
  \mathscr{M}(\widetilde{\lambda},D_{\widetilde{\set{in}}_x\rightarrow\widetilde{\set{out}}_x})$,
  that is $\mathcal{C}_{A,B}(x)$ is admissible.

  Let us now prove
that   if $\mathcal{C}_{A,B}(x)$ is
  admissible then
  $\mathcal{C}_{A,B}(D_x)\subseteq
  D_{\widetilde{\set{in}}_x\rightarrow\widetilde{\set{out}}_x}$.
By contradiction, let us assume that 
  $\mathcal{C}_{A,B}(D_x)\nsubseteq
  D_{\widetilde{\set{in}}_x\rightarrow\widetilde{\set{out}}_x}$.

  Then there exists a non empty string
  $\widetilde{\textbf{b}} \in \mathcal{C}_{A,B}(D_x) $ such that
 $\widetilde{\textbf{b}}\not
 \in
  D_{\widetilde{\set{in}}_x\rightarrow\widetilde{\set{out}}_x}$.
Let  
$\textbf{b}\in D_x$, be a string such that
$\mathcal{C}_{A,B}(\textbf{b}) = \widetilde{\textbf{b}}$.
Since
  $\widetilde{\textbf{b}}\neq {\varepsilon}$,
  then we have either $b_A = b_B = 0$ or $b_A =
  b_B = 1$.
Let us assume $b_A = b_B = 0$ and
let us consider the following operator
\begin{align}
  \label{Seq:59}
  \begin{aligned}
  T : = t_1\otimes\cdots\otimes
  t_A\otimes\cdots\otimes t_B\otimes\cdots\otimes
  t_n \in \set L_{\textbf{b}}\\
 t_i \in \set{L}_{b_i} \quad   t_A = t_B =
 \ketbra{0}{0} -\ketbra{1}{1}
  \end{aligned}
\end{align}
where $\ket{0}$ and $\ket{1}$ are orthonormal
states (we remember that a choice of an isomorphism between
$\hilb{H}_A $ and $\hilb{H}_B$ is implicitly
assumed in Equation \eqref{Seq:60}).
Then we have that $\mathcal{C}_{A,B}(T) \neq 0$.  
Consequently, considering
  $R=\lambda_x I+\epsilon T$, with
  $\epsilon\in \mathbb{R}$ small enough in order
  to have $R \geq 0$, we obtain that
  $R\in \set T_1(x)$ and
  $\mathcal{C}_{AB}(R)=\widetilde{\lambda}I_{\widetilde{\set{in}}_x\rightarrow\widetilde{\set{out}}_x}+\epsilon \,
  \mathcal{C}_{A,B}(T)$. Hence, since
  $\mathcal{C}_{A,B}(T)\in \set L_{\widetilde{\textbf{b}}}$ we
have that 
  $\mathcal{C}_{A,B}(R)\notin \set
  T_1(\widetilde{\set{in}}_x\rightarrow\widetilde{\set{out}}_x)$,
  which contradicts the admissibility of
  $\mathcal{C}_{A,B}(x)$.
 The same proof applies to the case
  $b_A=b_B=1$ by consider the operator 
 \begin{align}
  \label{}
  \begin{aligned}
  T : = t_1\otimes\cdots\otimes
  t_A\otimes\cdots\otimes t_B\otimes\cdots\otimes
  t_n \in \set L_{\textbf{b}}\\
 t_i \in \set{L}_{b_i} \quad   t_A = t_B =
I
  \end{aligned}
\end{align}
  % such that $R=\lambda_x I+ T\geq 0$,  \textcolor{blue}{ue, il fattore epsilon è superfluo in questo caso, giusto?}
  %  we have $R\in \set T_1(x)$ and
  %  $\mathcal{C}_{AB}(R)=\widetilde{\lambda}I_{\widetilde{\set{in}}_x\rightarrow\widetilde{\set{out}}_x}+ 
  % \mathcal{C}_{A,B}(T)$.
  % Knowing that $\mathcal{C}_{A,B}(T)\in \set L_{\widetilde{\textbf{b}}}$,
  % we have $\mathcal{C}_{A,B}(T)\notin L_{D_{\widetilde{\set{in}}_x\rightarrow\widetilde{\set{out}}_x}}\implies\mathcal{C}_{A,B}(T)\notin T_1(\widetilde{\set{in}}_x\rightarrow\widetilde{\set{out}}_x)$
  % reaching the contradiction again.
\end{proof}

\begin{corollary}
  \label{Scoroll:criticstringsetAB}
  Let $x$ be a type and let $A\in \set{in}_x$,
  $B\in \set{out}_x$ be non-trivial equivalent
  elementary types. Let us define the following set of
  binary strings
  \begin{align}
    \label{Seq:61}
    S^x_{AB}:=
    W_{\widetilde{\set{in}}}0_A
    \textbf{e}_{\widetilde{\set{out}}}0_B. 
  \end{align}
where we remind that $\widetilde{\set{in}} :=
\set{in}_x\setminus B$ and
$\widetilde{\set{out}} :=
\set{out}_x\setminus B$.
  Then
  $\mathcal{C}_{A,B}(x)$ is admissible if and only
  if $D_x\cap S^x_{AB}=\emptyset$.%  and we call
  % $S^x_{AB}$ the critical string set.
\end{corollary}
\begin{proof}

  If $\mathcal{C}_{AB}(x)$ is admissible, then
  $\textbf{e}_{\set{in}_x\setminus
  A}T_{\set{out}_x\setminus
  B}\subseteq\mathcal{C}_{AB}(D_x)\subseteq
  W_{\set{in}_x\setminus
  A}T_{\set{out}_x\setminus B}$ according to
  Proposition~\ref{Spropositionchannel}.
  Furthermore,
  $\mathcal{C}_{AB}(S^x_{AB})=W_{\set{in}_x\setminus
  A}\mathbf{e}_{\set{out}_x\setminus B}$, and
  then 
  $\mathcal{C}_{AB}(S^x_{AB})\cap\mathcal{C}_{AB}(D_x)=\emptyset$,
  implies that
  $D_x\cap S^x_{A,B}=\emptyset$.

  Let us now prove the implication
  $D_x\cap S^x_{A,B}=\emptyset \implies
  \mathcal{C}_{A,B}(x)$ is admissible.  By
  contradiction, let us assume that
  $\mathcal{C}_{A,B}(x)$ is not admissible.  By
  Lemma \ref{Sadmissibilitystring} this implies
  that
  % $\mathcal{C}_{A,B}(D_x) \not \subseteq
  % D_{\widetilde{\set{in}}_x\rightarrow\widetilde{\set{out}}_x}
  % \neq \emptyset$,
  $\mathcal{C}_{A,B}(D_x) \cap
  D^{\perp}_{\widetilde{\set{in}}_x\rightarrow\widetilde{\set{out}}_x}
  \neq \emptyset$.  Therefore
  $D_x \cap
  \mathcal{C}^{-1}_{A,B}(D^{\perp}_{\widetilde{\set{in}}_x\rightarrow\widetilde{\set{out}}_x})
  \neq \emptyset$

Since $
  D_{\widetilde{\set{in}}_x\rightarrow\widetilde{\set{out}}_x}^\perp
  = W_{\set{in}_x}
  \textbf{e}_{\set{out}_x}$
  we have that
$
\mathcal{C}^{-1}_{A,B}(D^{\perp}_{\widetilde{\set{in}}_x\rightarrow\widetilde{\set{out}}_x})
=  S^x_{A,B} \cup W_{\set{in}_x\setminus
  A}1_A\textbf{e}_{\set{out}_x}. $.
From Proposition~\ref{Spropositionchannel} we already know
that
$D_x \cap W_{\set{in}_x} \textbf{e}_{\set{out}_x} =
  \emptyset $
and therefore we have 
$D_x \cap S^x_{A,B}  \neq \emptyset $ which
contradicts the hypothesis.
\end{proof}
% Since $S^x$The thesis follows by 
% i.e
% Let us assume that
% Let us consider $D_x\cap S^x_{AB}=\emptyset$, then
% knowing that $D_x\subset W_{\set{in}_x}T_{\set{out}}$, the 
% only set of strings that could break the

% The thesis follows from Proposition
% \ref{Sadmissibilitystring} by noticing that
% $S^x_{A,B}$ is the complemement of  $ D_{\widetilde{\set{in}}_x\rightarrow\widetilde{\set{out}}_x}$.

\begin{corollary}
\label{Scor:badstringmultiplecontraction}
  Let $x$ be a type and let
$\set{H} := \{(A_i, A'_i)\}_{i=1}^n$
be a set of mutually disjoint pairs
$(A_i, A'_i) \in \set{Ele}_x \times \set{Ele}_x$,
of equivalent non
trivial elementary systems
such that $A_i \in \set{in}_x$ and
$A'_j \in \set{out}_x$.
Let us define the following sets:
\begin{align}
  \label{Seq:64}
  &\begin{aligned}
 & \begin{aligned}
    \set{in}_{\set{H}} := \{ &A_i \in \set{in}_x
    \mbox{ s.t.  }  (A_i , A'_j) \in \set{H}  \\
    &\mbox{ for some } A'_j  \in \set{out}_x\}
  \end{aligned}
  \\
   & \begin{aligned}
    \set{out}_{\set{H}} := \{ &A'_j \in \set{out}_x
    \mbox{ s.t.  }  (A_i , A'_j) \in \set{H}  \\
    &\mbox{ for some } A_i \in \set{in}_x\}
  \end{aligned}
  \\
  &\widetilde{\set{in}} : = \set{in}_x \setminus
  \set{in}_{\set{H}} \qquad
  \widetilde{\set{out}} : = \set{out}_x \setminus
  \set{out}_{\set{H}}   
  \end{aligned} \\
  % \begin{align}
  %   \label{Seq:63}
  %   S^x_{\set{H}}:=
  %   \bigcup_{
  %   \mathbf{b}^{({\set{in}_{\set{H}}})} ,
  %   {\mathbf{b}'}^{({\set{out}_{\set{H}}})} }
  %   W_{\widetilde{\set{in}}} \,
  %   \mathbf{b}^{({\set{in}_{\set{H}}})}
  %   \,
  %   \textbf{e}_{\widetilde{\set{out}}}
  %   \,{\mathbf{b}'}^{({\set{out}_{\set{H}}})} 
      %     \end{align}
  &\begin{aligned}
    \label{Seq:63}
   & S^x_{\set{H}}:=
    \bigcup_{
    \mathbf{b} ,
    {\mathbf{b}'} }
    W_{\widetilde{\set{in}}} \,
    \mathbf{b}\,
    \mathbf{e}_{\widetilde{\set{out}}}
    \,{\mathbf{b}'}
    \end{aligned}
  \end{align}
  where
  % $n = |\set{H}| = |\set{in}_{\set{H}}| =
% |\set{out}_{\set{H}}| $
% is the cardinality of the set $\set{H}$
% and
  the binary strings $\mathbf{b}$ and
  $\mathbf{b}'$
  are such that 
  $  \mathbf{b} \in
    T_{\set{in}_{\set{H}}}$,
$    \mathbf{b}' \in
T_{{\set{out}_{\set{H}}} } $ and
$b ^{({\set{in}_{\set{H}}})} _{Ai} =
    b^{({\set{out}_{\set{H}}})}_{A'j} $ if 
    $  (A_i , A'_j) \in \set{H} $.
    Then we have that
    \begin{align}
      \label{Seq:66} 
      \mathcal{C}_{\set H}(x) \mbox{ is admissible}
      \iff
      D_x \cap S^x_{\set{H}} = \emptyset.
    \end{align}    
% \begin{align}
%   \label{Seq:65}
%   &  \mathbf{b}^{({\set{in}_{\set{H}}})} \in
%     T_{\set{in}_{\set{H}}}  \qquad
%     \mathbf{b}^{({\set{out}_{\set{H}}})} \in
%     T_{{\set{out}_{\set{H}}} } \\
%   &(A_i , A'_j) \in \set{H} \implies
%     b ^{({\set{in}_{\set{H}}})} _{Ai} = b^{({\set{out}_{\set{H}}})}_{A'j} .
% \end{align}
  \end{corollary}
  \begin{proof}
    This result is a rather straightforward
    generalization of Corollary
    \ref{Scoroll:criticstringsetAB}.
    $\mathcal{C}_{\set{H}}$ is not admissible if
    and only if there exists a string $\mathbf{s} \in D_x$
    such that
    $\mathcal{C}_{\set{H}}(\mathbf{s}) \not \in
    D_{\widetilde{\set{in}} \to
      \widetilde{\set{out}} }$ i.e
    $\mathcal{C}_{\set{H}}(\mathbf{s}) =
    \mathbf{w}_{\widetilde{\set{in}}}\mathbf{e}_{\widetilde{\set{out}}}$
    for some string $\mathbf{w}$.
    However,
    $    \mathcal{C}_{\set{H}}(\mathbf{s}) \neq \varepsilon$
    if and only 
    $\mathbf{s} \in   
    \bigcup_{
    \mathbf{b} ,
    {\mathbf{b}'} }
    W_{\widetilde{\set{in}}} \,
    \mathbf{b}\,
    W_{\widetilde{\set{out}}}
    \,{\mathbf{b}'}
    $
      where
  $  \mathbf{b} \in
    T_{\set{in}_{\set{H}}}$,
$    \mathbf{b}' \in
T_{{\set{out}_{\set{H}}} } $ and
$b ^{({\set{in}_{\set{H}}})} _{Ai} =
    b^{({\set{out}_{\set{H}}})}_{A'j} $ if 
    $  (A_i , A'_j) \in \set{H} $.
    Then we have that
     $\mathcal{C}_{\set{H}}$ is not admissible if
     and only if there exist a string $\mathbf{s} \in D_x$ such that
     $\mathbf{s} \in S^x_{\set{H}}$.     
    %that can be
    % proved by induction on the cardinality $\# \set{H}$ of
    % $\set{H}$.
    % If $\# \set{H} = 1$ the result holds from Corollary
    % \ref{Scoroll:criticstringsetAB}.
    % Let us suppose that the thesis holds for
    % $\# \set{H} = n-1$, i.e.
    % $\set{H}' = \{(A_i, A'_i)\}_{i=1}^{n-1}$,
    % and let us consider $\set{H} = \{(A_i, A'_i)\}_{i=1}^{n}$.
    % Then $S^x_{\set{H}} =  $
    % Then
    % $\mathcal{C}_{A_n,A'_n}(x) $ is admissible  if and only if
    % $S^x_{A_n,A'_n} \cap D_x = \emptyset $ 
  \end{proof}

This result  shows that the problem of checking
whether a given set of contractions (or a type
composition)
is admissible, is solved by a simple algorithm
which check whether the set $D_x$ contains some of
the string of $S^x_{\set{H}}$.
For sake of clarity, let us see an explicit example:
\begin{align}
  \label{Seq:67}
  \begin{aligned}
 & x = (A\to B) \otimes (C \to D), \\
 & \begin{aligned}
     D_x = \{ &0_A0_B0_C0_D,   1_A0_B0_C0_D, \\
  &0_A0_B1_C0_D,
  1_A0_B1_C0_D,  \\
  &1_A1_B0_C0_D,
  1_A1_B1_C0_D, \\
  &0_A0_B1_C1_D,
  1_A0_B1_C1_D\}   ,
\end{aligned} \\
& \set{in}_x = \{A, C  \} \quad \set{out}_x = \{B,
D  \} ,\\
&\begin{aligned}
  S^{x}_{A,B} = &\{ 0_C0_A1_D0_B, 1_C0_A1_D0_B \} =\\
  =&\{ 0_A0_B0_C1_D, 0_A0_B1_C1_D \},
\end{aligned}\\
&\begin{aligned}
  S^{x}_{C,B} = &\{ 0_A0_C1_D0_B, 1_A0_C1_D0_B \} =\\
  =&\{ 0_A0_B0_C1_D, 1_A0_B0_C1_D \},
\end{aligned}\\
&\begin{aligned}
  S^{x}_{A,D} = &\{ 0_C0_A1_B0_D, 1_C0_A1_B0_D \} =\\
  =&\{ 0_A1_B0_C0_D, 0_A1_B1_C0_D \},
\end{aligned}\\
&\begin{aligned}
  &\set{H} := \{ (C,B) , (A,D)\}, \\
  &\begin{aligned}
  S^{x}_{\set{H}} = &\{ 0_A0_C0_B0_D, 
  0_A1_C0_B1_D, 1_A0_C1_B0_D,\}  \\
  = & \{ 0_A0_B0_C0_D ,  0_A0_B1_C1_D, 1_A1_B0_C0_D   \}.
\end{aligned}
\end{aligned}
  \end{aligned}
\end{align}
We then have that $\mathcal{C}_{C,B}(x)$ and
$\mathcal{C}_{A,D}(x)$ are admissible while
$\mathcal{C}_{A,B}(x) $ and
$\mathcal{C}_{\set{H}}(x) $ are not.  This example
also clarify that the admissibility of a set of
contraction cannot be reduced to the admissibility
of the single contractions in the set. Indeed, it
is possible to make either $\mathcal{C}_{C,B}(x)$
or $\mathcal{C}_{A,D}(x)$ but not both.  Moreover,
corollary \ref{Scor:badstringmultiplecontraction}
does exclude the possibility that the individual
contractions (or a subset of contractions) of a
set are not admissible but the set as whole is
admissible.
However, our next result, which is stated in
Proposition \ref{Sprop:subsetofcontractionarenotadmissibleyoucannotrecover},
rules out this
possibility. 
The key to prove this result is the following lemma.
\begin{lemma}
  \label{SLemm:promoting0to1andviceversa}
  Let $x$ be a type, let $D_x$ be its set of strings.
  \begin{enumerate}
  \item
    For an arbitrary $A \in \set{in}_x$ consider a
  string  $\set{s} $ such that
  $\set{s} = \set{w} 0_A \set{w'}$ where
  $\set{w} \in W_{\set{in}_x \setminus A}$
  and
  $\set{w'} \in W_{\set{out}_x }$.
  If $\set{s} \in  D_x$ then also
  $\set{s'} :=  \set{w} 1_A \set{w'} \in D_x$.
\item
   For an arbitrary $B \in \set{out}_x$ consider a
  string  $\set{s} $ such that
  $ \set{s} = \set{w} 1_B \set{w'} $ where
  $ \set{w} \in W_{\set{in}_x } $
  and
  $\set{w'} \in W_{\set{out}_x \setminus B }$.
  If $\set{s} \in  D_x$ then also
  $\set{s'} :=  \set{w} 0_B \set{w'} \in D_x$.
  \end{enumerate}
\end{lemma}
\begin{proof}
If $x = E $ is an elementary type then $\set{in}_x = \emptyset$
and statement $1$ is trivially true.
Since  $D_E = \{ 0_E \}$ also item $2$ is
trivially true.
Let us assume that both the statements hold for 
arbitrary $x$ and $y$.
Our goal is to prove that the statements hold for
$x\to y$. We will split the proof into two
part. In the first part we will prove that the
lemma holds for $\overline{x}$ and in the second
part we will prove that it holds for
$x \otimes y$.
The thesis for $x\to y$ follows from the identity
$x\to y = \overline{x \otimes \overline{y}}$ of
Proposition \ref{Sprop:typeequivalences}.

We now prove that the statements
hold for $\overline{x}$. we start with item $1$. Let us consider an arbitrary $A \in
\set{in}_{\overline{x}}$
and let us suppose that  $\set{s}\in D_{\overline{x}}$ where
we defined the string
\begin{align}
  \label{eq:1}
  \set{s} = \set{w} 0_A \set{w'}, \quad
   \set{w} \in W_{\set{in}_{\overline{x}} \setminus A }, \quad
  \set{w'} \in W_{\set{out}_{\overline{x}}}.
\end{align}
From Equation \eqref{eq:4}
we have that $D_{\overline{x}} = \overline{D_x}$
and therefore
$\set{s}\in \overline{D_{{x}}}$.
Since $\set{in}_{\overline{x}}  = \set{out}_{{x}}$
and using the inductive hypothesis with item $2.$
we have that
\begin{align}
  \label{eq:6}
  &\set{s} \not \in {D_{{x}}} \implies
  \set{s'} \not \in {D_{{x}}}, \\
   &\set{s'} := \set{w} 1_A \set{w'}, \quad
   \set{w} \in W_{\set{in}_{\overline{x}} \setminus A }, \;
  \set{w'} \in W_{\set{out}_{\overline{x} } }.
\end{align}
If $\set{s'} \not \in {D_{{x}}} $ then either 
 $\set{s'} = \set{e}_x$ or 
$\set{s'}  \in \overline{D_{{x}}} $.
If $\set{s'} = \set{e}_x$  then it must be 
$\set{s} = \set{e}_{\set{in}_x} 0_A
\set{e_{\set{out}_x \setminus A}} \in
\set{e}_{\set{in}_x} T_{\set{out}_x} $.
However, from Proposition
\ref{Spropositionchannel}
we have that
$\set{e}_{\set{in}_x} T_{\set{out}_x} \subseteq
D_x$ which contradicts the hypothesis
$\set{s} \in \overline{D_x}$.  Then it must be
$\set{s'} \in \overline{D_{{x}}} =
D_{\overline{x}}$ which proves that item $1$ of
the lemma holds for $\overline{x}$.

The proof of item $2$ is similar.
Let us consider an arbitrary $B \in
\set{out}_{\overline{x}}$
and let us suppose that  $\set{s}\in D_{\overline{x}}$ where
we defined the string
\begin{align}
  \label{eq:1}
  \set{s} = \set{w} 1_B \set{w'}, \quad
   \set{w} \in W_{\set{in}_{\overline{x}} }, \;
  \set{w'} \in W_{\set{out}_{\overline{x}} \setminus B}.
\end{align}
% From Equation \eqref{eq:4}
% we have that $D_{\overline{x}} = \overline{D_x}$
% and therefore
% $\set{s}\in \overline{D_{{x}}}$.
% Since $\set{in}_{\overline{x}}  = \set{out}_{{x}}$
Using the inductive hypothesis with item $1$ of
the lemma
we have that
\begin{align}
  \label{eq:6}
  &\set{s} \not \in {D_{{x}}} \implies
  \set{s'} \not \in {D_{{x}}}, \\
   &\set{s'} := \set{w} 0_B \set{w'}, \quad
   \set{w} \in W_{\set{in}_{\overline{x}} }, \;
  \set{w'} \in W_{\set{out}_{\overline{x} } \setminus A }.
\end{align}
If $\set{s'} \not \in {D_{{x}}} $ then either 
 $\set{s'} = \set{e}_x$ or 
$\set{s'}  \in \overline{D_{{x}}} $.
If $\set{s'} = \set{e}_x$  then it must be 
$\set{s} = \set{e}_{x}$ which contradicts the
hypothesis that
$\set{s} \in \overline{D_{{x}}} \subseteq T_x$.
Then it must be
$\set{s'} \in \overline{D_{{x}}} =
D_{\overline{x}}$ which proves that item $2$ of
the lemma holds for $\overline{x}$.

Let us now assume that both the statements hold
for arbitrary $x$ and $y$. We will now prove that
the statements hold for $x \otimes y$.

Let us start with item $1$.
Let us remind that
$\set{in}_{x \otimes y} =
\set{in}_{x } \cup \set{in}_{ y}$,
$\set{out}_{x \otimes y} =
\set{out}_{x } \cup \set{out}_{ y}$
and
$  D_{x\otimes y} =  \textbf{e}_x D_y \cup
  D_x\textbf{e}_y  \cup D_xD_y   $ (see
  Eq.~\eqref{Seq:4}).
  Let us consider an arbitrary $A \in
\set{in}_{x \otimes y}$
and let us suppose that  $\set{s}\in D_{x \otimes
  y}$ where
\begin{align}
  \label{eq:10}
  \set{s} = \set{w} 0_A \set{w'}, \quad
   \set{w} \in W_{\set{in}_{x \otimes y} \setminus A }, \;
  \set{w'} \in W_{\set{out}_{x \otimes y}}.
\end{align}
Since we have
\begin{align}
  \label{eq:14}
  A \in \set{in}_x \implies  \set{s}\in
  D_x\textbf{e}_y  \cup D_xD_y \\
  A \in \set{in}_y \implies  \set{s}\in
  \textbf{e}_x D_y  \cup D_xD_y 
\end{align}
the thesis follows from the inductive hypothesis.
We have thus proved that item $1$ holds for $x
\otimes y$.

Let us now prove item $2$.
  Let us consider an arbitrary $B \in
\set{out}_{x \otimes y}$
and let us suppose that  $\set{s}\in D_{x \otimes
  y}$ where
\begin{align}
  \label{eq:10}
  \set{s} = \set{w} 1_B \set{w'}, \quad
   \set{w} \in W_{\set{in}_{x \otimes y} }, \;
  \set{w'} \in W_{\set{out}_{x \otimes y}\setminus B}.
\end{align}
We have to consider four cases.
If $B \in \set{out}_x$ and
$\set{b} \in D_x\textbf{e}_y \cup D_xD_y $ the
thesis follows from the inductive hypothesis.
If $B \in \set{out}_x$ and
$\set{b} \in \textbf{e}_x  D_y $ then it must
be
\begin{align}
  \label{eq:15}
  \set{b} = \textbf{e}_{\set{in}_x}1_B
  \textbf{e}_{\set{out}_x \setminus B} \set{w}_y,
  \quad \set{w}_y \in D_y.
\end{align}
Then we have
\begin{align}
  \set{b}' := \textbf{e}_{\set{in}_x}0_B
  \textbf{e}_{\set{out}_x \setminus B} \set{w}_y
  \in \textbf{e}_{\set{in}_x}T_{\set{out}_x} D_y
  \subseteq
  D_xD_y
\end{align}
which prove the thesis.  By exchanging the role of
$x$ and $y$ we can prove the thesis for the 
remaining cases given by the $B \in \set{out}_y$.
We have then prove that the lemma holds for $x
\otimes y$.
\end{proof}
\begin{corollary}
  \label{Scor:promoting0to1multiplesystems}
    Let $x$ be a type, let $D_x$ be its set of strings.
  \begin{enumerate}
  \item
    For an arbitrary $\set{A} \subseteq \set{in}_x$ consider a
  string  $\set{s} $ such that
  $\set{s} = \set{w} \set{t}_{\set{A}} \set{w'}$ where
  $\set{w} \in W_{\set{in}_x \setminus \set{A}}$
  $\set{w'} \in W_{\set{out}_x }$ and
$\set{t} \in T_{\set{A} }$.
  If $\set{s} \in  D_x$ then also
  $\set{s'} :=  \set{w} \set{e}_{\set{A}} \set{w'}
  \in D_x$ .
\item
   For an arbitrary $\set{B} \subseteq \set{out}_x$ consider a
  string  $\set{s} $ such that
  $ \set{s} = \set{w} \set{e}_{\set{B}} \set{w'} $ where
  $ \set{w} \in W_{\set{in}_x } $ and
  $\set{w'} \in W_{\set{out}_x \setminus \set{B} }$.
  If $\set{b} \in  D_x$ then 
  $\set{s'} :=  \set{w} \set{t}_{\set{B}} \set{w'}
  \in D_x$
  for any $\set{t}_{\set{B}} \in T_{\set{B}}$.
  \end{enumerate}
\end{corollary}
\begin{proof}
  The result follows by iterating Lemma~\ref{SLemm:promoting0to1andviceversa}.
\end{proof}

\begin{proposition}
  \label{Sprop:subsetofcontractionarenotadmissibleyoucannotrecover}
  Let $x$  be a type and let 
  $\set{H} := \{ (A_i, A'_i) \}_{i \in \set{I}}$,
  $\set{K} := \{ (A_i, A'_i) \}_{i \in \set{J}}$
be two sets of mutually disjoint pairs
of equivalent non
trivial elementary systems
such that $A_i \in \set{in}_x$ and
$A'_j \in \set{out}_x$.
Let us suppose that $\set{H} \subseteq
\set{K}$. If 
$  \mathcal{C}_{\set{H}}(x)$ is not
  admissible then also  
$  \mathcal{C}_{\set{K}}(x)$ is not
  admissible.
\end{proposition}
\begin{proof}
  Let us assume, for sake of contradiction, that
  $  \mathcal{C}_{\set{H}}(x)$ is not
  admissible then also  
$  \mathcal{C}_{\set{K}}(x)$ but is
admissible.
From Corollary
\ref{Scor:badstringmultiplecontraction}
we have that
\begin{align}
  \label{eq:16}
  D_x \cap S^x_{\set{H}} \neq \emptyset, \quad
  D_x \cap S^x_{\set{K}} = \emptyset.
\end{align}
Let us consider the string
\begin{align}
  \label{eq:17}
  \begin{aligned}
     &\set{s} \in D_x \cap S^x_{\set{H}}, \\  
  &\set{s} =
  \set{u}_{\set{in}_{x} \setminus \set{in}_\set{K}}
  \set{v}_{\set{in}_{\set{K}} \setminus \set{in}_\set{H}}
  \set{b}_{\set{in}_{\set{H}}}
\set{e}_{\set{out}_x \setminus \set{out}_{\set{K}}}
  \set{e}_{\set{out}_{\set{K}} \setminus \set{out}_{\set{H}}}
  \set{b}'_{\set{in}_{\set{H}}} \\
 & \set{b}_{\set{in}_{\set{H}}} \in T_{\set{in}_{\set{H}}}, \;
  \set{b} '_{\set{out}_{\set{H}}} \in
  T_{\set{out}_{\set{H}}}, \;  b_{Ai} =
    b'_{A'j}.
  \end{aligned} 
\end{align}
Since $\set{s} \in D_x$, from Corollary
\ref{Scor:promoting0to1multiplesystems}
we have that $\set{s}' \in D_x$ where
\begin{align}
  \label{eq:18}
  \set{s}' :=
  \set{u}_{\set{in}_{x} \setminus \set{in}_\set{K}}
  \set{e}_{\set{in}_{\set{K}} \setminus \set{in}_\set{H}}
  \set{b}_{\set{in}_{\set{H}}}
\set{e}_{\set{out}_x \setminus \set{out}_{\set{K}}}
  \set{e}_{\set{out}_{\set{K}} \setminus \set{out}_{\set{H}}}
  \set{b}'_{\set{in}_{\set{H}}} 
\end{align}
However, we clearly have $\set{s}' \in
S^x_{\set{K}}$ which contradicts the hypothesis. 
\end{proof}
We conclude this section with the following
proposition, which shows that the admissibility
of a set of contraction is equivalent to a type
inclusion.
\begin{proposition}
  \label{Slmm:admisssupermap2}
Let us assume the same definitions as in Corollary
\ref{Scor:badstringmultiplecontraction}.
Then we have
\begin{align}
  \label{eq:19}
  \begin{aligned}
  &\mathcal{C}_{\set{H}}(x) \mbox{ is admissible }\\
  &\iff
  x \subseteq \left ( \bigotimes_{i=1}^n (A'_i \to
  A_i)\right) \to
(  \widetilde{\set{in}} \to  \widetilde{\set{out}} ).  
  \end{aligned}
  \end{align}
\end{proposition}
\begin{proof}
  In order to lighten the notation, let us define
  $z :=  \left ( \bigotimes_{i=1}^n (A'_i \to
  A_i)\right) $
If the inclusion
  $x \subseteq z \to
(  \widetilde{\set{in}} \to  \widetilde{\set{out}} )$
holds then
$R * \bigotimes_{i=1}^n \Phi_{A_iA'_i} \in
\set{T}_1(\widetilde{\set{in}} \to
\widetilde{\set{out}} )$ for any
$R \in \set{T}_1(x )$, i.e.
$\mathcal{C}_{\set{H}}(x) $ is admissible.

Let us now suppose that
$\mathcal{C}_{\set{H}}(x) $ is admissible, i.e.
$S_{\set{H}}^x \cap D_x = \emptyset $.
From Equation \eqref{Seq:setofstrings20}
one can prove that
\begin{align}
  D_{z
  \to(  \widetilde{\set{in}} \to  \widetilde{\set{out}} )}
  &= W_z D_{\widetilde{\set{in}} \to
 \widetilde{\set{out}} }
   \cup
     \overline{D_z}
     D^\perp_{\widetilde{\set{in}} \to
   \widetilde{\set{out}} }  = \nonumber \\
\label{eq:20}
  &=W_z W_{\widetilde{\set{in}}}
     T_{\widetilde{\set{out}} }
    \cup
     \overline{D_z}
     W_{\widetilde{\set{in}}}
     \set{e}_{\widetilde{\set{out}}} \\
  D^\perp_{z
  \to(  \widetilde{\set{in}} \to
  \widetilde{\set{out}} )}
  &= B\,
W_{\widetilde{\set{in}}}
\set{e}_{\widetilde{\set{out}}} 
 \\
  B &= B_1B_2\dots
        B_n \\
  B_{i} &= \{0_{A_i}0_{A'_i},0_{A_i}1_{A'_i}, 1_{A_i}1_{A'_i}\}
\end{align}
Let us suppose, for sake of contradiction, that 
  $x \not \subseteq z \to
(  \widetilde{\set{in}} \to  \widetilde{\set{out}}
)$.
Then there must exist a string $\set{s}$ such that
\begin{align}
  \label{eq:22}
  \set{s} &\in D_x \cap D^\perp_{z
  \to(  \widetilde{\set{in}} \to
  \widetilde{\set{out}} )}, \\
  \set{s} &= \set{b}_1 \set{b}_2\dots \set{b}_n
  \set{w} \,\set{e}_{\widetilde{\set{out}}}, \\
  \set{b}_i &\in B_{i}, \quad
  \set{w} \in W_{\tilde{\set{in}}}.
\end{align}
Since $\set{s} \in D_x$ and $A'_i \in \set{out}_x$
for any $i$, from
Corollary~\ref{Scor:promoting0to1multiplesystems}
we have that
\begin{align}
  \label{eq:21}
  \set{s}' &:=
  \set{b}'_1 \set{b}'_2\dots \set{b}'_n
\set{w}
  \,\set{e}_{\widetilde{\set{out}}} \in D_x,\\
  \set{b}'_i &:=
  \begin{cases}
    \set{b}_i & \mbox{ if }
    \set{b}_i  \in \{ 0_{A_i}0_{A'_i},
    1_{A_i}1_{A'_i} \}, \\
      0_{A_i}0_{A'i} & \mbox{ if }
      \set{b}_i  =0_{A_i}1_{A'_i}.
  \end{cases}
\end{align}
However, $\set{s'} \in S_{\set{H}}^x$
which contradict the hypothesis that
$\mathcal{C}_{\set{H}}(x) $ is admissible.
\end{proof}
\begin{corollary}
  \label{Slmm:admisssupermap}
  Given a type $x$, $A \in \set{in}_x$, and
  $ B\in \set{out}_x$, then $\mathcal{C}_{AB}(x)$
  is admissible if and only if
  $x\subseteq (B\rightarrow
  A)\rightarrow(\widetilde{\set{in}}\rightarrow
  \widetilde{\set{out}})$, where
  $\widetilde{\set{in}}=\set{in}_x\setminus A$ and
  $\widetilde{\set{out}}=\set{out}_x\setminus B$.
\end{corollary}

\section{Signalling and admissible composition}

In this section we will show how the admissibility
of contractions is related to the signalling
structure between the input systems and the output 
systems of a type.
Let us start by recalling the definition of a
no-signalling channel.
\begin{definition}
  Let $\mathcal{R}$ be a 
   bipartite quantum channel $\mathcal{E}:
  \mathcal{L}(\hilb{H}_A \otimes \hilb{H}_B) \to
  \mathcal{L}(\hilb{H}_C \otimes \hilb{H}_D)$
  and let $R \in \mathcal{L}(\hilb{H}_A \otimes
  \hilb{H}_B \otimes \hilb{H}_C \otimes \hilb{H}_D)$
  be its corresponding Choi operator (i.e. $R$ is
  a deterministic map of type $AB \to CD$).
  We say that $\mathcal{E}$
  is \emph{no-signalling} from $A$ to $C$
  if $\Tr_D[R] = I_A \otimes R'_{BC}$ for some
  $R'_{BC} \in \mathcal{L}(
  \hilb{B} \otimes \hilb{C} )$.
\end{definition}
Thanks to Proposition \ref{Spropositionchannel},
the previous definition can be strightforwardly
generalized to maps of arbitrary type
\begin{definition}
  Let $x$ be a type, $A \in \set{in}_x$, $B \in \set{out}_x$ 
  and let $R$ be a deterministic map of type $x$.
  We say that $R$ is \emph{no-signalling} from
  $A$ to $B$ and we write $A \not\rightsquigarrow_{R}B$ 
  if $R$ is no-signalling from $A$ to $B$
  when regarded as a channel from
  $\hilb{H}_{\set{in}_x}$, $\hilb{H}_{\set{out}_x}$,
  i.e. if and only if
$    \Tr_{\widetilde{\set{out}}}[R] = I_{A}\otimes R' $,
  where
  $R' \in
  \mathcal{L}(\hilb{H}_{\widetilde{\set{in}}}
  \otimes \hilb{H}_B )$,
  $\widetilde{\set{in}} := \set{in}_x \setminus B
  $, and
  $\widetilde{\set{out}} := \set{out}_x \setminus
  B $.

  We say that the type $x$ is
  \emph{no-signalling} from $A$ to $B$ and we
  write  $A \not\rightsquigarrow_{x}B$, if, for any
  $R \in \Evd{x}$, $R$ is no-signalling from $A$
  to $B$.
\end{definition}
We now can prove that the admissibility
of  a contraction is equivalent to a no-signalling
condition.
 \begin{proposition}\label{Sprop:causal}
   Let $x$ be a type, $A\in \set{in}_x$, and $B
   \in \set{out}_x$. Then we have that
   \begin{align}
     \label{Seq:74}
   \mathcal{C}_{AB}(x) \mbox{ is admissible} \iff  A
 \not\rightsquigarrow_{x} B  
   \end{align}
 \end{proposition}
 \begin{proof} 
   From Corollary \ref{Slmm:admisssupermap}
   we have that
   $\mathcal{C}_{AB}(x)$ is
   admissible if and only if 
$x \subseteq (A \to B) \to (\widetilde{\set{in}} \to
\widetilde{\set{out}} )$. Then,  the realisation
theorem of supermaps \cite{PhysRevLett.101.060401,chiribella2008transforming} implies that
$A
 \not\rightsquigarrow_{R} B \iff R \in \Evd{(A \to B) \to (\widetilde{\set{in}} \to
\widetilde{\set{out}} )} $. The thesis then
follows.
\end{proof}

\begin{corollary}
  Given a type $x$ with $A\in\set{in}_x$ and
  $B\in\set{out}_x$, then $A \not
  \rightsquigarrow_{x} B$
  if and only if $D_x\cap S^x_{AB}=\emptyset$.
\end{corollary}

This result proves that to check the validity of a
contraction corresponds to verify a non-signalling
condition.  Moreover, from Proposition
\ref{Sprop:subsetofcontractionarenotadmissibleyoucannotrecover}
we have that verifying the admissibility of a set
of contractions is given by chaining together the
verifications of each singular contraction in
sequence. In other words, a set of contractions is
admissible if and only if, by performing the
contractions in an (arbitrary) sequence, we map a
valid channel to a valid channel at each step of
the procedure.  Then, the admissibility of a set
of contractions involves verifying a sequence of
no-signaling conditions.

The final result we want to prove shows that the
signalling relation between a pair of systems can
be directly infered from the expression of the
type itself.
In order to prove this result we need the
following definition and
some
auxiliary lemmas.

\begin{definition}\label{SFS}
  Given a type $x$ with $A\in\set{in}_x$ and
  $B\in\set{out}_x$, then we say that $x$ is
  \emph{full-signalling} from $A$ to $B$ and we write
  $A \stackrqarrow{f}{} B$ if the type
$B$ to $A$,
  i.e $A \not
  \rightsquigarrow_{\overline{x}} B$.
\end{definition}
% \begin{lemma}
% Given a type $x$ with $A\in\set{in}_x$ and $B\in\set{out}_x$, then
% \begin{equation*}
%     x \:\textit{is}\: (NS)_{AB}\iff \overline{x}\:\textit{is}\: (FS)_{BA}
% \end{equation*}
% \end{lemma}
% \begin{proof}
% Let us take $x$ $(NS)_{AB}$, then $\overline{\overline{x}}$ is 
% $(NS)_{AB}$. For definition~\ref{SFS}, $\overline{x}$ results 
% to be $(FS)_{BA}$.
% %Given $x$ $(NS)_{AB}$, then $D_x\cap S^x_{AB}=\emptyset$. Now we  define 
% %$S^x_{AB}\cap S^x_{BA}=\textbf{e}_{\set{in}_x\setminus A}0_A\textbf{e}_{\set{out}_x\setminus B}0_B=:\overline{\textbf{e}}$ and noticing that
% %$D_x\cap \overline{\textbf{e}}=\emptyset$ and $\overline{\textbf{e}}\subseteq D_{\set{out}_x\rightarrow \set{in}_x}$, we obtain $\overline{\textbf{e}}\subseteq D_{\overline{x}}$, namely  $D_{\overline{x}}\cap S^x_{BA}\neq\emptyset$. 
% The opposite 
% implication steams directly from definition~\ref{SFS}.
% \end{proof}

\begin{lemma}
  \label{Slmm:strongnosignal1}
  Given a type $x$ with $A\in\set{in}_x$ and
  $B\in\set{out}_x$, then
\begin{align}
A \stackrqarrow{f}{}_x B \iff
\forall y,\: A \stackrqarrow{f}{}_{x \otimes y} B.
\end{align}
\end{lemma}
\begin{proof}
  First we notice that, since
  $x \equiv \overline{\overline{x}}$ (see Equation
  \eqref{Seq:53}) we have
  that $A \stackrqarrow{f}{}_x B$ implies
  $B \not\stackrqarrow{}{}_{\overline{x}} A$.  Let
  us now consider the type 
  $\overline{x\otimes y}=\overline{y\otimes
    x}=y\rightarrow \overline{x}$ (see Equations
  \eqref{Seq:52} and \eqref{Seq:55}) for an
  arbitrary type $y$. From
  Proposition~\ref{Slmm:charactstring} we have
  $D_{y\rightarrow \overline{x}}=W_y
  D_{\overline{x}} \cup
  D_{\overline{y}}D_{\overline{x}}^\perp$ and from
  Corollary \ref{Scoroll:criticstringsetAB} we have
  $S^{y\rightarrow
    \overline{x}}_{BA}=W_{\set{out}_y}W_{{\set{out}_x
    \setminus B}}\, 0_B\, \textbf{e}_{\set{in}_y}\textbf{e}_{\set{in}_x\setminus
    A}0_A= 
  W_{\set{out}_y}\textbf{e}_{\set{in}_y}S^{\overline{x}}_{BA}$.
  It is now straightforward to verify that
\begin{align*}
  \begin{aligned}
  S^{y \rightarrow \overline{x}}_{BA}\cap
  D_{y\rightarrow \overline{x}} = \emptyset
  \implies
  A \not \rightsquigarrow_{\overline{x \otimes y}}
  B
  \implies 
  A \stackrqarrow{f}{}_{x\otimes y}
  B  .
  \end{aligned}
\end{align*}
  which proves the implication
$  A \stackrqarrow{f}{}_x B$ $\implies$
$ A \stackrqarrow{f}{}_{x \otimes y} B$ for any $y$.
% which implies that
% holds if and only if $x\otimes y$ is
% $(FS)_{AB}$.
The inverse implication is trivial.
\end{proof}
\begin{lemma}
    \label{Slmm:strongnosignal2}
  Given a type $x$ with $A\in\set{in}_x$ and
  $B\in\set{out}_x$, then
\begin{align}
  A \not \rightsquigarrow_x B \iff
  A \not \rightsquigarrow_{x \otimes y} B  \;\;
  \forall y.  
\end{align}
\end{lemma}
\begin{proof}
  Let us assume that
  $ A \not \rightsquigarrow_x B $.  Then we have
  $S^x_{AB} \cap D_x=\emptyset $. Let us now
  consider
  $D_{x\otimes
    y}=D_x\textbf{e}_y\cup\textbf{e}D_y\cup D_x
  D_y$ and the set
  $S^{x\otimes y}_{AB}=W_{\set{in}_{x\otimes
      y}\setminus
    A}0_A\textbf{e}_{\set{out}_{x\otimes
      y}\setminus B}0_B=W_{\set{in}_x\setminus
    A}W_{\set{in}_y}0_A\textbf{e}_{\set{out}_y}\textbf{e}_{\set{out}_x\setminus
    B}0_B=S^x_{AB}W_{\set{in}_y}\textbf{e}_{\set{out}_y}$.
  We clearly have
  $D_{x\otimes y}\cap S^{x\otimes
    y}_{AB}=\emptyset $, i.e. $A \not
  \rightsquigarrow_{x\otimes y} B$.  The inverse implication
  is trivial.
\end{proof}

\begin{lemma}
    Given a type $x$ with $A\in\set{in}_x$ and
    $B\in\set{out}_x$, then we have
    \begin{align}
      \label{Seq:76}
      A \stackrqarrow{f}{}_x B
      \implies
      A \stackrqarrow{f}{}_{y \to x} B
      \mbox{ and }
      B \not \rightsquigarrow_{x \to y} A \\
A  \not \rightsquigarrow_x B
      \implies
      A \not \rightsquigarrow_{y \to x} B
      \mbox{ and }
      B \stackrqarrow{f}{}_{x \to y} A
    \end{align}
\end{lemma}
\begin{proof}
  This result follows from
   Lemmas \ref{Slmm:strongnosignal1} and
   \ref{Slmm:strongnosignal2} and the type equivalence
  $x\to y \equiv \overline{ x \otimes
    \overline{y}}$.
\end{proof}

\begin{lemma}\label{SFSinout}
  Let us define the type $z=x\rightarrow y$.
  Then we have
  \begin{align}
    \label{Seq:71}
    \begin{aligned}
      A\in \set{in}_y, B\in \set{in}_x \implies
      A \stackrqarrow{f}{}_z B   \\
      A\in \set{out}_y, B\in \set{out}_x \implies
      B \stackrqarrow{f}{}_z A 
  \end{aligned}
  \end{align}
% or
%   , then we have
% \begin{align}
% A\in \set{in}_y, B\in \set{in}_x\lor A\in \set{out}_x, B\in \set{out}_y  \implies z\;  (FS)_{AB}.
% \end{align}
%\begin{enumerate}
  % \item[(i).] if $A\in \set{in}_y$ and $B\in \set{in}_x$ then $z$ is $(FS)_{AB}$ 
   % \item[(ii).] if $A\in \set{out}_x$ and $B\in \set{out}_y$ then $z$ is $(FS)_{AB}$
%\end{enumerate}
\end{lemma}
\begin{proof}
  Let us prove the first implication.  % We know
  % that $z$ is $(FS)_{AB}\iff$ $\overline{z}$ is
  % $(NS)_{BA}$.
  % Then we
  Let us consider
  $\overline{z}=\overline{x\rightarrow
    y}=x\otimes\overline{y}$. we have
  $D_{x\otimes\overline{y}}=D_x\textbf{e}_y\cup\textbf{e}_xD_{\overline{y}}\cup
  D_x D_{\overline{y}}$ and 
  $S^{x\otimes\overline{y}}_{BA}=W_{\set{in}_{x\otimes
      \overline{y}}\setminus
    B}0_B\textbf{e}_{\set{out}_{x\otimes
      \overline{y}}\setminus
    A}0_A=W_{\set{in}_x\setminus
    B}W_{\set{out}_y}0_B\textbf{e}_{\set{out}_x}\textbf{e}_{\set{in}_y\setminus
    A}0_A$.
  Since we have that 
  $D_x\subseteq W_{\set{in}_x}T_{\set{out}_x}$ and
  $D_{\overline{y}}\subseteq
  W_{\set{out}_y}T_{\set{in}_y}$, we can
  verify that
  $D_{x\otimes\overline{y}}\cap
  S^{x\otimes\overline{y}}_{BA}=\emptyset$. This
  means that $B \not
  \rightsquigarrow_{\overline{z}} A$ which implies
  $A \stackrqarrow{f}{}_z B$. 
The proof of the second implication
is analogous.
\end{proof}

We are now ready to prove the main result of this section.
\begin{proposition} \label{Sprp:algosignallingsupp}
   Let $x$ be a type, $A \in \set{in}_x$,
   $B \in \set{out}_x$. Then, there exists a
   unique type $y \prec x$ such that
   $\{A, B \} \subseteq \set{Ele}_y$ and $y' \prec y
   \implies \{A, B \} \not \subseteq
   \set{Ele}_y$. Moreover, we have
   \begin{align}
     \label{Seq:8ert}
     A \in \set{in}_y  &\implies   A \stackrqarrow{f}{}_{ x} B,
\\
     A \in \set{out}_y &\implies A
     \not\rightsquigarrow_{x}  B .
   \end{align}
 \end{proposition}
 \begin{proof}
   The existence an uniqueness of such a type
   $y = y_1 \to y_2$ was proved in Lemma
   \ref{Sinnertype}.  Let us now assume $ A \in
   \set{in}_y $.
   Then
   it must be $A \in \set{out}_{y_1}$ and
   $B \in \set{out}_{y_2}$ which from Lemma
   \ref{SFSinout} implies
   $A\stackrqarrow{f}{}_{ y } B  $.
   Let us now introduce the following notation:
  \begin{align}
    \label{Seq:73}
        \mathcal{C}_0^z(x):=z\rightarrow x,\quad \mathcal{C}_{1}^z(x):=x\rightarrow z.
  \end{align}
  Since $y\preceq x$, we have
  $x=\mathcal{C}^{z_n}_{s_n}(\mathcal{C}^{z_{n-1}}_{s_{n-1}}
  (\cdots\mathcal{C}^{z_1}_{s_1}(y)\cdots))$
  for some types $z_i$.
  Since $A \in \set{in}_x$, it implies that
$ \sum_{i=1}^ns_i (\text{mod 2} )=0$ which in turn
gives  $A \stackrqarrow{f}{}_{ x} B$.

The proof of the implication $ A \in \set{out}_y \implies A
     \not\rightsquigarrow_{x}  B $ is analogous.
  % This implies that
  % \begin{align}
  %   \label{Seq:75}
  %   \sum_{i=1}^ns_i (\text{mod 2} )=1 \iff \\
  %   \sum_{i=1}^ns_i (\text{mod 2} )=0
  % \end{align}
\end{proof}
\end{document}